\documentclass{article}

\usepackage[vmargin=1.18in,hmargin=1.1in]{geometry}
\usepackage{amsmath,amssymb,amsfonts,graphicx,amsthm,nicefrac,mathtools,bm}
\usepackage{hyperref}
\usepackage[numbers]{natbib}
\usepackage[english]{babel}
\usepackage{times}
\usepackage[T1]{fontenc}
\usepackage{bbm,mdframed}
\usepackage[dvipsnames]{xcolor}
\usepackage{xspace}
\usepackage{rotating,multirow}
\usepackage{enumerate,graphics,enumitem}
\usepackage{subcaption}

\setlist[itemize]{noitemsep, topsep=0pt}
\setlist[enumerate]{itemsep=5pt, topsep=5pt, leftmargin=25pt}

\newtheorem{theorem}{Theorem}

\definecolor{verylightblue}{rgb}{0.7,0.8,1}
  {\begin{mdframed}[backgroundcolor=verylightblue]\begin{theorem}}%
  {\end{theorem}\end{mdframed}}

\definecolor{verylightgray}{gray}{0.95}
  {\begin{mdframed}[backgroundcolor=verylightgray]\begin{proof}}%
  {\end{proof}\end{mdframed}}

\newtheorem{lemma}{Lemma}
\definecolor{verylightred}{rgb}{1,0.8,0.8}
  {\begin{mdframed}[backgroundcolor=verylightred]\begin{lemma}}%
  {\end{lemma}\end{mdframed}}

\newtheorem{proposition}{Proposition}
  {\begin{mdframed}[backgroundcolor=verylightblue]\begin{proposition}}%
  {\end{proposition}\end{mdframed}}

\theoremstyle{definition}

\theoremstyle{remark}

\makeatletter

\newtheorem*{rep@theorem}{\rep@title}
\newcommand{\newreptheorem}[2]
{\newenvironment{rep#1}[1]
{\def\rep@title{#2 \ref{##1}} \begin{rep@theorem}}%
 {\end{rep@theorem}}}
\makeatother
\newreptheorem{theorem}{Theorem}
\newreptheorem{lemma}{Lemma}
\newreptheorem{corollary}{Corollary}
\newreptheorem{proposition}{Proposition}

\newcommand{\figref}[1]{Figure~\ref{fig:#1}}

\newcommand{\secref}[1]{Section~\ref{sec:#1}}

\newcommand{\lemref}[1]{Lemma~\ref{lem:#1}}

\newcommand{\propref}[1]{Proposition~\ref{prop:#1}}

\newcommand{\thmref}[1]{Theorem~\ref{thm:#1}}

\newcommand{\eqnref}[1]{\eqref{eqn:#1}}

\newcommand{\PP}[1]{\textnormal{Pr}\!\left\{{#1}\right\}} 
\newcommand{\EE}[1]{\mathbb{E}\left[{#1}\right]} 

\newcommand{\EEst}[2]{\mathbb{E}\left[{#1}\ \middle| \ {#2}\right]} 
\newcommand{\PPst}[2]{\text{Pr}\!\left\{{#1}\ \middle| \ {#2}\right\}} 

\def\R{\mathbb{R}}

\newcommand{\ignore}[1]{}

\usepackage{tikz}
\usetikzlibrary{arrows}
\usetikzlibrary{shapes}

\newcommand{\thedate}{\today}
\newcommand{\theauthor}{}
\newcommand{\thetitle}{SAFFRON: an adaptive algorithm for online control\\ of the false discovery rate}

\date{\thedate}
\author{\theauthor}
\title{\thetitle}

\usepackage[mmddyy]{datetime}

\newcommand{\nulls}{\mathcal{H}^0}

\newcommand{\fdp}{\textnormal{FDP}}
\newcommand{\fdr}{\textnormal{FDR}}

\newcommand{\mfdr}{\textnormal{mFDR}}

\newcommand{\fdphat}{\widehat{\fdp}}

\newcommand{\One}[1]{{\bf{1}}\left\{{#1}\right\}}

\def\N{\mathbb N}

\def\F{\mathcal{F}}

\def\cR{\mathcal{R}}
\def\LORD{\mathrm{LORD}}
\def\SAFFRON{\mathrm{SAFFRON}}
\def\BH{\mathrm{BH}}
\def\StBH{\mathrm{StBH}}
\def\pihat{\widehat{\pi_0}}

\makeatletter
\newcommand{\dotfrac}[2]{
\mathchoice
{\ooalign{$\genfrac{}{}{0pt}{0}{#1}{#2}$\cr\leavevmode\cleaders\hb@xt@ .22em{\hss $\displaystyle\cdot$\hss}\hfill\kern\z@\cr}}
{\ooalign{$\genfrac{}{}{0pt}{1}{#1}{#2}$\cr\leavevmode\cleaders\hb@xt@ .22em{\hss $\textstyle\cdot$\hss}\hfill\kern\z@\cr}}
{\ooalign{$\genfrac{}{}{0pt}{2}{#1}{#2}$\cr\leavevmode\cleaders\hb@xt@ .22em{\hss $\scriptstyle\cdot$\hss}\hfill\kern\z@\cr}}
{\ooalign{$\genfrac{}{}{0pt}{3}{#1}{#2}$\cr\leavevmode\cleaders\hb@xt@ .22em{\hss $\scriptscriptstyle\cdot$\hss}\hfill\kern\z@\cr}}
}
\makeatother

\usepackage{algorithm,algorithmic}

\newcommand{\defn}{\ensuremath{:\, =}}

\makeatletter
\long\def\@makecaption#1#2{
        \vskip 0.8ex
        \setbox\@tempboxa\hbox{\small {\bf #1:} #2}
        \parindent 1.5em  
        \dimen0=\hsize
        \advance\dimen0 by -3em
        \ifdim \wd\@tempboxa >\dimen0
                \hbox to \hsize{
                        \parindent 0em
                        \hfil 
                        \parbox{\dimen0}{\def\baselinestretch{0.96}\small
                                {\bf #1.} #2
                                } 
                        \hfil}
        \else \hbox to \hsize{\hfil \box\@tempboxa \hfil}
        \fi
        }
\makeatother

\begin{document}

\author{ Aaditya Ramdas, ~ ~ Tijana Zrnic, ~ ~ Martin J. Wainwright, ~
  ~ Michael I. Jordan\\ Departments of Statistics and EECS, University
  of California, Berkeley\\ {\small \texttt{$\{$aramdas, tijana,
      wainwrig, jordan$\}$ @eecs.berkeley.edu}} } \maketitle

\begin{abstract}   
In the online false discovery rate (FDR) problem, one observes a
possibly infinite sequence of $p$-values $P_1,P_2,\dots$, each testing
a different null hypothesis, and an algorithm must pick a sequence of
rejection thresholds $\alpha_1,\alpha_2,\dots$ in an online fashion,
effectively rejecting the $k$-th null hypothesis whenever $P_k \leq
\alpha_k$. Importantly, $\alpha_k$ must be a function of the past, and
cannot depend on $P_k$ or any of the later unseen $p$-values, and must
be chosen to guarantee that for any time $t$, the FDR up to time $t$
is less than some pre-determined quantity $\alpha \in (0,1)$. In this
work, we present a powerful new framework for online FDR control that
we refer to as ``SAFFRON''. Like older alpha-investing (AI) algorithms,
SAFFRON starts off with an error budget, called alpha-wealth, that it
intelligently allocates to different tests over time, earning back
some wealth on making a new discovery. However, unlike
older methods, SAFFRON's threshold sequence is based on a novel
estimate of the alpha fraction that it allocates to true null
hypotheses. In the offline setting, algorithms that employ an estimate
of the proportion of true nulls are called adaptive methods, and
SAFFRON can be seen as an online analogue of the famous offline
Storey-BH adaptive procedure. Just as Storey-BH is typically more
powerful than the Benjamini-Hochberg (BH) procedure under
independence, we demonstrate that SAFFRON is also more powerful than
its non-adaptive counterparts, such as LORD and other generalized
alpha-investing algorithms. Further, a monotone version of the
 original AI algorithm is recovered as a special case of SAFFRON,
 that is often more stable and powerful than the original. Lastly, the derivation
 of SAFFRON provides a novel template for deriving new online FDR rules.
\end{abstract}



\section{Introduction}

It is now commonplace in science and technology to make thousands or
even millions of related decisions based on data analysis. As a
simplified example, to discover which genes may be related to
diabetes, we can formulate the decision-making problem in terms of
hypotheses that take the form ``gene X is not associated with
diabetes,'' for many different genes X, and test for which of these
null hypotheses can be confidently rejected by the data.  As first
identified by Tukey in a seminal 1953 manuscript
\cite{tukey1953problem}, the central difficulty when testing a large
number of null hypotheses is that several of them may appear to be
false, purely by chance. Arguably, we would like the set of rejected
null hypotheses $\cR$ to have high \emph{precision}, so that most
discovered genes are indeed truly correlated with diabetes and further
investigations are not fruitless. Unfortunately, separately
controlling the false positive rate for each individual test actually
does not provide any guarantee on the precision. This motivated the
development of procedures that can provide guarantees on an error
metric called the false discovery rate (FDR) \cite{BH95}, defined as:
\begin{align*}
\fdr \equiv \EE{\fdp(\cR)} = \EE{\frac{|\nulls \cap \cR|}{|\cR|}},
\end{align*}
where $\nulls$ is the unknown set of truly null hypotheses, and
$0/0\equiv0$. Here the FDP represents the ratio of falsely rejected
nulls to the total number of rejected nulls, and since the set of
discoveries $\cR$ is data-dependent, the FDR takes an expectation over
the underlying randomness. The evidence from a hypothesis test can
typically be summarized in terms of a $p$-value, and so offline
multiple testing algorithms take a set of $p$-values $\{P_i\}$ as
their input, and a target FDR level $\alpha \in (0,1)$, and produce a
rejected set $\cR$ that is guaranteed to have $\fdr \leq \alpha$.  Of
course, one also desires a high recall, or equivalently a low false
negative rate, but without assumptions on many uncontrollable factors
like the frequency and strength of signals, additional guarantees on
the recall are impossible.

While the offline paradigm previously described is the classical
setting for multiple decision-making, the corresponding online problem
is emerging as a major area of its own.  For example, large
information technology companies run thousands of A/B tests every week
of the year, and decisions about whether or not to reject the
corresponding null hypothesis must be made without knowing the
outcomes of future tests; indeed, future null hypotheses may depend on
the outcome of the current test. The current standard of setting all
thresholds $\alpha_k$ to a fixed quantity such as $0.05$ does not
provide any control of the FDR. Hence, the following hypothetical
scenario is entirely plausible: a company conducts $1000$ tests in one
week, each with a target false positive rate of $0.05$; it happens to
make $80$ discoveries in total of which $50$ are accidental false
discoveries, ending up with an FDP of $5/8$.  Such uncontrolled error
rates can have severe financial and social consequences.

The first method for online control of the FDR was the alpha-investing algorithm of \citet{foster2008alpha}, later
extended to generalized alpha-investing (GAI) algorithms by
\citet{aharoni2014generalized}. Recently, \citet{javanmard2016online}
proposed variants of GAI algorithms that control the FDR (as opposed
to the modified FDR controlled in the original
paper~\cite{foster2008alpha}), including a new algorithm called
LORD. The GAI++ algorithms by \citet{RYWJ17} improved the earlier GAI
algorithms (uniformly), and the improved LORD++ (henceforth LORD)
method arguably represents the current state-of-the-art in online
multiple hypothesis testing.

The current paper's central contribution is the derivation and
analysis of a powerful new class of online FDR algorithms called
``SAFFRON'' (\underline{S}erial estimate of the \underline{A}lpha
\underline{F}raction that is \underline{F}utilely \underline{R}ationed
\underline{O}n true \underline{N}ull hypotheses). As an instance of
the GAI framework, the SAFFRON method starts off with an error budget,
referred to as \emph{alpha-wealth}, that it allocates to different
tests over time, earning back some alpha-wealth whenever it makes a
new discovery. However, unlike earlier work in the online setting,
SAFFRON is an adaptive method, meaning that it is based on an estimate
of the proportion of true nulls.  In the offline setting, adaptive
methods were proposed by Storey~\cite{Storey02,Storey04}, who showed
that they are more powerful than the Benjamini-Hochberg (BH)
procedure~\cite{BH95} under independence assumptions; this advantage
usually increases with the proportion of non-nulls and the signal
strength. Thus, the SAFFRON method can be viewed as an online analogue
of Storey's adaptive version of the BH procedure. As shown in
\figref{mean3}, our simulations show that SAFFRON demonstrates the
same types of advantages over its non-adaptive counterparts, such as
LORD and alpha-investing. Furthermore, the ideas behind SAFFRON's
derivation can provide a natural template for the design and analysis
of a suite of other adaptive online methods.

\begin{figure}[H]
\centerline{\includegraphics[width=0.35\textwidth]{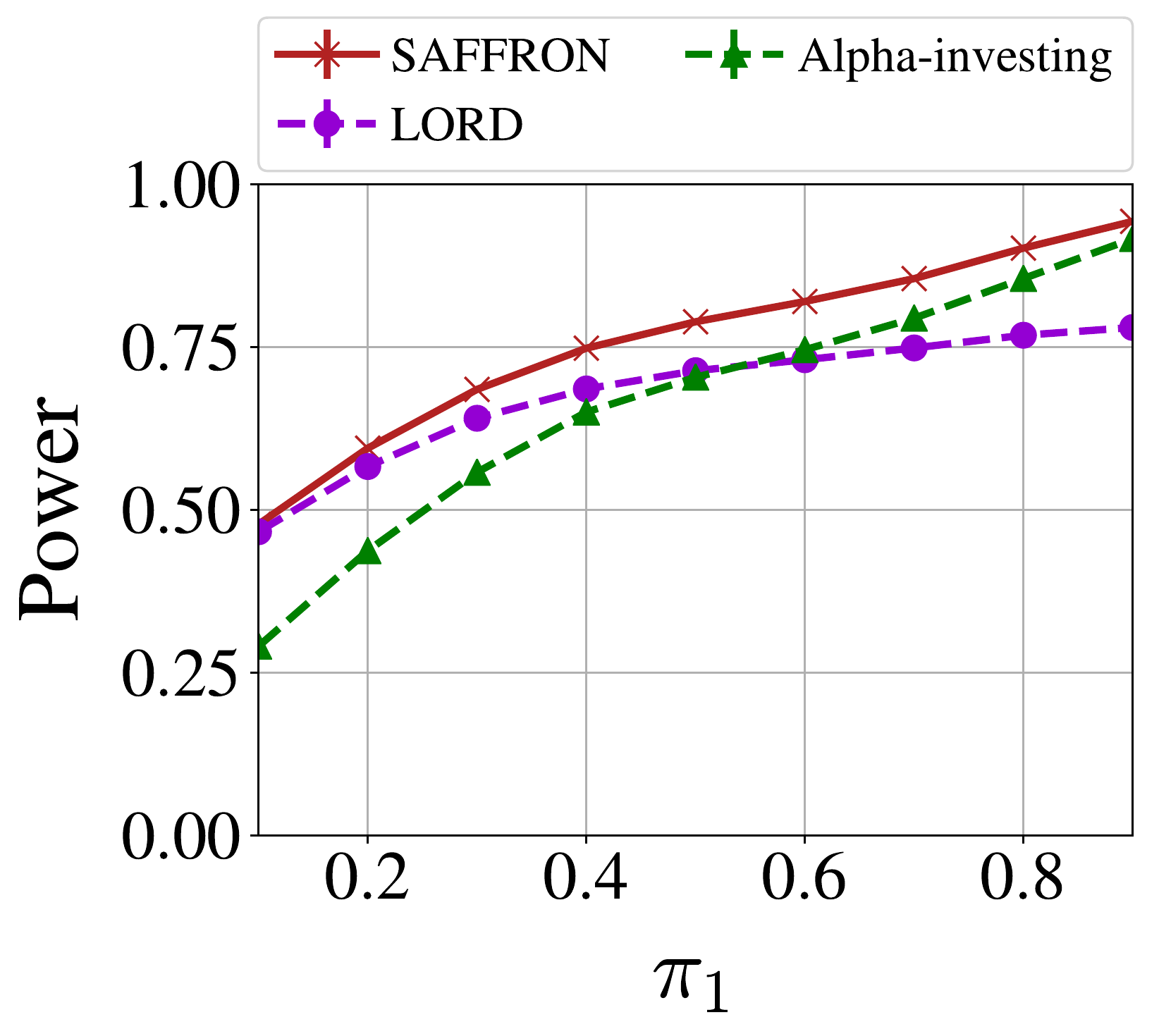}
\includegraphics[width=0.35\textwidth]{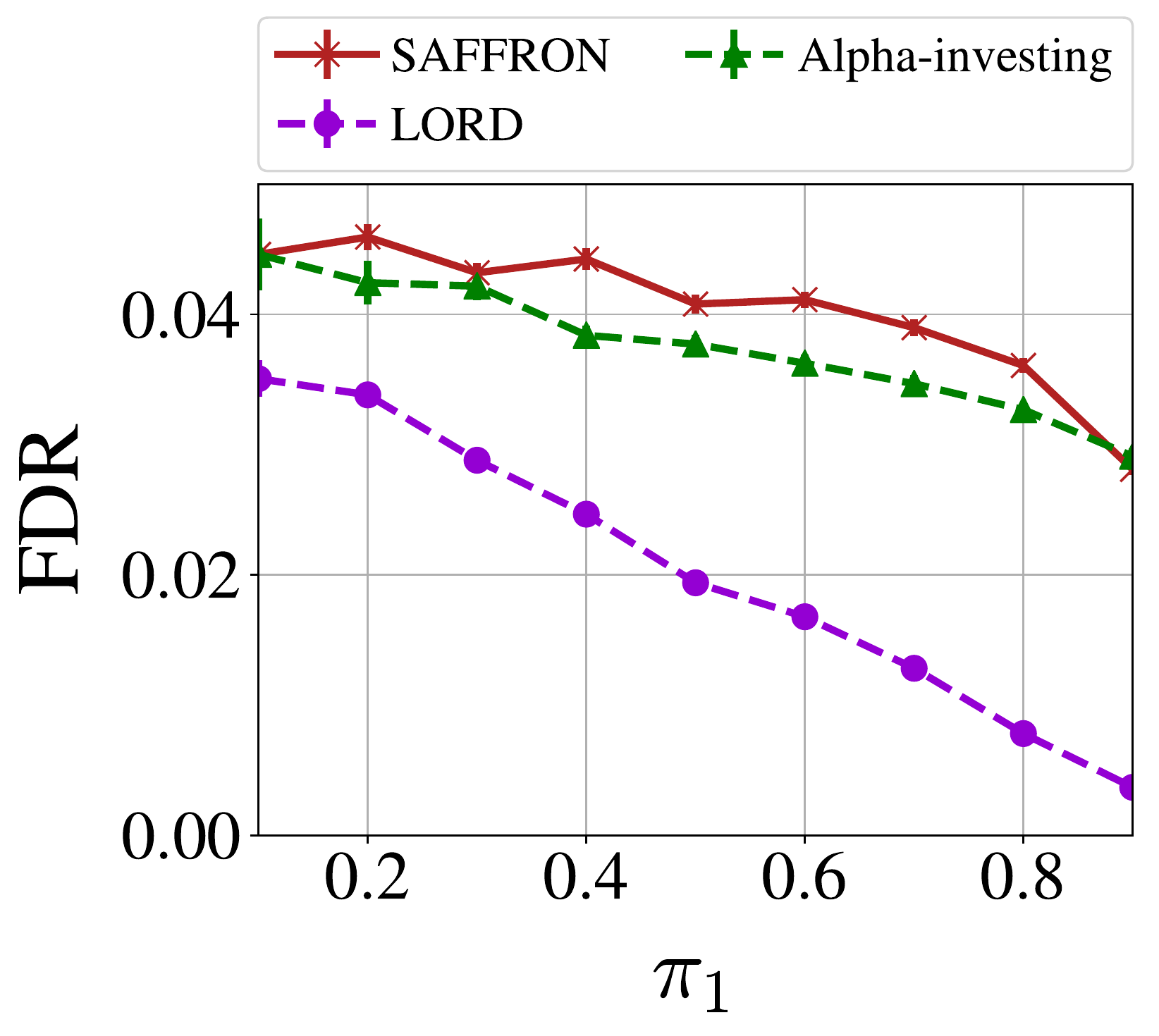}}
\caption{Statistical power and FDR versus fraction of non-null hypotheses $\pi_1$ for SAFFRON, LORD and alpha-investing at target FDR level $\alpha = 0.05$.
 The $p$-values are drawn as $P_i = \Phi(-Z_i)$, where $\Phi$ is the standard Gaussian CDF, and $Z_i \sim N(\mu_i,1)$, where nulls have $\mu_i=0$ and non-nulls have $\mu_i\sim N(3,1)$.} 
\label{fig:mean3}
\end{figure}

The rest of this paper is organized as follows. In \secref{derivation}, we derive the SAFFRON algorithm from first principles, leaving the proof of a central technical lemma for \secref{lemma}. In \secref{sims}, we investigate the practical choice of tuning parameters, and  demonstrate the effectiveness of our recommended choice using simulations. We provide proofs of the results of this paper in \secref{proofs}, and at the end present a short summary in \secref{disc}.


\section{Deriving the SAFFRON algorithm}
\label{sec:derivation}

Before deriving the SAFFRON algorithm, it is useful to recap a few
concepts. By definition of a $p$-value, if the hypothesis $H_i$ is truly
null, then the corresponding $p$-value is stochastically larger than the
uniform distribution (``super-uniformly distributed,'' for short),
meaning that:
\begin{equation}
\label{eqn:pvalue}
\text{If the null hypothesis $H_i$ is true, then } \PP{P_i \leq u}
\leq u \text{ for all } u\in[0,1].
\end{equation}
For any online FDR procedure, let the rejected set after $t$ steps be
denoted by $\cR(t)$. More precisely, this set consists of all $p$-values
among the first $t$ ones for which the indicator for rejection is
equal to 1; i.e., $R_j \defn \One{P_j \leq \alpha_j} = 1$, for all
$j\leq t$.  While we have already defined the classical FDP and FDR in
the introduction, several authors, including \citet{foster2008alpha},
have considered a modified FDR, defined as:
\begin{align}
\mfdr(t) & \defn \frac{\EE{|\nulls \cap \cR(t)|}}{\EE{|\cR(t)|}}.
\end{align}
In the sequel, we provide guarantees for both $\mfdr$ and $\fdr$. Our
guarantees on $\mfdr$ hold under the following weakening of
\eqref{eqn:pvalue}.  Define the filtration formed by the sequence of
sigma-fields $\F^t \defn \sigma(R_1, \ldots, R_t)$, and let $\alpha_t
\defn f_t(R_1, \ldots, R_{t-1})$, where $f_t$ is an arbitrary function of the
first $t-1$ indicators for rejection. Then, we say that the null
$p$-values are \emph{conditionally super-uniformly distributed} if the
following holds:
\begin{align}
\label{eqn:superuniformity-cond1}
\text{If the null hypothesis $H_i$ is true, then } \PPst{P_t \leq
  \alpha_t}{\F^{t-1}} \leq \alpha_t .
\end{align}


\subsection{An oracle estimate of the FDP  and a naive overestimate}

To understand the motivation behind the new procedure, it is necessary
to expand on an perspective on existing online FDR procedures,
recently suggested by \citet{RYWJ17}.  We begin by defining an
\emph{oracle estimate} of the FDP as:
\begin{align*}
\fdp^*(t) \defn \frac{\sum\limits_{j \leq t, j \in \nulls}
  \alpha_j}{|\cR(t)|}.
\end{align*}
The word oracle indicates that $\fdp^*$ cannot be calculated by the
scientist, since $\nulls$ is unknown. Intuitively, the numerator
$\sum_{j \leq t, j \in \nulls} \alpha_j$ overestimates the number of
false discoveries, and $\fdp^*(t)$ overestimates the FDP, as
formalized in the claim below:
\begin{proposition}
  \label{prop:fdp-oracle}
If the null $p$-values are conditionally super-uniformly
distributed~\eqnref{superuniformity-cond1}, then we have
\begin{enumerate}
\item[(a)] $\EE{\sum\limits_{j \leq t, j \in \nulls} \alpha_j} \geq
  \EE{|\nulls \cap \cR(t)|}$;
\item[(b)] If $\fdp^*(t) \leq \alpha$ for all $t \in \N$, then
  $\mfdr(t) \leq \alpha$ for all $t \in \N$.
\end{enumerate}
Further, if the null $p$-values are independent of each other and of the
non-nulls, and $\{\alpha_t\}$ is a monotone function of past rejections,
then: 
\begin{enumerate}
\item[(c)] $\EE{\fdp^*(t)} \geq \EE{\fdp(t)} \equiv \fdr(t)$ for all
  $t \in \N$;
\item[(d)] The condition $\fdp^*(t) \leq \alpha$ for all $t \in \N$
  implies that $\fdr(t) \leq \alpha$ for all $t \in \N$.
\end{enumerate}
\end{proposition}
To clarify, the word \emph{monotone} means that $\alpha_t$ is a
coordinatewise non-decreasing function of the vector $R_1,\dots,R_{t-1}$.
\propref{fdp-oracle} follows from the results of \citet{RYWJ17}, and
we prove it in \secref{proofs} for completeness.  Even though
$\fdp^*(t)$ cannot be directly calculated and used,
\propref{fdp-oracle} is a useful way to identify what would be ideally
possible. One natural way to convert $\fdp^*(t)$ to a truly empirical
overestimate of $\fdp(t)$ is to define:
\begin{align*}
\fdphat_{\LORD}(t) \defn \frac{\sum_{j \leq t} \alpha_j}{|\cR(t)|}.
\end{align*}
Since it is trivially true that $\fdphat_{\LORD}(t) \geq \fdp^*(t)$,
we immediately obtain that Proposition~\ref{prop:fdp-oracle} also
holds with $\fdp^*(t)$ replaced by $\fdphat_{\LORD}(t)$. The subscript
LORD is used because \citet{RYWJ17} point out that their variant of
the LORD algorithm of \citet{javanmard2016online} can be derived by
simply assigning $\alpha_j$ in an online fashion to ensure that the
condition $\fdphat_{\LORD}(t) \leq \alpha$ is met for all times $t$.


\subsection{A better estimate of the alpha-wealth spent on testing nulls}

The main drawback of $\fdphat_{\LORD}$ is that if the underlying
(unknown) truth is such that the proportion of non-nulls (true
signals) is non-negligible, then $\fdphat_{\LORD}(t)$ is a very crude
and overly conservative overestimate of $\fdp^*(t)$, and hence also of
the true unknown FDP.  With this drawback in mind, and knowing that we
would expect non-nulls to typically have smaller $p$-values, we propose
the following novel estimator:
\begin{align*}
\fdphat_{\SAFFRON(\lambda)}(t) \equiv \fdphat_{\lambda}(t) \defn
\frac{\sum_{j \leq t} \alpha_j \frac{\One{P_j >
      \lambda_j}}{1-\lambda_j}}{|\cR(t)|},
\end{align*}
where $\{\lambda_j\}_{j=1}^\infty$ is a predictable sequence of
user-chosen parameters in the interval $(0,1)$. Here the term
\emph{predictable} means that $\lambda_j$ is a deterministic function
of the information available from time 1 to $j-1$, which will be formalized later, together with additional requirements. 
For simplicity, when $\lambda_j$ is chosen to
be a constant for all $j$, we will drop the subscript and just write
$\lambda$, and we will consider $\lambda=1/2$ as our default
choice. SAFFRON is based on the idea that the numerator of
$\fdphat_\lambda$ is a much better estimator of the quantity $\sum_{j
  \leq t, j \in \nulls} \alpha_j$ than LORD's naive estimate $\sum_{j
  \leq t} \alpha_j$.

So as to provide some intuition for why we expect $\fdphat_\lambda$ to
be a fairly tight estimate of $\fdp^*$, note that $\frac{\One{P_j >
    \lambda_j}}{1-\lambda_j}$ has a unit expectation whenever $P_j$ is
uniformly distributed (null), but would typically have a much smaller
expectation whenever $P_j$ is stochastically much smaller than uniform
(non-null). The following theorem shows that, even though
$\fdphat_{\lambda}(t)$ is not necessarily always larger than
$\fdp^*(t)$, a direct analog of \propref{fdp-oracle} is nonetheless
valid.  In order to state this claim formally, we need to slightly
modify the assumption \eqnref{superuniformity-cond1}. As before,
denote by $R_j \defn \One{P_j \leq \alpha_j}$ the indicator for
rejection, and let $C_j := \One{P_j \leq \lambda_j}$ be the indicator
for candidacy. Accordingly, we refer to the $p$-values for which $C_j
= 1$ as candidates.  Moreover, we let $\alpha_t \defn
f_t(R_1,\dots,R_{t-1},C_1,\dots,C_{t-1})$, where $f_t$ denotes an arbitrary function of the first $t-1$ indicators for rejection and candidacy,
and define the filtration generated from sigma-fields $\F^t \defn
\sigma(R_1,\dots,R_t,C_1,\dots,C_t)$. With respect to this filtration,
we introduce a conditional super-uniformity condition on the null
$p$-values similar to \eqref{eqn:superuniformity-cond1}:
\begin{equation}
  \label{eqn:superuniformity-cond2}
\text{If the null hypothesis $H_i$ is true, then } \PPst{P_t \leq
  \alpha_t}{\F^{t-1}} \leq \alpha_t ,
 \end{equation}
which can be rephrased as:
\begin{align}
\label{eqn:superuniformity-cond3}
\EEst{\dotfrac{\One{P_t > \alpha_t}}{1-\alpha_t}}{\F^{t-1}} \geq 1 \geq
\EEst{\dotfrac{\One{P_t \leq \alpha_t}}{\alpha_t}}{\F^{t-1}} .
\end{align}
Note that again marginal super-uniformity \eqnref{pvalue} implies this
condition, provided that the $p$-values are independent.
\begin{theorem}\label{thm:fdp-saff}
If the null $p$-values are conditionally super-uniformly
distributed~\eqnref{superuniformity-cond2}, then we have:
\begin{enumerate}
\item[(a)] $\EE{\sum\limits_{j \leq t, j \in \nulls} \alpha_j
  \frac{\One{P_j > \lambda_j}}{1-\lambda_j} } \geq \EE{|\nulls \cap
  \cR(t)|}$;
\item[(b)] The condition $\fdphat_\lambda(t) \leq \alpha$ for all $t
  \in \N$ implies that $\mfdr(t) \leq \alpha$ for all $t \in \N$.
\end{enumerate}
Further, if the null $p$-values are independent of each other and of
the non-nulls, $\alpha_t$ and $\lambda_t$ are monotone functions of the vector $R_1,...,R_{t-1},C_1,...,C_{t-1}$, then we additionally have:
\begin{enumerate}
\item[(c)] $\EE{\fdphat_\lambda(t)} \geq \EE{\fdp(t)} \equiv \fdr(t)$
  for all $t \in \N$;
\item[(d)] The condition $\fdphat_\lambda(t) \leq \alpha$ for all $t
  \in \N$ implies that $\fdr(t) \leq \alpha$ for all $t \in \N$.
\end{enumerate}
\end{theorem}

The proof of this theorem is given in \secref{proofs}, and is based
on a ``reverse super-uniformity lemma'' that is discussed in the next
section.  This lemma, though of a technical nature, may be of
independent interest in deriving new algorithms. The statements on
mFDR control allow SAFFRON to be used in place of LORD in applications
in which $p$-values are not independent, but are conditionally
super-uniformly distributed, such as the MAB-FDR framework (based on
multi-armed bandits) proposed by~\citet{yang2017multi}.


\subsection{The SAFFRON algorithm for constant $\lambda$}
\label{sec:const-lbd}

We now present the SAFFRON algorithm at a high level.  For simplicity,
we consider the constant $\lambda$ setting, which performs well in experiments, though it may be a useful direction for future work to construct good heuristics for time-varying sequences $\{\lambda_j\}_{j=1}^\infty$.
\begin{enumerate}
  \item Given a target FDR level $\alpha$, the user first picks a
    constant $\lambda \in (0,1)$, an initial wealth $W_0 <
    \alpha$, and a positive non-increasing sequence
    $\{\gamma_j\}_{j=1}^\infty$ of summing to one. For example, given
    a parameter $s > 1$, we might pick $\gamma_j \propto j^{-s}$ for
    some $s > 1$.
  \item We use the term ``candidates'' to refer to $p$-values smaller
    than $\lambda$, since SAFFRON will never reject a $p$-value larger
    than $\lambda$. Recalling the indicator for candidacy $C_t \defn
    \One{P_t \leq \lambda}$, and denoting by $\tau_j$ be the time of
    the $j$-th rejection (and setting $\tau_0=0$), define the
    candidates after the $j$-th rejection as $C_{j+} = C_{j+}(t) =
    \sum_{i = \tau_j + 1}^{t-1} C_i$.
  \item SAFFRON begins by allocation $\alpha_1 = \min\{(1-\lambda)\gamma_1
    W_0,\lambda\}$, and then at time $t = 2, 3, \ldots$, it allocates:
\begin{align*}
\alpha_t \defn \min\{ \lambda, \widetilde \alpha_t\}, \text{~ where ~}
\widetilde \alpha_t \defn (1-\lambda)\Big(W_0\gamma_{t - C_{0+}} + (\alpha -
W_0) \gamma_{t-\tau_1 - C_{1+}} + \sum_{j \geq 2} \alpha
\gamma_{t-\tau_j - C_{j+}}\Big).
\end{align*}
\end{enumerate} 
In words, SAFFRON starts off with an alpha-wealth $(1-\lambda)W_0 <
(1-\lambda)\alpha$, never loses wealth when testing candidate
$p$-values, gains wealth of $(1-\lambda)\alpha$ on every rejection
except the first. If there is a significant fraction of non-nulls, and
the signals are fairly strong, then SAFFRON may make more rejections
than LORD.

To clarify, SAFFRON guarantees FDR control for any $\lambda \in (0,1)$
and any chosen sequence $\{\gamma_j\}_{j=1}^\infty$, but the
algorithm's power, or ability to detect signals, varies as a function
of these parameters.  Given the minimal nature of our assumptions,
there is no universally optimal constant or sequence: specifically, we
do not make assumptions on the frequency of true signals, or on how
strong they are, or on their order, all of which are factors that
affect the power. We discuss reasonable default choices in the
experimental section.

\subsection{A monotone update rule for SAFFRON with predictable $\lambda_t$}\label{sec:monotone-lambda}

Here, we present a lemma that is central to the proof of
FDR control for SAFFRON.  We later use this lemma to prove
\propref{fdp-oracle} and \thmref{fdp-saff} in \secref{proofs}.  Let us first recall and
set up some preliminary notation. In what follows, $\alpha_t,
\lambda_t$ are random variables in $(0,1)$ that always satisfy
$\alpha_t \leq \lambda_t$. We denote the indicator for rejection at
the $t$-th step by $R_t \defn \One{P_t \leq \alpha_t}$, and recall
that since only $p$-values smaller than $\lambda_t$ are candidates for
rejection, we had earlier defined the indicator for candidacy as $C_t
\defn \One{P_t \leq \lambda_t}$. If we denote $\bar C_t = 1 - C_t$,
then it is clear that $R_t \bar C_t = 0$, since $R_t$ and $\bar C_t$ cannot both equal one simultaneously. Also let $R_{1:t}
\defn \{R_1,\dots,R_{t}\}$ and $C_{1:t} \defn \{C_1,\dots,C_t\}$.  As
before, we consider the filtration $\F^t \defn
\sigma(R_{1:t},C_{1:t})$. In what follows, we insist that the
sequences $\{\alpha_t\}_{t=1}^\infty$ and $ \{\lambda_t
\}_{t=1}^\infty$ are \emph{predictable}, meaning that they are
functions of the information available from time 1 to $t-1$ only;
specifically, we insist that $\alpha_t,\lambda_t$ are measurable with
respect to the sigma-field $\F^{t-1}$. We will also require that the
$\{\alpha_t\}$ and $\{\lambda_t\}$ are \emph{monotone}, meaning that $\alpha_t =
f_t(R_{1:t-1}, C_{1:t-1})$ and $\lambda_t =
g_t(R_{1:t-1}, C_{1:t-1})$ for some coordinatewise non-decreasing
functions $f_t:~ \{0,1\}^{2(t-1)} \to (0,1), g_t:~ \{0,1\}^{2(t-1)} \to (0,1)$. 
A generalization of our default update rule described in \secref{const-lbd} for time-varying $\lambda_t$ is:
\begin{align}
    \alpha_1 &= (1-\lambda_1)\gamma_1
    W_0, ~ \text{for some } W_0\leq \alpha\\
    t\geq 2: \alpha_t &= (1-\lambda_t)\Big(W_0\gamma_{t - C_{0+}} + (\alpha -
W_0) \gamma_{t-\tau_1 - C_{1+}} + \sum_{j \geq 2} \alpha
\gamma_{t-\tau_j - C_{j+}}\Big).
\end{align}
One simple way to ensure the monotonicity of the above rule for $\alpha_t$, is to insist that $\lambda_t$ is non-decreasing in $\alpha_t$. This is satisfied, for example, if $\{\lambda_t\}$ is a deterministic sequence of constants, or when $\lambda_t = \alpha_t$, as in the case of alpha-investing. We state this formally below, and prove it in \secref{proofs}.

\begin{lemma}
\label{lem:monotone}
If we choose $\lambda_t = h_t(\alpha_t)$ for some non-decreasing function $h_t:(0,1)\rightarrow (0,1)$, such that $h_t(x)\geq x$, the update rule (6-7) guarantees that $\alpha_t$ and $\lambda_t$ are monotone.
\end{lemma}

We note that the assumption of $\lambda_t$ being non-decreasing in $\alpha_t$ is \emph{not} necessary but simply sufficient; indeed, any other choice of monotone $\alpha_t$ and $\lambda_t$ that controls $\fdphat_{\lambda}(t)$ also results in a valid FDR controlling procedure.

\section{Relationship to other procedures}

Here, we compare SAFFRON to existing procedures in the
literature, emphasizing commonalities that allow us to give a unified view of seemingly
disparate algorithms.

\subsection{Alpha-investing (AI)} 

Even though the motivation that we have
presented for SAFFRON relates it to the LORD algorithm, we find it
interesting that the original AI algorithm of
\citet{foster2008alpha} is recovered by choosing $\lambda_j=\alpha_j$
in $\fdphat_\lambda$, and attempting to ensure that
$\fdphat_\lambda(t) \leq \alpha$ for all times $t \in \N$. In order to
see this fact, first note that with this choice of $\lambda_j$, the
indicator $\One{P_j > \lambda_j}$ simply indicates when the $j$-th
hypothesis is not rejected.  Consequently, the numerator of
$\fdphat_\lambda$ reads as $\sum_{j \leq t}
\frac{\alpha_j}{1-\alpha_j} \One{j \notin \cR(t)} $. Hence, ensuring
that $\fdphat_\lambda(t) \leq \alpha$ at all times $t \in \N$, is
equivalent to ensuring that $\sum_{j \leq t}
\frac{\alpha_j}{1-\alpha_j} \One{j \notin \cR(t)} $ never exceeds
$\alpha (|\cR(t)| \vee 1)$, which, in the language of alpha-investing,
is equivalent to ensuring that the algorithm's wealth never becomes
negative.\footnote{Recall that the AI algorithm starts
  off with an alpha-wealth of $\alpha$, reduces its alpha-wealth by
  $\frac{\alpha_j}{1-\alpha_j}$ after tests that fail to reject, and
  increase the wealth by $\alpha$ on rejections.} Just as
\citet{RYWJ17} were able to reinterpret and rederive LORD in terms of
a particular estimate of the FDP, the current work allows us to
reinterpret and rederive AI in terms of SAFFRON's
FDP. 

Nevertheless, despite these  similarities, SAFFRON's update rule for $\alpha_j$ as stated in \secref{monotone-lambda} is different from the update used in AI. Originally \cite{foster2008alpha}, $\alpha_j$ was set to a fraction of the available wealth $W_j$; however, this simple update prevents alpha-investing from being a monotone procedure, meaning that there is no guarantee that $f_j:(R_1,...,R_{j-1})\mapsto\alpha_j$ is a coordinatewise nondecreasing function. For this reason, the original alpha-investing provably controls only $\mfdr$, and not the FDR. However, we may derive a novel monotone version of AI by using $\lambda_j=\alpha_j$ in SAFFRON's update rule from \secref{monotone-lambda}, immediately yielding $\fdr$ control under independence. Simulations indicate that this new monotone SAFFRON-AI algorithm performs comparably to the original non-monotone AI, or sometimes even outperforms it, as demonstrated in the first subplot of \figref{SAFFvsAI}. Further, as a consequence of monotonicity, SAFFRON-AI allocates $\alpha_j$ in a more stable manner compared to the non-monotone AI, as shown in the third subplot of \figref{SAFFvsAI}.

\begin{figure}[h!]
\centerline{\includegraphics[width=0.3\textwidth]{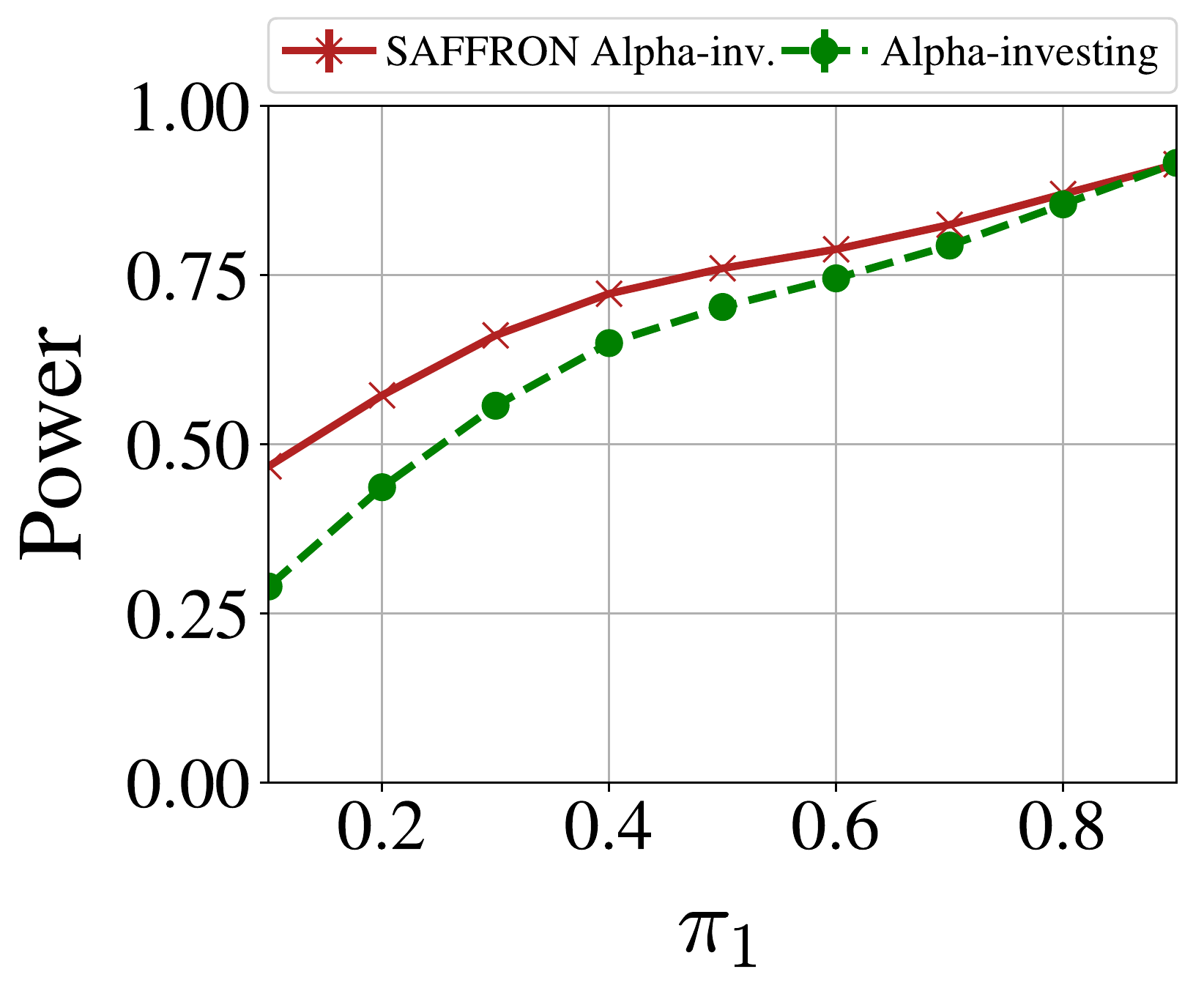}
\includegraphics[width=0.3\textwidth]{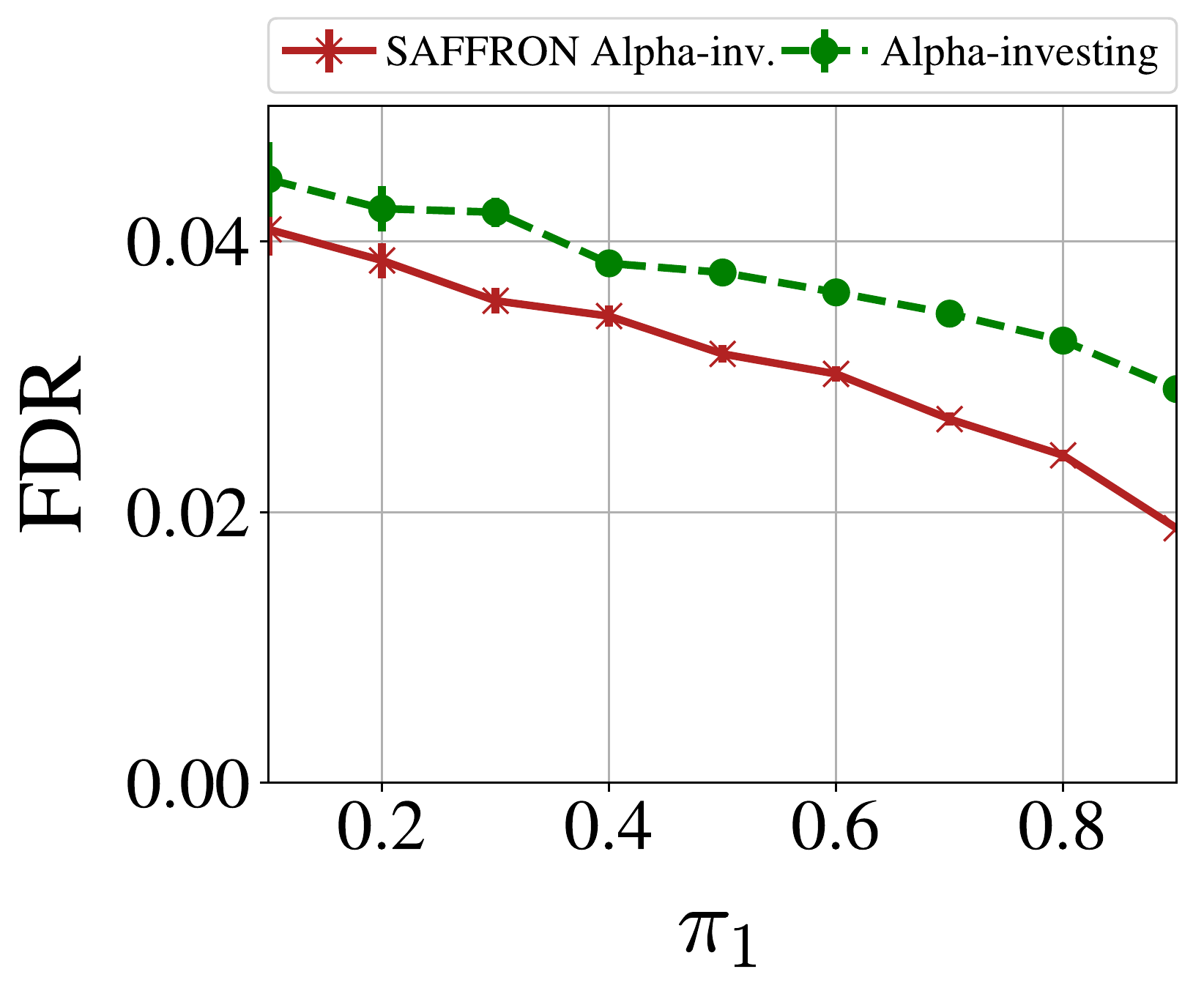}\includegraphics[width=0.3\textwidth]{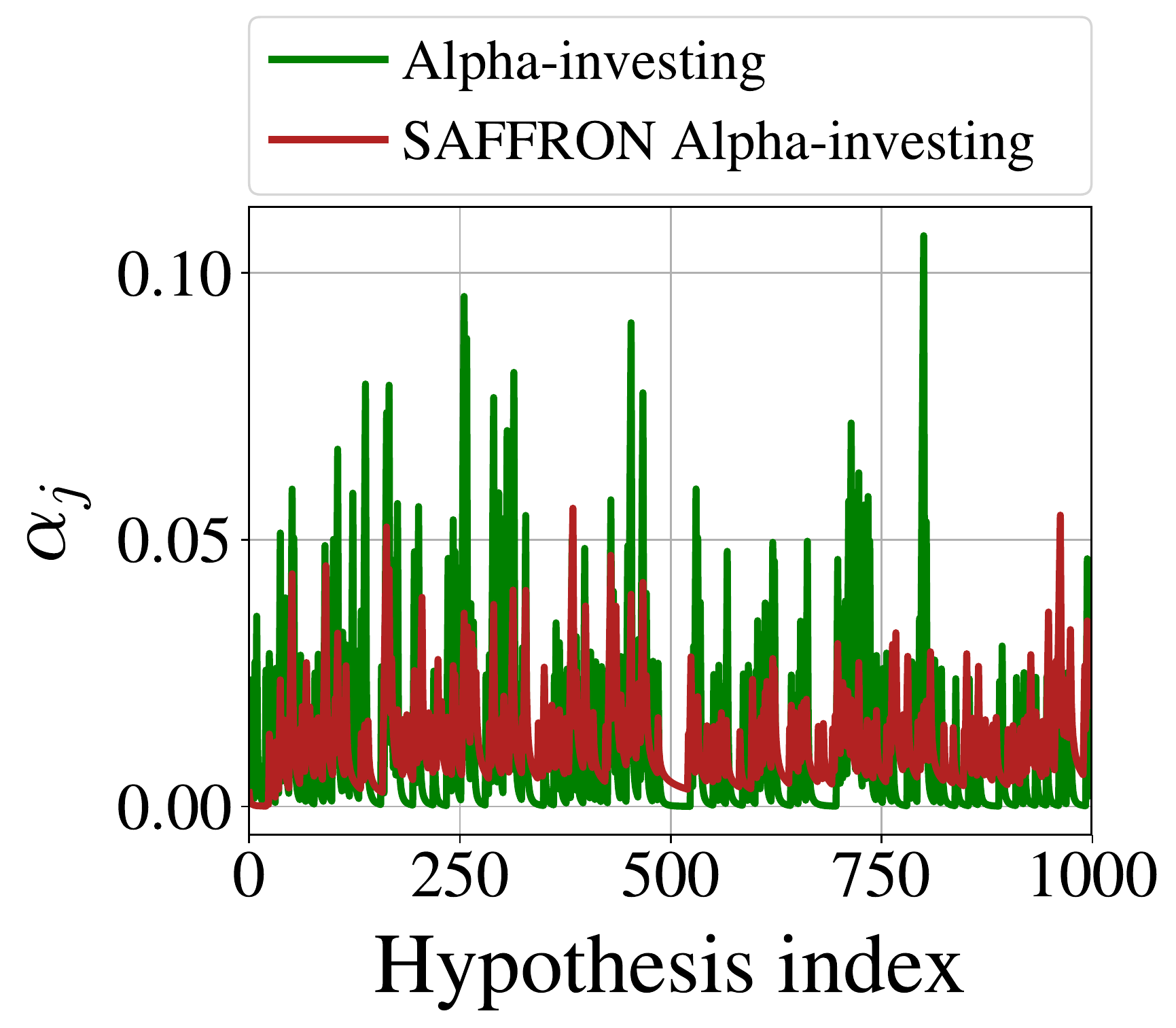}}
\caption{Statistical power and FDR versus fraction of non-null
  hypotheses $\pi_1$ (left), and allocated $\alpha_j$ versus hypothesis index (right), for SAFFRON with $\lambda_j=\alpha_j$ and the original alpha-investing (at target level $\alpha = 0.05$). The observations under the alternative are Gaussian
  with $\mu_i\sim N(3,1)$ and standard deviation 1, and are converted
  into one-sided $p$-values as $P_i=\Phi(-Z_i)$. SAFFRON-AI is sometimes more powerful than AI (first subplot), and also more stable (third subplot), across a variety of choices of tuning parameters for both algorithms.}
\label{fig:SAFFvsAI}
\end{figure}

\subsection{Storey-BH} 

In offline multiple testing, where all $p$-values are
immediately available to the scientist, the Benjamini-Hochberg (BH)
procedure~\cite{BH95} is a classical method for guaranteeing FDR
control. Although the initial motivation for the BH method was
different, it was reinterpreted by Storey et
al. \cite{Storey02,Storey04} in the following manner. Since the small
$p$-values are more likely to be non-null, suppose that one rejects
all $p$-values below some fixed threshold $s \in (0,1)$, meaning that
$\cR(s) = \{i : P_i \leq s \}$. Then, an oracle estimate for the FDP
is given by
\begin{align*}
\fdp^*_{\BH}(s) := \frac{|\nulls| \cdot s}{|\cR(s)|}.
\end{align*}
The numerator is a sensible estimate because the nulls are uniformly
distributed, and hence we would expect about $|\nulls|\cdot s$ many
nulls to be below $s$. This is an ``oracle'' estimate because the
scientist does not know $|\nulls|$. Ideally, one would like to choose
a data-dependent $s$ using the rule
\begin{align*}
s^* \defn \max\{s : \fdp^*_{\BH}(s) \leq \alpha \},
\end{align*}
and then reject the set $\cR(s^*)$. Given $n$ $p$-values, the BH procedure overestimates the
oracle FDP by the empirically computable quantity
\begin{align*}
\fdphat_{\BH}(s) \defn \frac{n \cdot s}{|\cR(s)|},
\end{align*}
and then rejecting the set $\cR(\widehat s_{\BH})$, where $\widehat
s_{\BH} \defn \max\{s : \fdphat_{\BH}(s) \leq \alpha \}$.  On
interpreting the BH procedure in terms of an estimated $\fdp$, Storey
et al.~\cite{Storey02,Storey04} noted that when the $p$-values are
independent, the estimate $\fdphat_{\BH}$ is unnecessarily
conservative. Indeed, when the $p$-values are exactly uniform, it is
known \cite{BY01,ramdas2017unified} to satisfy the stronger bound
$\fdr = \alpha |\nulls| / n$, which demonstrates that BH underutilizes
the FDR budget of $\alpha$ provided to it. Instead, Storey et
al.\ pick a constant $\lambda \in (0,1)$, and calculate
\begin{align*}
\fdphat_{\StBH}(s) \defn \frac{n \cdot s \cdot \pihat}{|\cR(s)|},
\end{align*}
where the unknown proportion of nulls $\pi_0 = |\nulls|/n$ is estimated as
\begin{align*}
\pihat \defn \frac{1 + \sum_{i=1}^n \One{P_i > \lambda} }{n(1-\lambda)}.
\end{align*}
Then, this procedure, which we refer to as ``Storey-BH,'' calculates
$\widehat s_{\StBH} \defn \max\{s : \fdphat_{\StBH}(s) \leq \alpha \}$
and rejects the set $\cR(\widehat s_{\StBH})$ which satisfies the
bound $\fdr \leq \alpha$.  Storey et al.\ demonstrated via simulations
that the Storey-BH procedure is typically more powerful than the BH
procedure, the improvement increasing with the fraction of non-nulls,
and the strength of underlying signal. Procedures such as Storey-BH
are known in the multiple testing literature as \emph{adaptive} procedures,
since they adapt to the unknown proportion of nulls.

Returning to the setting of online FDR, what matters is not the the
proportion of nulls $\pi_0$, but instead a running estimate of the
amount of alpha-wealth that was spent testing nulls thus far; this
difference arises because, unlike the offline setting where all
$p$-values are compared to the same level $\widehat s$, different
$p$-values have to pass different thresholds $\alpha_i$.  In light of
the above discussion, and comparing to the derivation of SAFFRON, it
should be apparent that Storey-BH is to BH as SAFFRON is to
LORD. Indeed, both LORD and BH result from a trivial upper bound on an
oracle estimate of the FDP, and both Storey-BH and SAFFRON
respectively try to better estimate the proportion of nulls or the
amount of alpha-wealth spent on testing nulls.

It is in the above sense that SAFFRON is an adaptive online FDR
method. As mentioned earlier in this section, Foster and Stine's
alpha-investing procedure is a special case of SAFFRON; hence,
strictly speaking, alpha-investing would count as the first adaptive
online FDR procedure (even though the motivation for alpha-investing
in the original paper was entirely different, and did not mention
estimating the FDP, or adaptivity). However, as noted in simulations
by \citet{javanmard2016online}, and re-confirmed in our simulations,
alpha-investing seems \emph{less} powerful than the non-adaptive
algorithm LORD (and LORD++). As shown by simulations in the sequel,
SAFFRON with constant $\lambda=1/2$ is more powerful than LORD across
a variety of signal proportions and strengths, and hence is arguably
the first adaptive algorithm in the online FDR setting that can
compete with the non-adaptive algorithms.

\subsection{Accumulation tests, like SeqStep} 

Note that $\EE{2I(P>1/2)} \geq 1$  for null p-values (with equality when they are exactly uniformly distributed, simply because $\int_0^1 2I(p>1/2)dp = 1$). One may actually use any non-decreasing function $h$ such that $\int_0^1 h(p) dp$ in the formula for $\fdphat_\lambda$. Such \emph{accumulation functions} were studied in the (offline) context of ordered testing~\cite{li2017accumulation}, and may seamlessly be transferred to the online setting considered here, yielding mFDR control using the same proof. In initial experiments, the use of other functions is not advantageous, and under some additional assumptions in the offline ordered testing setting, the aforementioned authors argued that the step function $(1-\lambda)^{-1}I(I > \lambda)$ is asymptotically optimal for power. In this light, SAFFRON can also be seen as an online analog of adaptive SeqStep \cite{lei2016power}, which is a variant of Selective SeqStep \cite{barber2015controlling} and SeqStep \cite{li2017accumulation}.



\section{Numerical simulations}
\label{sec:sims}

In this section, we provide the results of some
numerical experiments that compare the performance of SAFFRON with
current state-of-the-art algorithms for online FDR control, namely the
aforementioned LORD and alpha-investing procedures.\footnote{The code for all simulations described in this section is available at: https://github.com/tijana-zrnic/SAFFRONcode} In particular, for
each method, we provide empirical evaluations of its power while
ensuring that the FDR remains below a chosen value. We consider two
settings, one in which the $p$-values are computed from Gaussian
observations, and another in which the $p$-values under the
alternative are drawn from a beta
distribution~\cite{javanmard2016online}. The following two subsections
separately analyze these experimental settings; in both cases, SAFFRON
outperforms the competing algorithms, with mild dependence on the exact
 choice of sequence $\{\gamma_j\}$. In all our experiments we control the FDR under $\alpha = 0.05$ and
estimate the FDR and power by averaging over 200 independent
trials.

As was previously mentioned, the constant sequence $\lambda_j
= 1/2$ for all $j$ was found to be particularly successful, so this is
our default choice in comparison with prior work and we drop the index for
simplicity. In a separate subsection, however, we also compare the performance of SAFFRON with $\lambda_j
= 1/2$ to the alpha-investing version of SAFFRON, obtained by setting $\lambda_j = \alpha_j$. We do so over different choices of sequence $\{\gamma_j\}$.


\subsection{Testing with Gaussian observations}

We use the simple experimental setup of testing the mean of a Gaussian
distribution with $T = 1000$ components. More precisely, for each
index $i \in \{1, \ldots, T \}$, the null hypothesis takes the form
$H_i: \mu_i=0$. The observations consist of independent Gaussian
variates $Z_i \sim N(\mu_i,1)$, which are converted into one-sided
$p$-values using the transform $P_i = \Phi(-Z_i)$, where $\Phi$ is the
standard Gaussian CDF. The motivation for one-sided conversion lies in A/B testing, where one wishes to detect \emph{larger} effects, not smaller. The parameter $\mu_i$ is chosen according to
the following mixture model:
\begin{align*}
\mu_i = \begin{cases} 0 & \mbox{with probability $1-\pi_1$} \\
F_1 & \mbox{with probability $\pi_1$,}
  \end{cases}
\end{align*}
where the random variable $F_1$ is of the form $N(\mu_c,1)$ for some
constant $\mu_c$. We ran simulations for $\mu_c \in \{2,3\}$, thus
seeing how changing signal strength
 affects the performance of SAFFRON.

In what follows, we compare SAFFRON's achieved power and FDR
to those of LORD and alpha-investing. The
constant infinite sequence $\gamma_j\propto\frac{\log(j\vee2)}{j
  e^{\sqrt{\log j}}}$, where the proportionality constant is
determined so that the sequence sums to one, was shown to be
asymptotically optimal for testing Gaussian means via the LORD method in the paper~\cite{javanmard2016online}. Since SAFFRON loses
wealth only when testing non-candidates whereas LORD loses wealth at
every step, it is expected to behave more conservatively and not use
up its wealth at the same rate, conditioned on both using the same
sequence $\{\gamma_j\}$. For this reason, informally speaking, it can
reuse this leftover wealth, hence the sequence $\{\gamma_j\}$ chosen
for SAFFRON is more aggressive, in the sense that more wealth is
concentrated around the beginning of the sequence. In particular, we choose sequences of the form $\gamma_j\propto j^{-s}$, where the
parameter $s >1$ controls the aggressiveness of the procedure; the
greater the constant $s$, the more wealth is concentrated around small
values of $j$. We also consider these sequences for LORD, thus observing the difference in performance resulting from using a more aggressive sequence in the regime of a finite sequence of $p$-values.

In \figref{SAFFRONmean2} and \figref{LORDmean2} we consider $F_1 =
N(2,1)$, and show how the level of aggressiveness of the sequence
$\{\gamma_j\}$ affects the power and FDR of SAFFRON and LORD
respectively. \figref{comparisonmean2} compares alpha-investing,
SAFFRON and LORD, the latter two using the highest performing sequence
chosen among six possible sequences, in the same testing
scenario. \figref{SAFFRONmean3}, \figref{LORDmean3} and
\figref{mean3} demonstrate these results in the same order
for a similar but somewhat easier testing problem, with $F_1=N(3,1)$. Experiments indicate that increasing the fraction of non-null hypotheses allows SAFFRON to achieve a faster increase of power than LORD, thus performing considerably better than both LORD and the alpha-investing procedure in settings with a great number of non-null observations.

\begin{figure}[H]
\centerline{\includegraphics[width=0.35\textwidth]{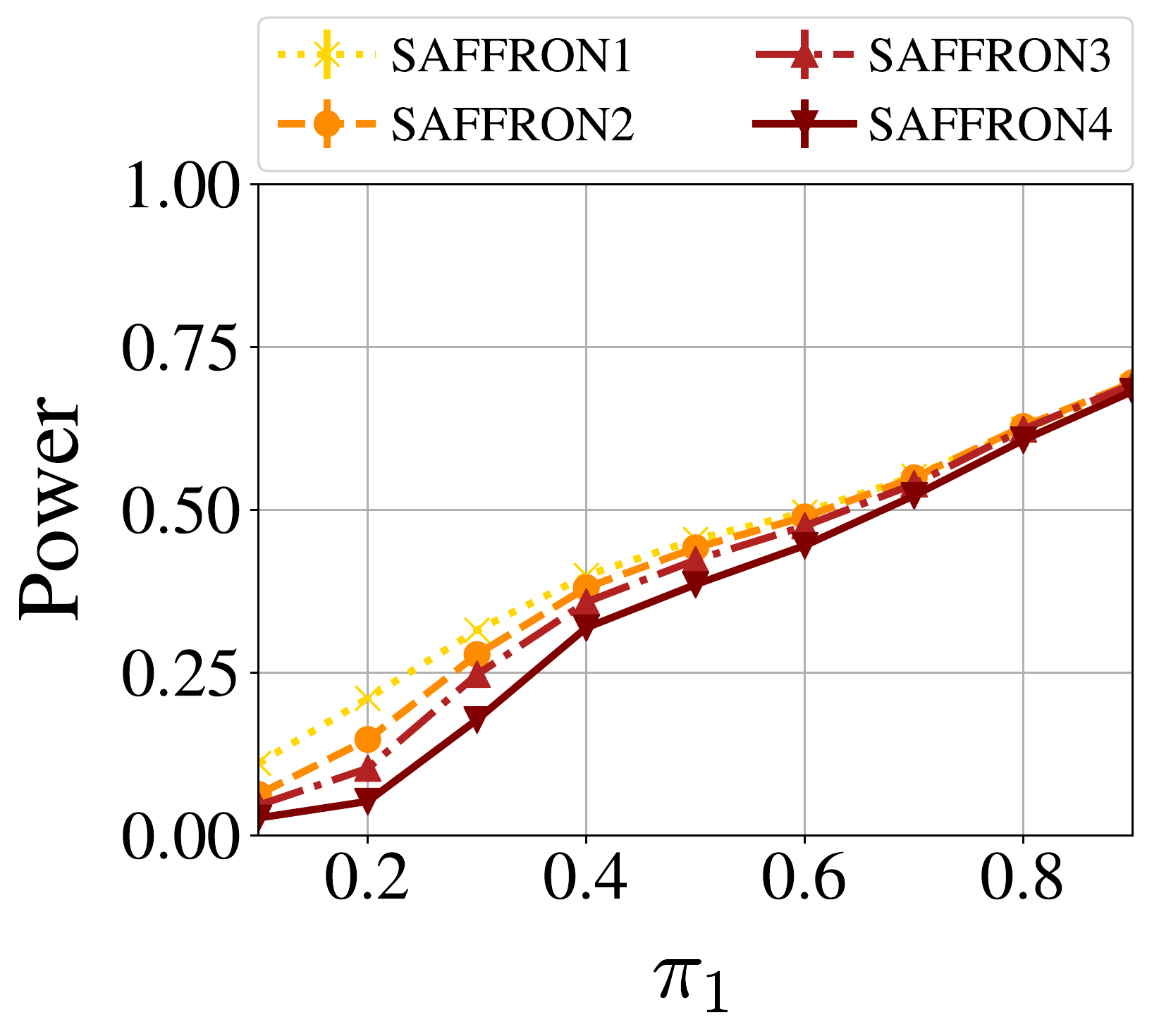}
\includegraphics[width=0.35\textwidth]{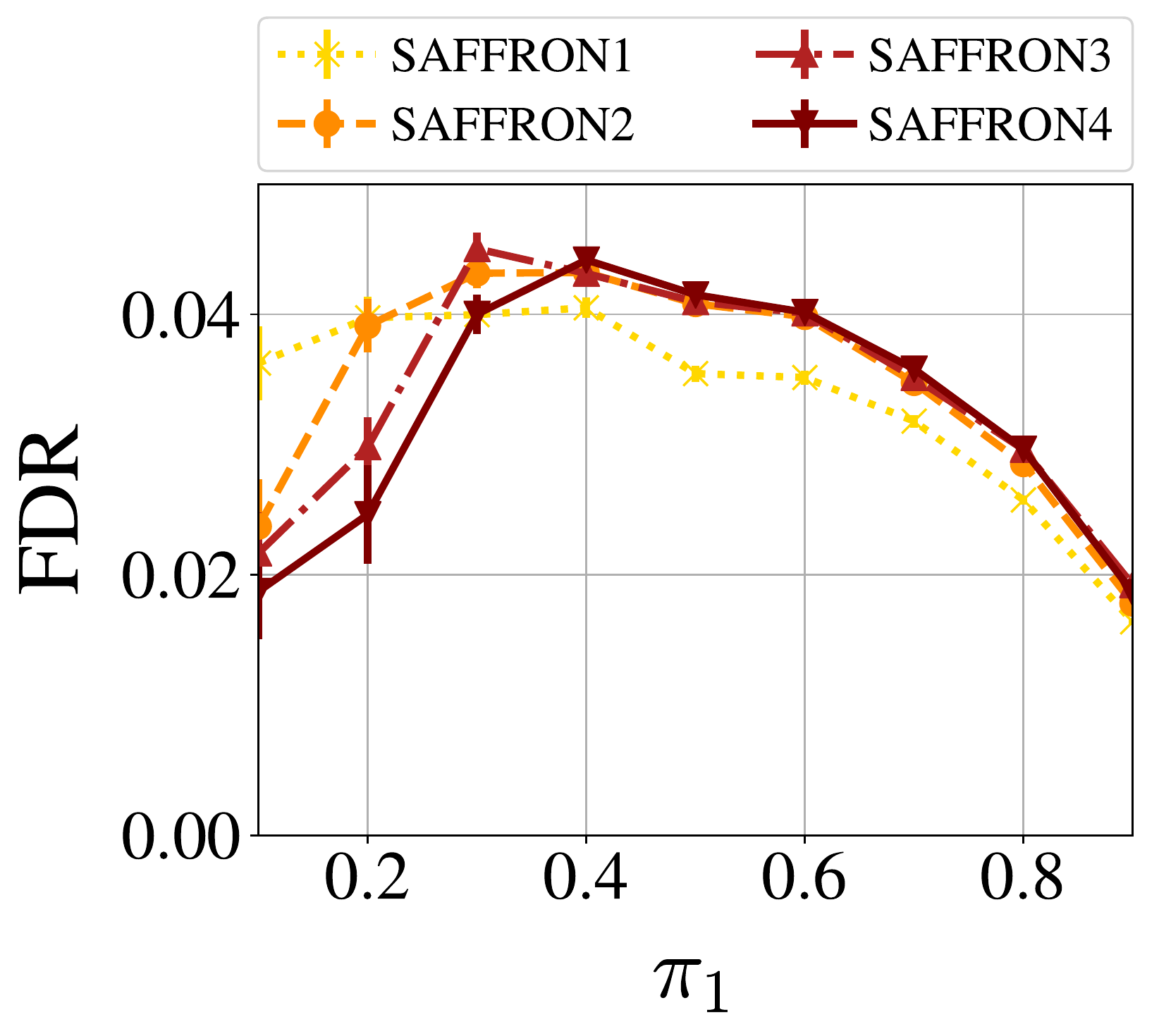}}
\caption{Statistical power and FDR versus fraction of non-null
  hypotheses $\pi_1$ for SAFFRON (at target level $\alpha = 0.05$)
  using four different sequences $\{\gamma_j\}$ of increasing
  aggressiveness. The observations under the alternative are $N(\mu_i,1)$
  with $\mu_i\sim N(2,1)$, and are converted
  into one-sided $p$-values as $P_i=\Phi(-Z_i)$. (See also \figref{comparisonmean2}.)}
\label{fig:SAFFRONmean2}
\end{figure}

\begin{figure}[H]
  \centerline{\includegraphics[width=0.35\textwidth]{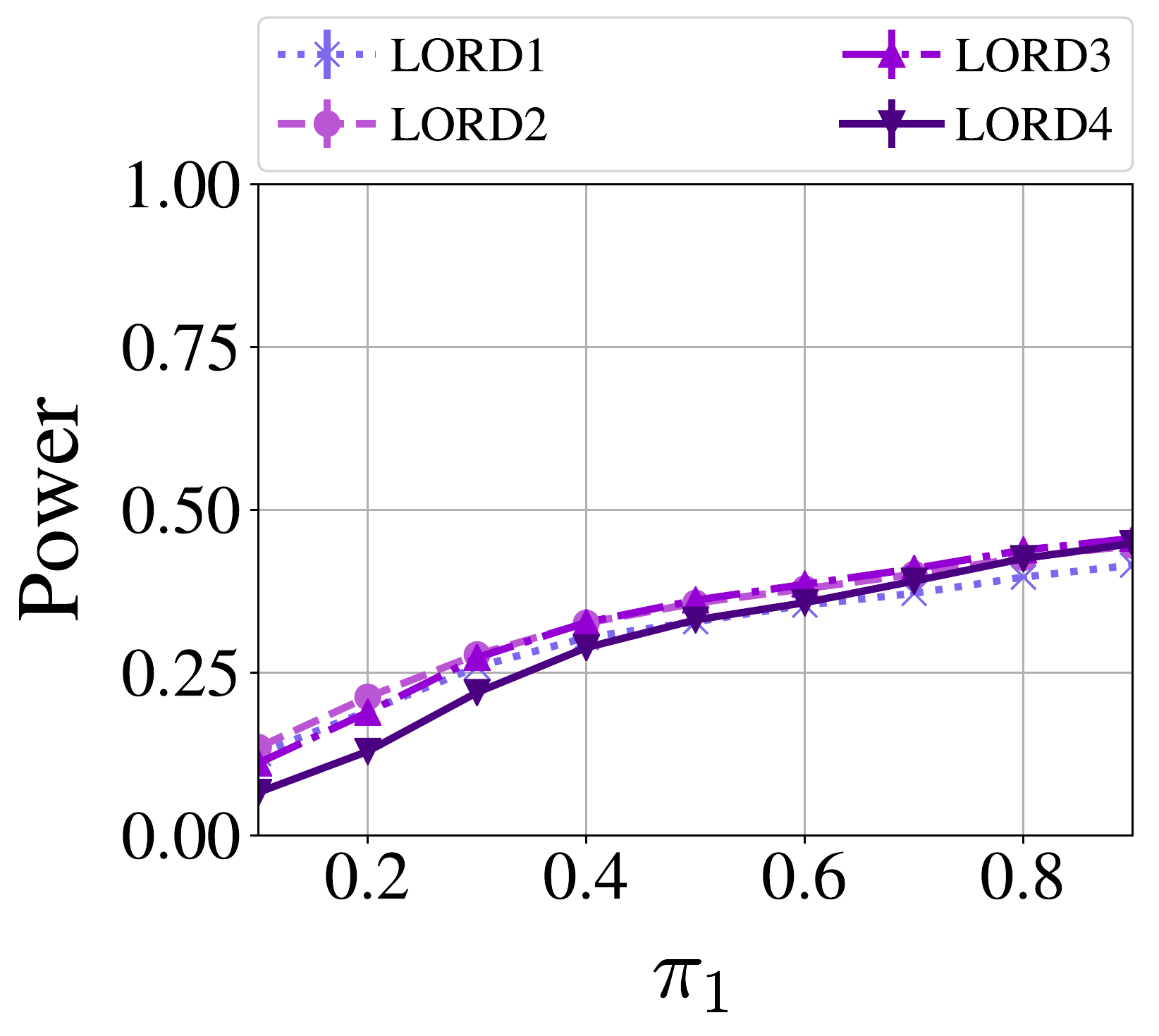}
\includegraphics[width=0.35\textwidth]{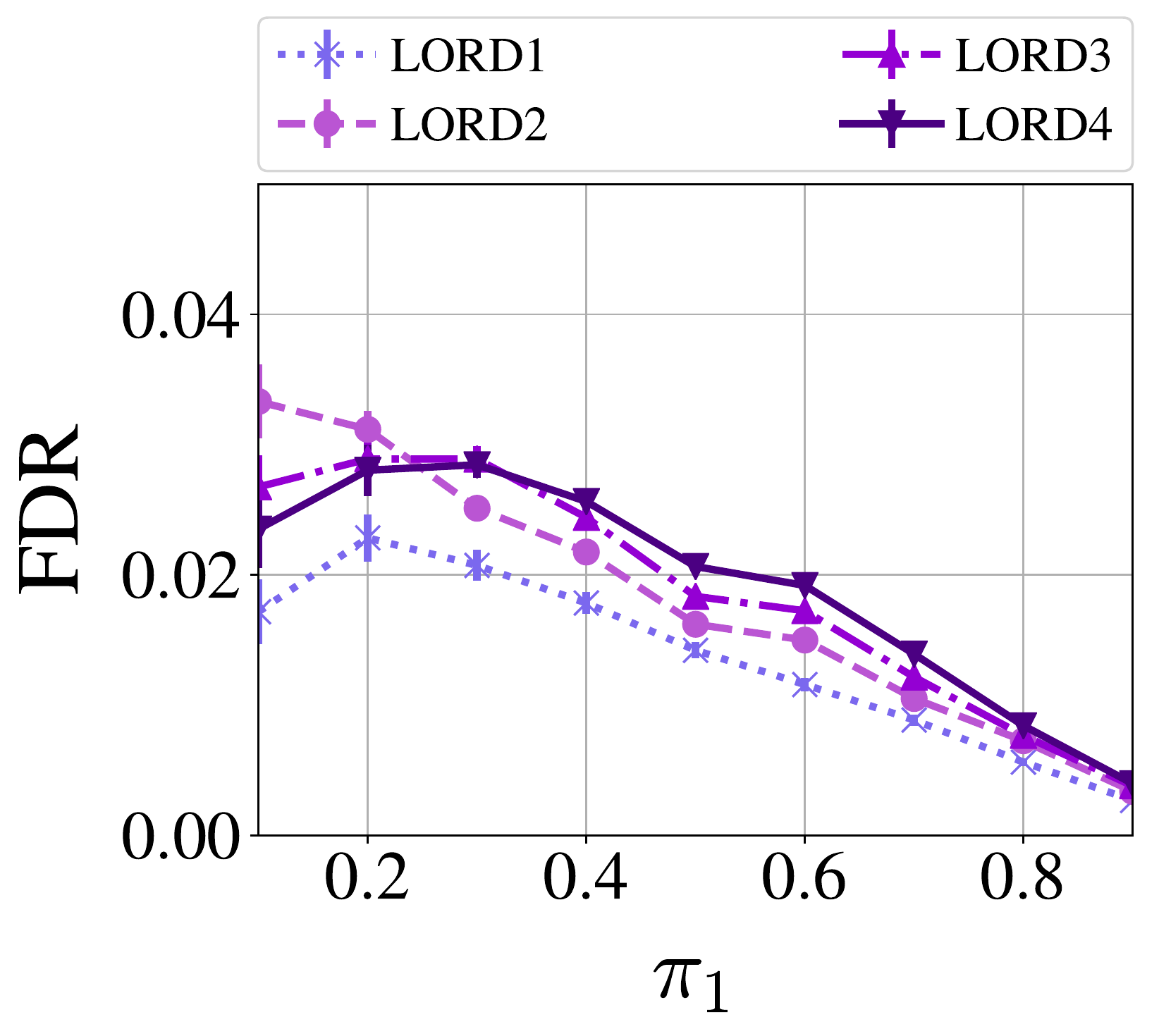}}
  
\caption{Statistical power and FDR versus fraction of non-null
  hypotheses $\pi_1$ for LORD (at target level $\alpha = 0.05$) using
  four different sequences $\{\gamma_j\}$ of increasing
  aggressiveness. The LORD1 method uses the sequence proposed in
  the paper~\cite{javanmard2016online}. The observations under the alternative are $N(\mu_i,1)$ with
  $\mu_i\sim N(2,1)$, and are converted into
  one-sided $p$-values as $P_i=\Phi(-Z_i)$. (See also \figref{comparisonmean2}.)}
\label{fig:LORDmean2}
\end{figure}

\begin{figure}[H]
\centerline{\includegraphics[width=0.35\textwidth]{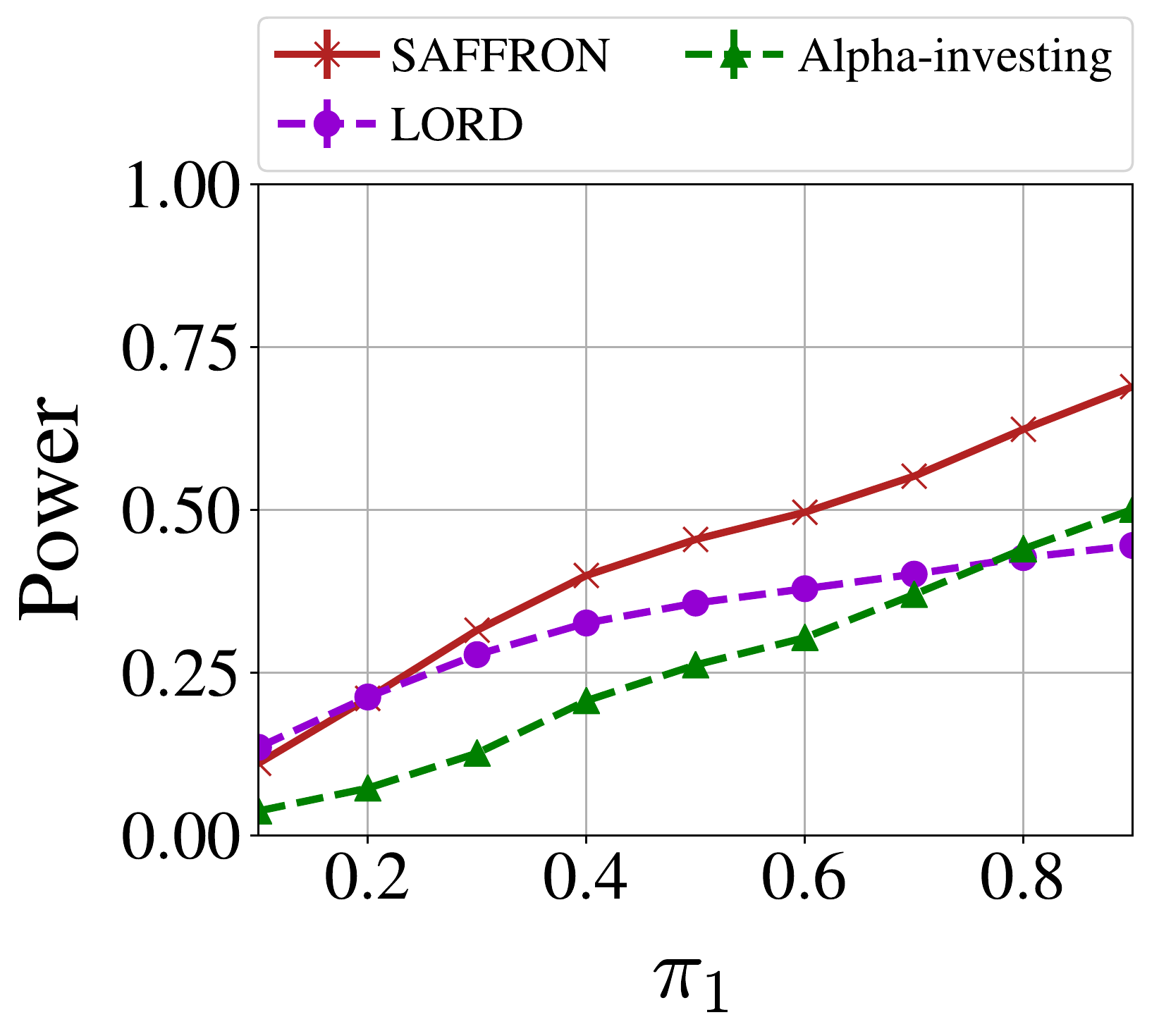}
\includegraphics[width=0.35\textwidth]{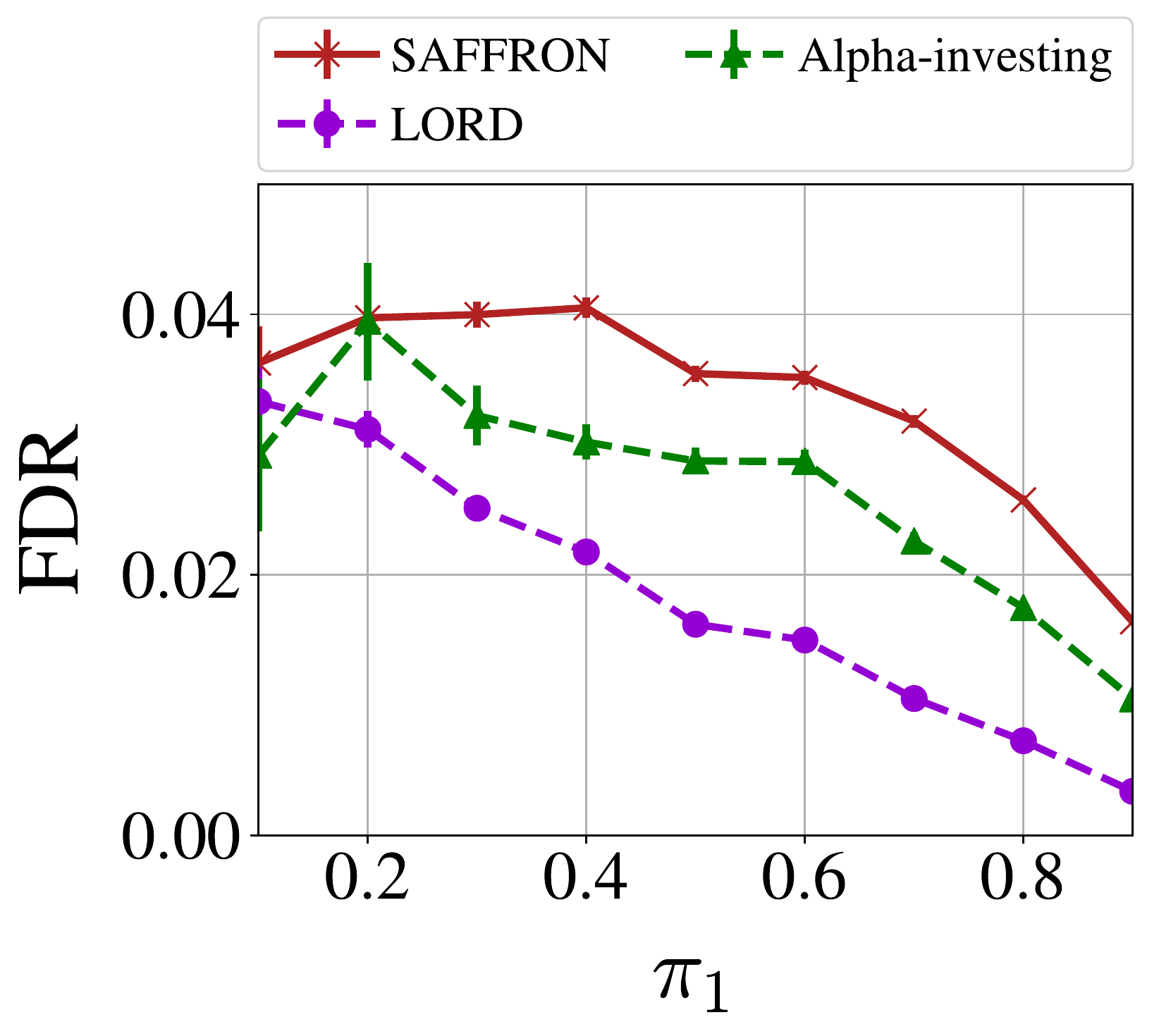}}
\caption{Statistical power and FDR versus fraction of non-null
  hypotheses $\pi_1$ for SAFFRON, LORD and alpha-investing (at target
  level $\alpha = 0.05$), the first two using the sequence
  $\{\gamma_j\}$ which achieves the highest power for each of them
  (chosen over six sequences of varying aggressiveness). The
  observations under the alternative are $N(\mu_i,1)$ with $\mu_i\sim
  N(2,1)$, and are converted into one-sided
  $p$-values as $P_i = \Phi(-Z_i)$.}
\label{fig:comparisonmean2}
\end{figure}

\begin{figure}[H]
\centerline{\includegraphics[width=0.35\textwidth]{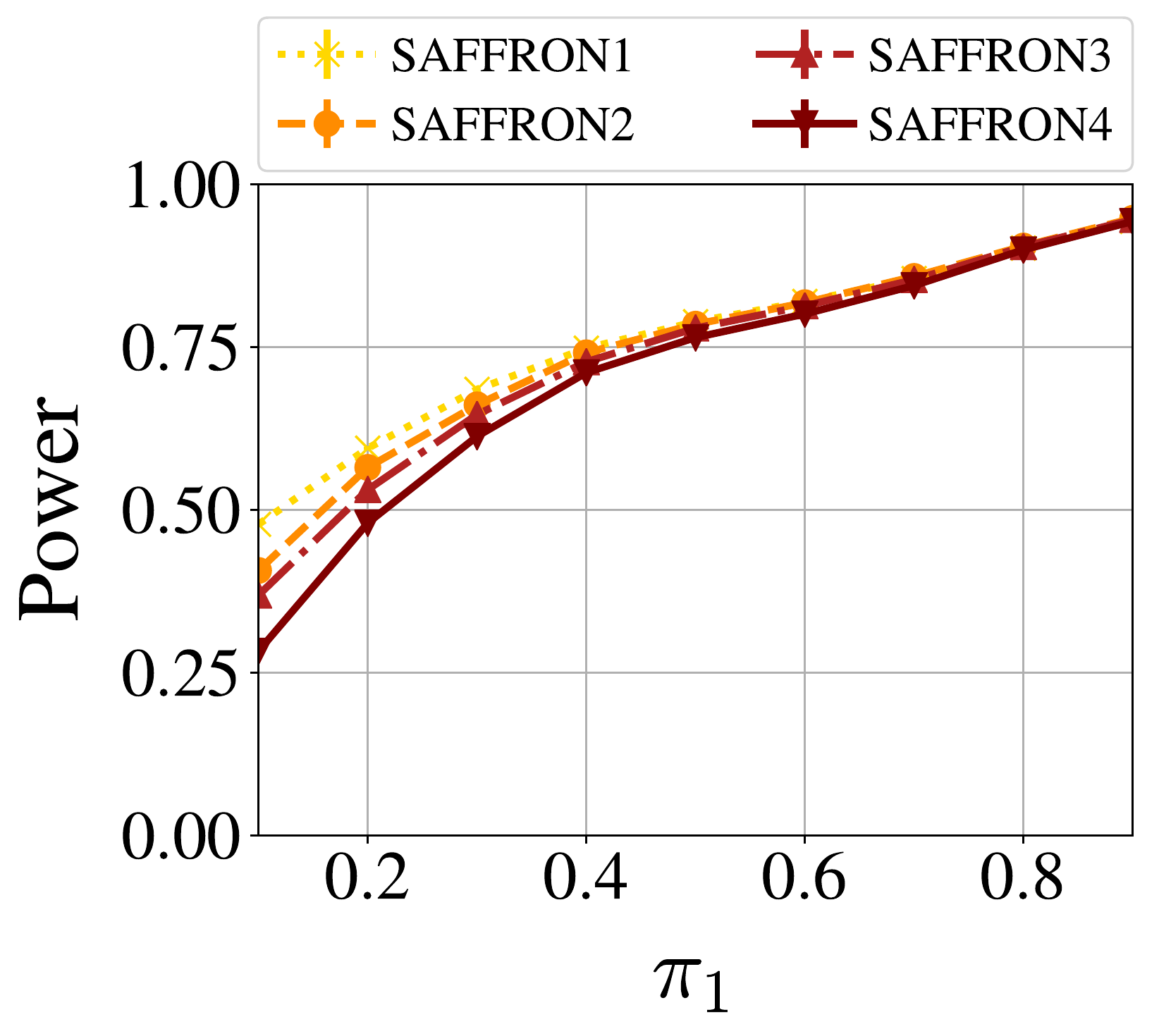}
\includegraphics[width=0.35\textwidth]{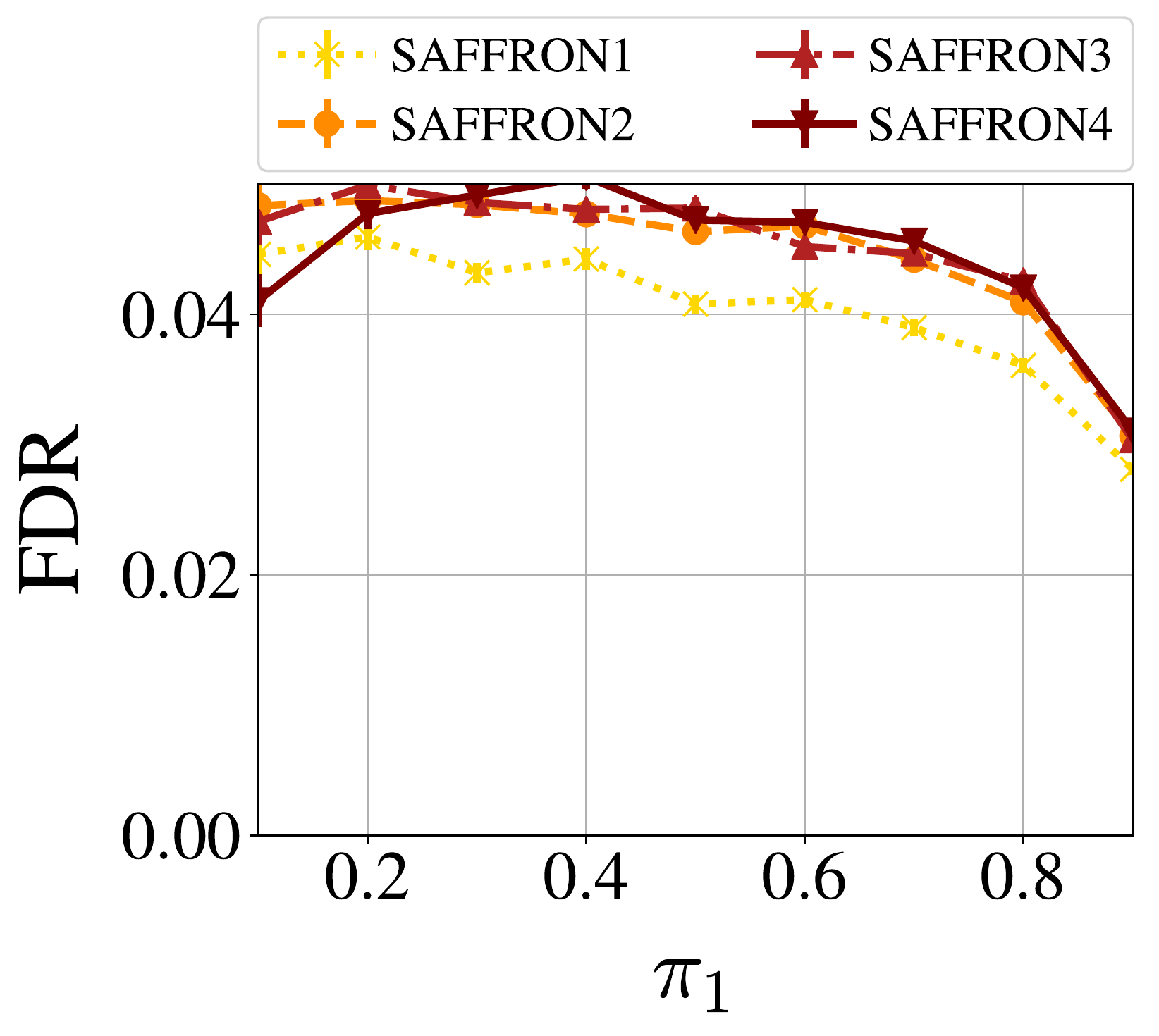}}
\caption{Statistical power and FDR versus fraction of non-null
  hypotheses $\pi_1$ for SAFFRON (at target level $\alpha = 0.05$)
  using four different sequences $\{\gamma_j\}$ of increasing
  aggressiveness. The observations under the alternative are $N(\mu_i,1)$
  with $\mu_i\sim N(3,1)$, and are converted
  into one-sided $p$-values as $P_i=\Phi(-Z_i)$. (See also \figref{mean3}.)}
\label{fig:SAFFRONmean3}
\end{figure}

\begin{figure}[H]
\centerline{\includegraphics[width=0.35\textwidth]{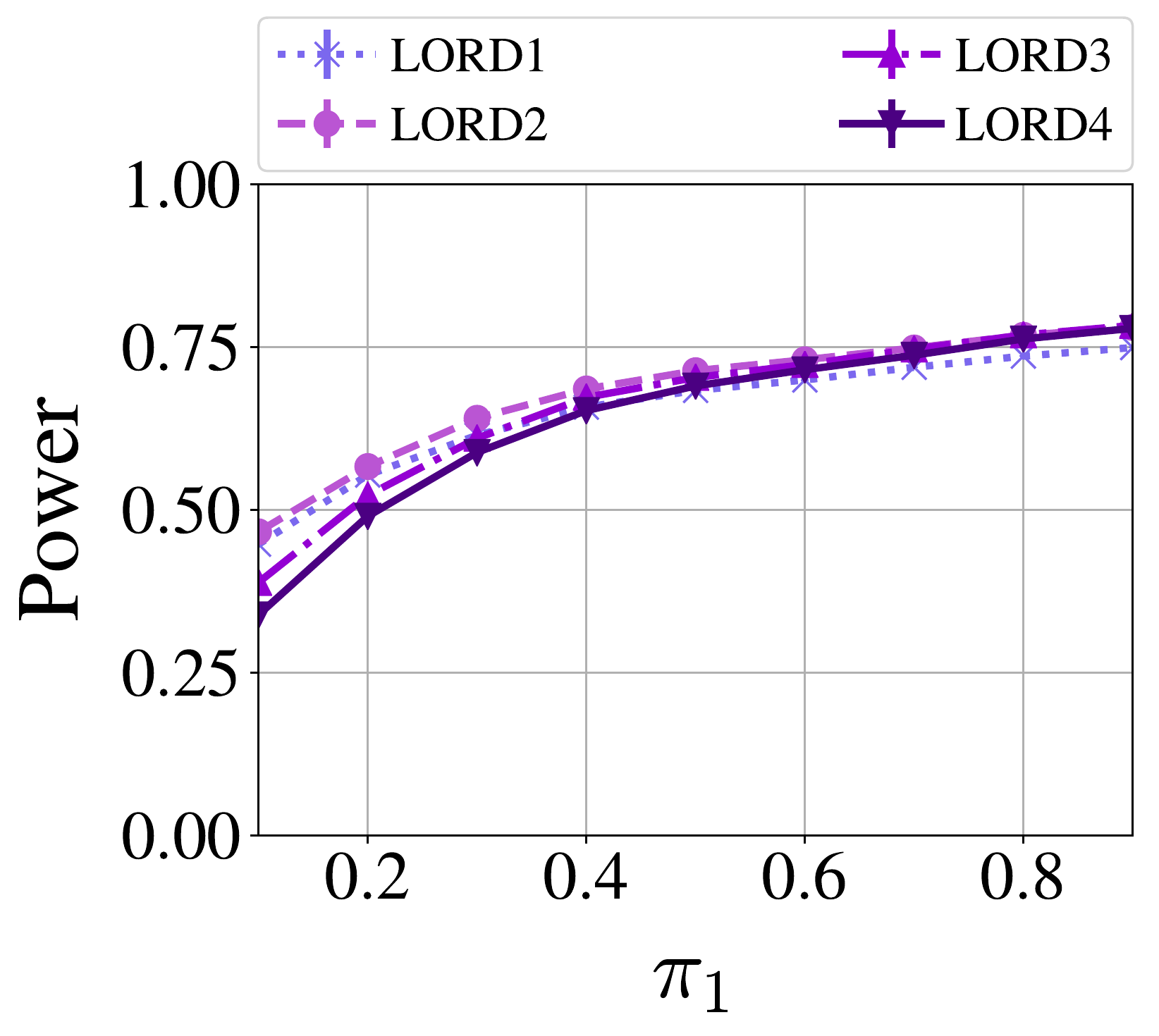}
\includegraphics[width=0.35\textwidth]{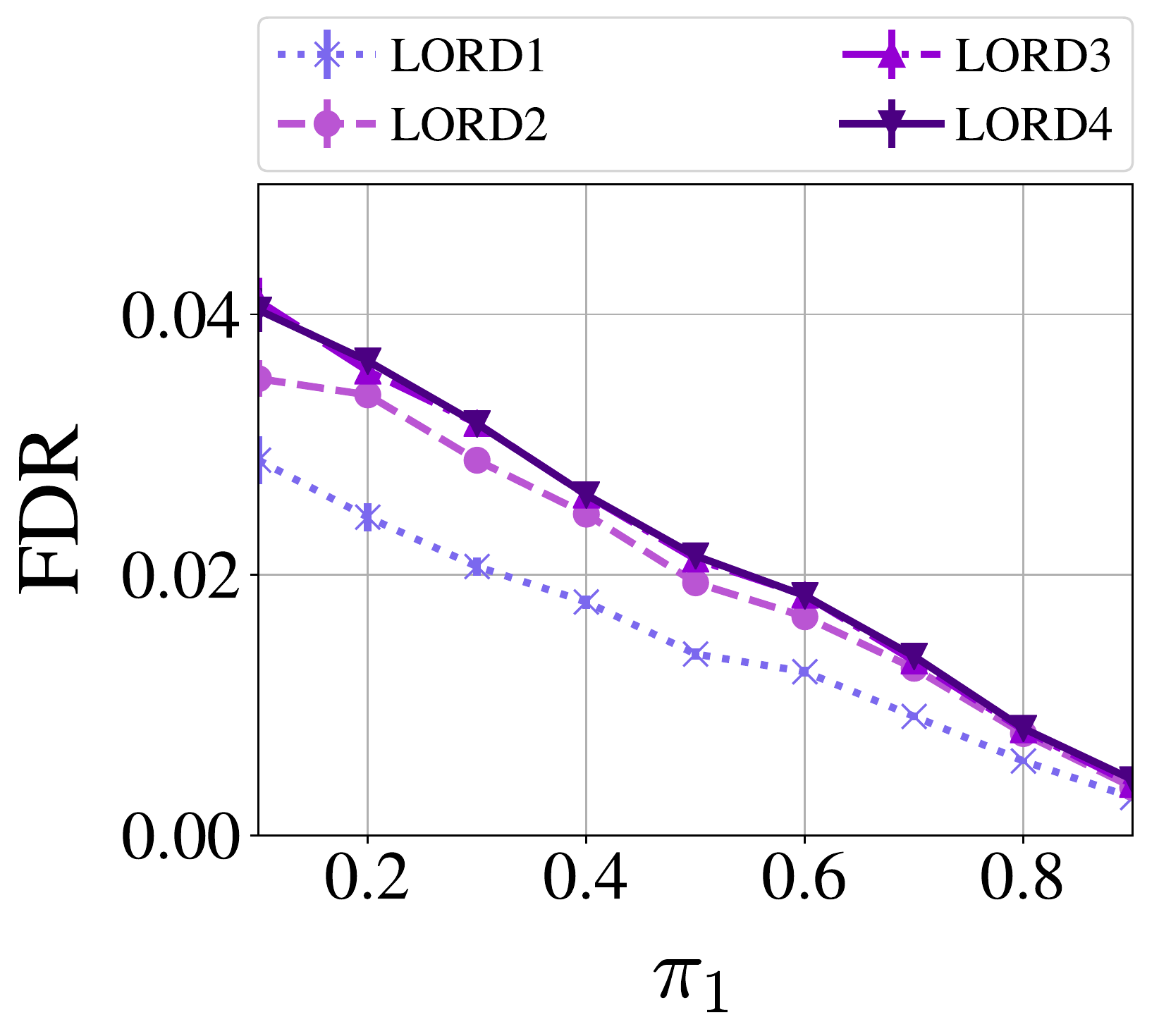}}
\caption{Statistical power and FDR versus fraction of non-null
  hypotheses $\pi_1$ for LORD (at target level $\alpha = 0.05$) using
  four different sequences $\{\gamma_j\}$ of increasing
  aggressiveness. The LORD1 method uses the sequence proposed in
  the paper~\cite{javanmard2016online}. The observations under the alternative
  are $N(\mu_i,1)$ with $\mu_i\sim N(3,1)$, and
  are converted into one-sided $p$-values as $P_i=\Phi(-Z_i)$. (See also \figref{mean3}.)}
\label{fig:LORDmean3}
\end{figure}

\subsection{Testing with beta alternatives}

In this setting we generate the $p$-value sequence according to the
following model:
\begin{align*}
  P_i \sim \begin{cases} \text{Unif}[0,1], & \mbox{with probability $1-\pi_1$} \\ \text{Beta}(m,n), &
    \mbox{with probability $\pi_1$,}
  \end{cases}
\end{align*}
where $i\in[T]$ and $T=1000$, as before. Again we compare the
performance of SAFFRON, alpha-investing and LORD in terms of the
achieved power with the FDR controlled under a chosen level. For LORD,
the asymptotically optimal sequence $\{\gamma_j\}$ was derived in the paper~\cite{javanmard2016online} and is of the form
$\gamma_j\propto(\frac{1}{j}\log j)^{1/m}$ for $m<1$ and $n\geq1$. As
in the Gaussian case, for SAFFRON and additionally for LORD we
consider the sequence $\gamma_j\propto j^{-s}$ with varying $s$,
which, unlike the previously mentioned sequence, does not depend on
the parameters of the distribution. For the particular distribution of
the observed $p$-values we choose $m=0.5$ and $n=5$. The following
plots compare the achieved power and FDR of SAFFRON, LORD and
alpha-investing, the first two with several different sequences
$\{\gamma_j\}$ obtained by varying the parameter $s$. In particular,
\figref{SAFFRON_beta} and \figref{LORD_beta} show the changes in
performance of SAFFRON and LORD respectively with increasing $s$;
i.e., increasing aggressiveness of the sequence
$\{\gamma_j\}$. \figref{comparisonbeta} compares the performance of
SAFFRON, LORD and alpha-investing, where the first two use the highest
performing sequence chosen among six considered sequences, as in the
setting with Gaussian tests. Although the simulation results show SAFFRON performing similarly to LORD and alpha-investing for small fractions of non-null hypotheses, it significantly outperforms its competitors in terms of power and using up available wealth with a higher number of $p$-values coming from the alternative.

\begin{figure}[H]
\centerline{\includegraphics[width=0.35\textwidth]{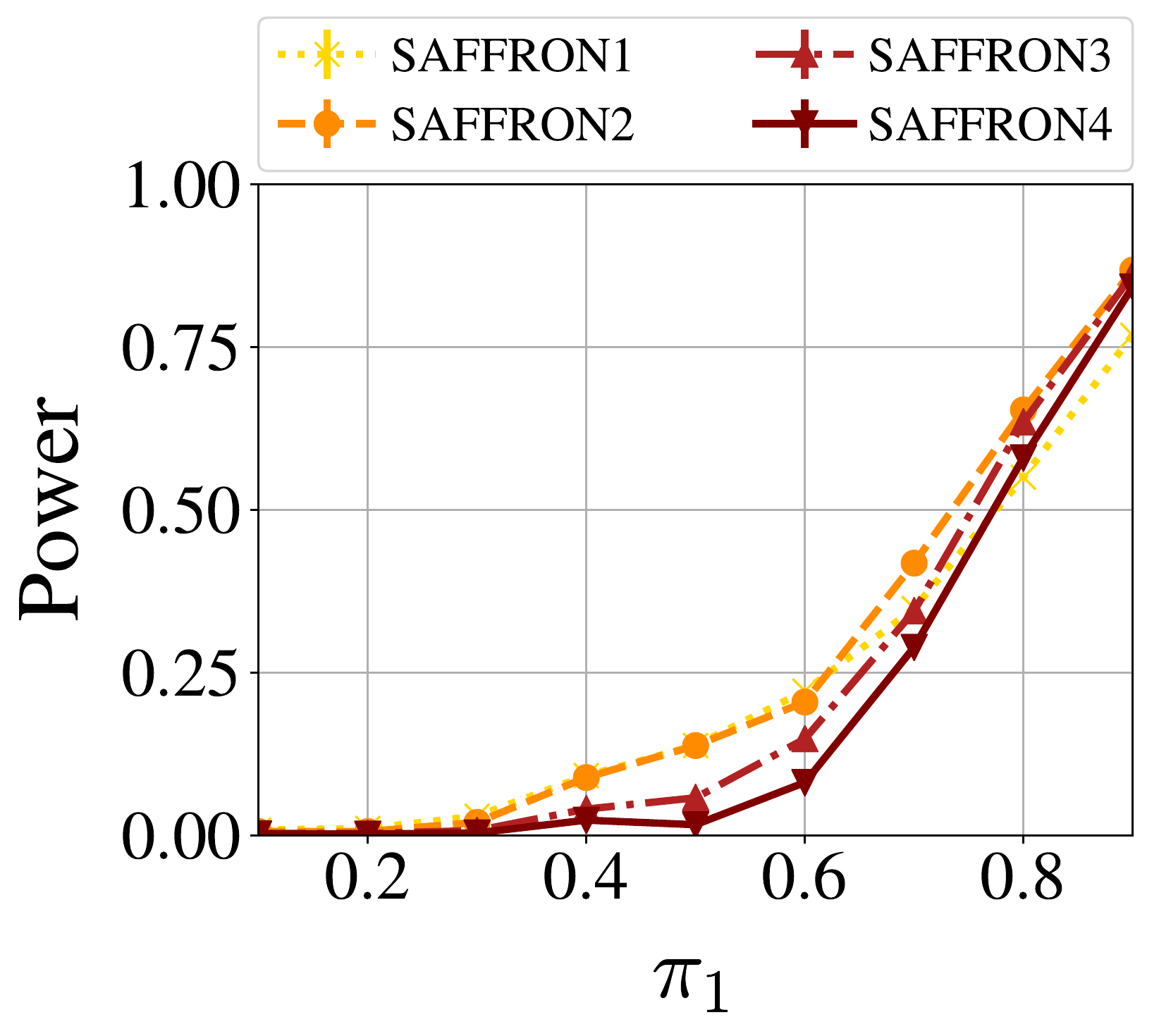}
\includegraphics[width=0.35\textwidth]{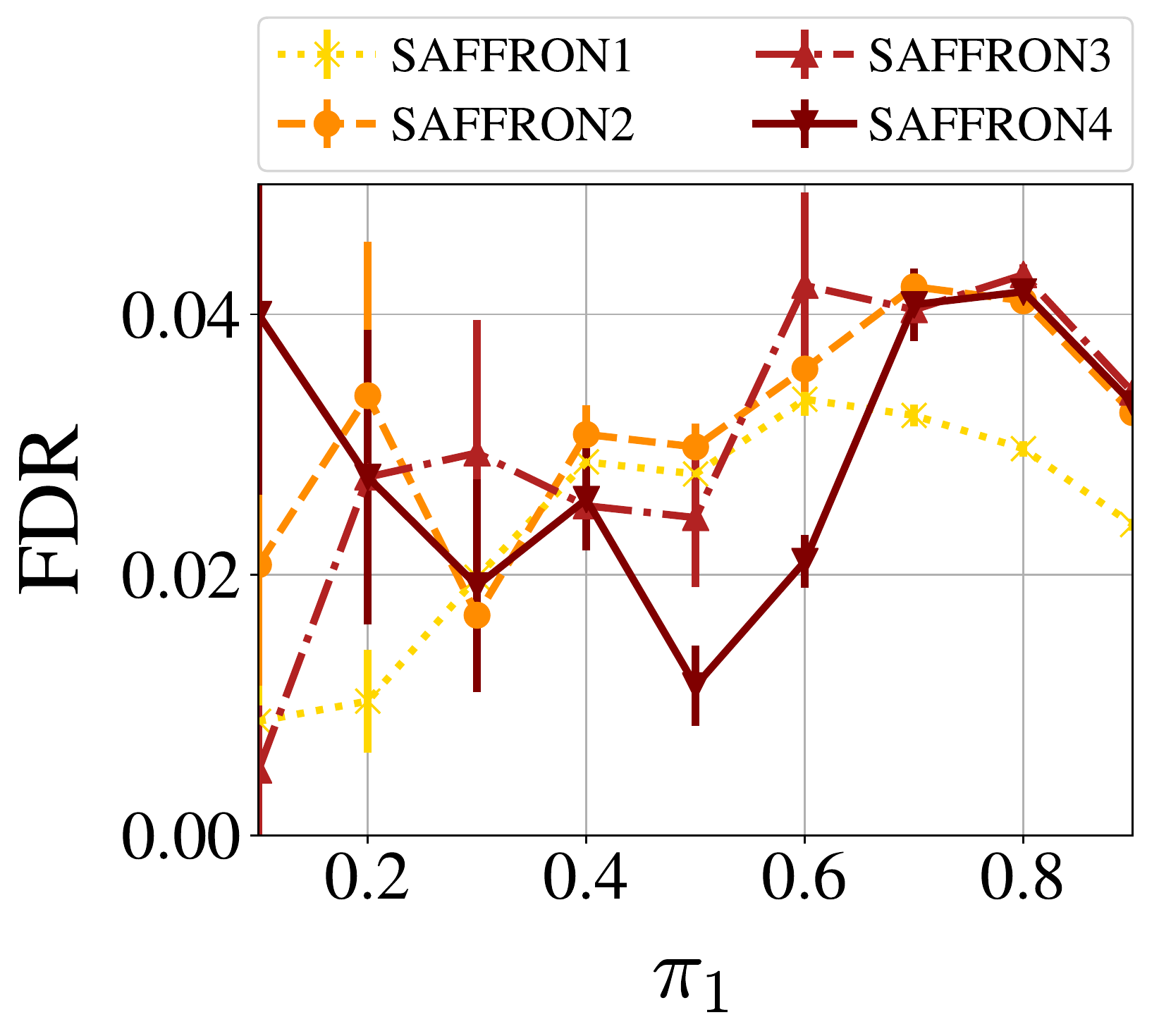}}
\caption{Statistical power and FDR versus fraction of non-nulls $\pi_1$ for SAFFRON (at target level $\alpha = 0.05$) using four different sequences $\{\gamma_j\}$ of increasing aggressiveness. Non-null $p$-values are distributed as $\text{Beta}(0.5,5)$. (See also \figref{comparisonbeta}.)}
\label{fig:SAFFRON_beta}
\end{figure}

\begin{figure}[H]
\centerline{\includegraphics[width=0.35\textwidth]{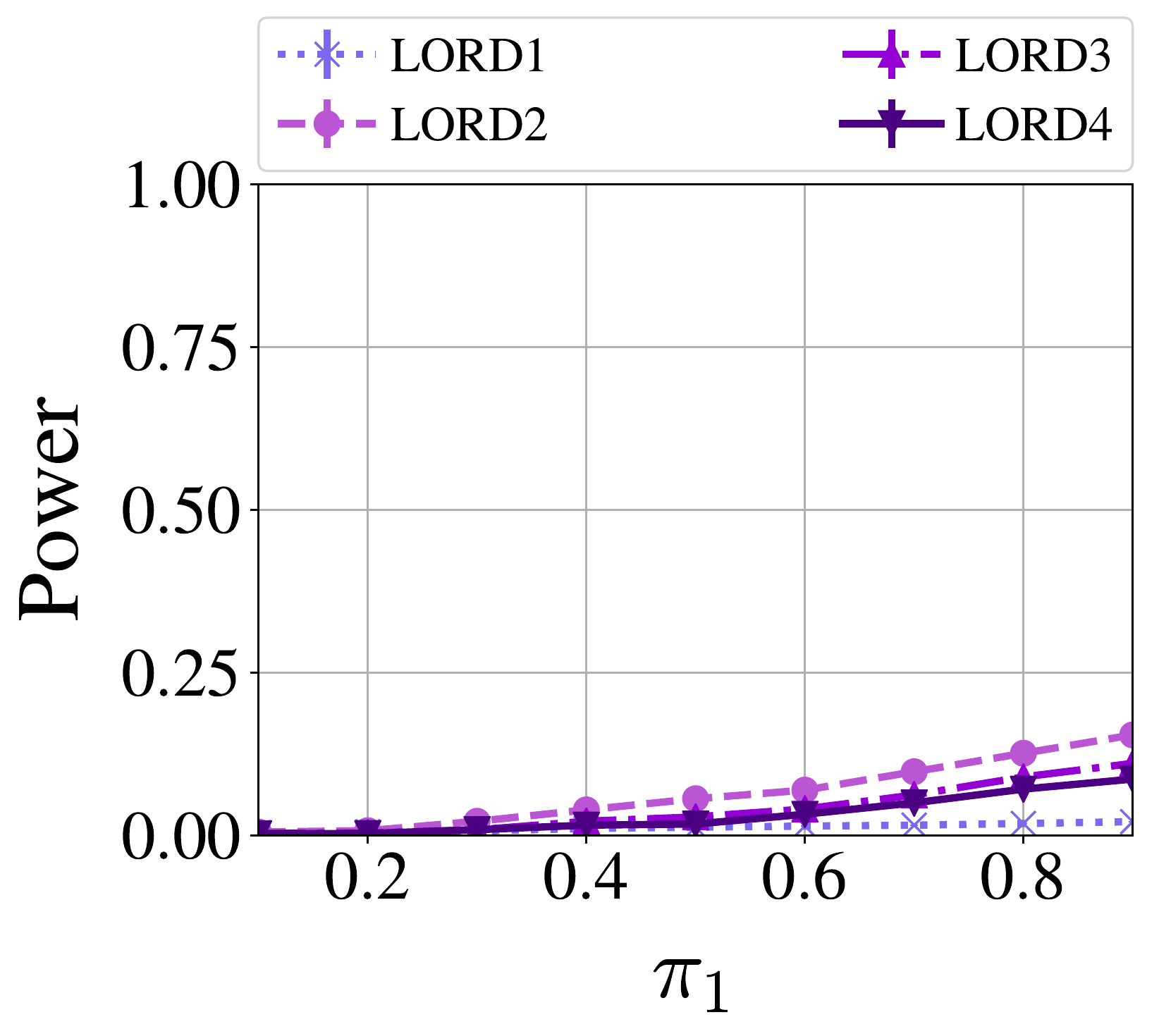}
\includegraphics[width=0.35\textwidth]{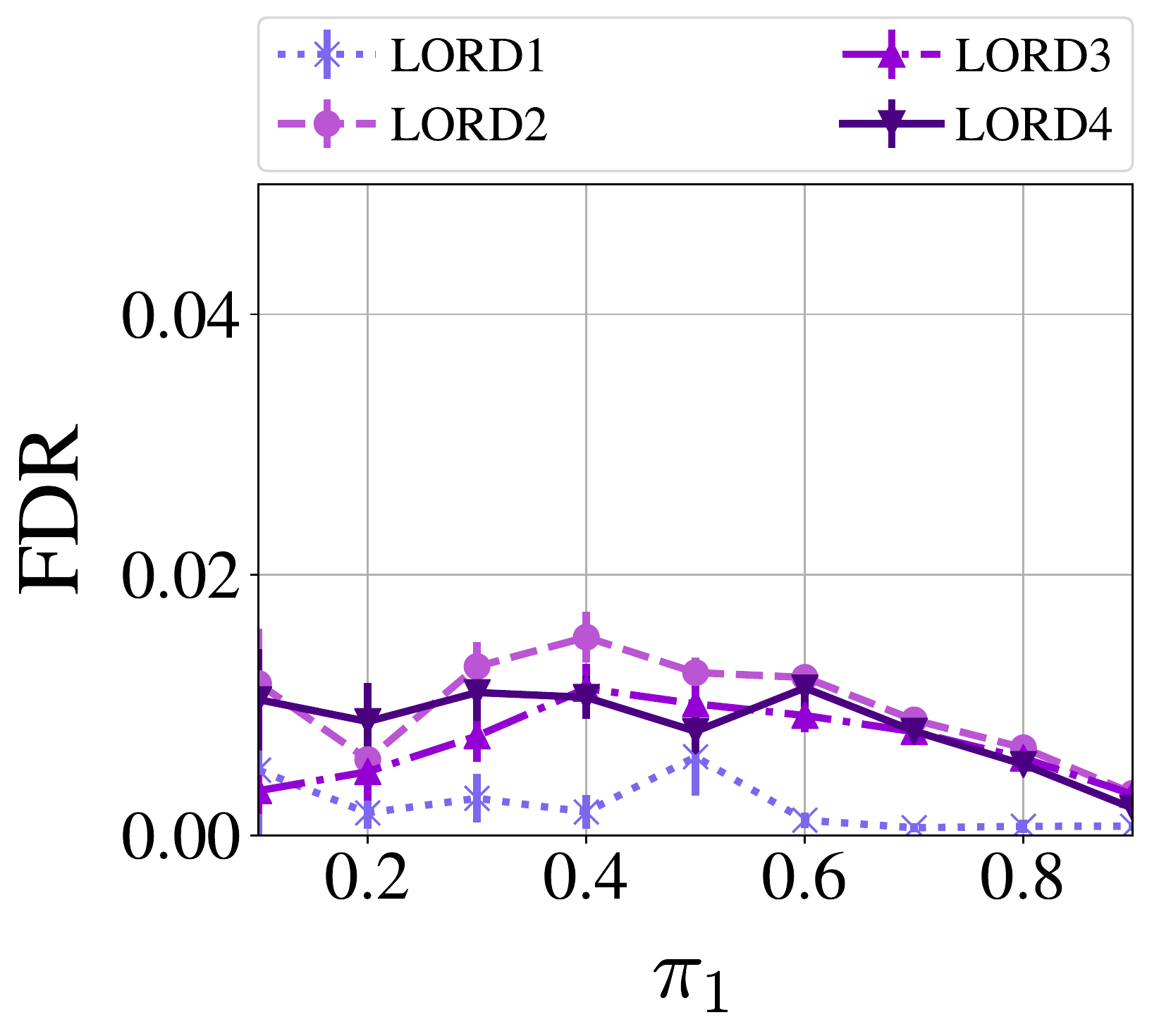}}
\caption{Statistical power and FDR versus fraction of non-null
  hypotheses $\pi_1$ for LORD (at target level $\alpha = 0.05$) using
  four different sequences $\{\gamma_j\}$ of increasing
  aggressiveness. The LORD1 method uses the sequence proposed in
  the paper~\cite{javanmard2016online}. Under the alternative the $p$-values are distributed as $\text{Beta}(0.5,5)$. (See also \figref{comparisonbeta}.)}
\label{fig:LORD_beta}
\end{figure}

\begin{figure}[H]
\centerline{\includegraphics[width=0.35\textwidth]{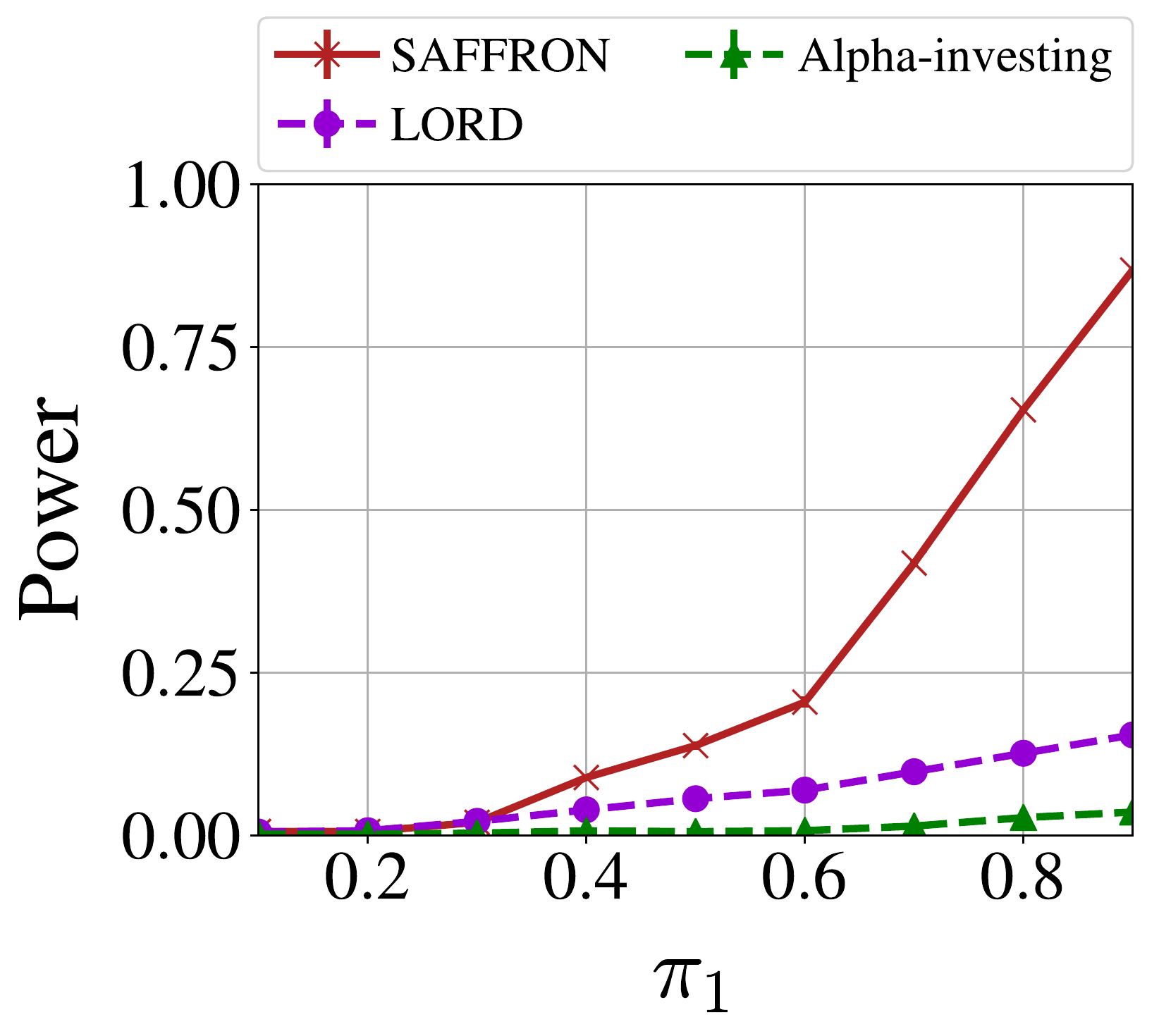}
\includegraphics[width=0.35\textwidth]{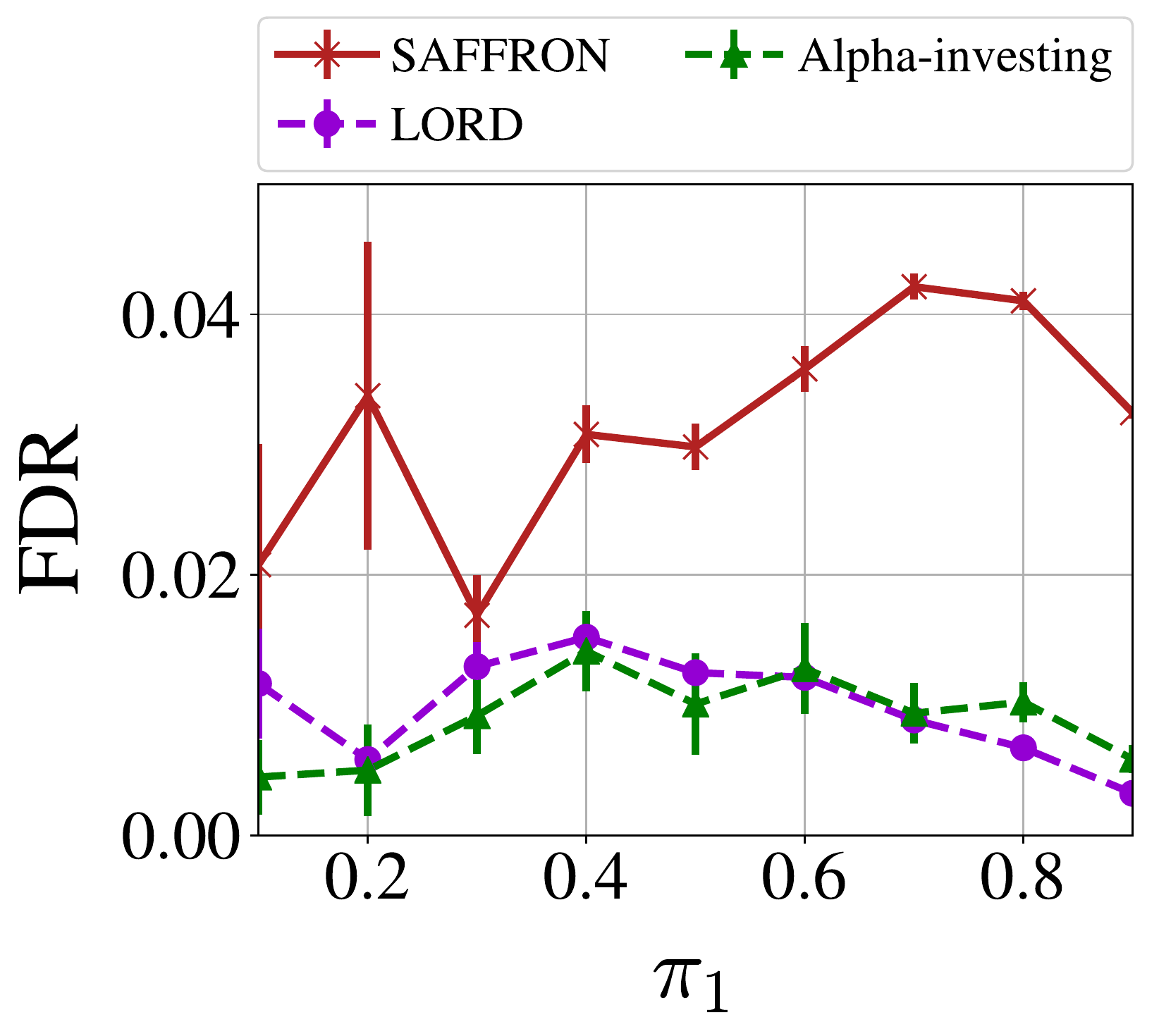}}
\caption{Statistical power and FDR versus fraction of non-null
  hypotheses $\pi_1$ for SAFFRON, LORD and alpha-investing (at target
  level $\alpha = 0.05$), using the sequence
  $\{\gamma_j\}$ which achieves the highest power for each of them
  (chosen over six sequences of varying aggressiveness). Under the
  alternative the $p$-values are distributed as $\text{Beta}(0.5,5)$.}
\label{fig:comparisonbeta}
\end{figure}

\subsection{SAFFRON with constant $\lambda$ vs SAFFRON alpha-investing}

Here we provide a comparison of SAFFRON for constant $\lambda$, set to the default value $1/2$, and the SAFFRON version of alpha-investing, obtained by setting $\lambda_j = \alpha_j$. We adopt the two settings described earlier in this section, namely testing with Gaussian observations and beta alternatives. As earlier, we vary the mean of Gaussian observations in a similar way.

In \figref{comparisongauss2} and \figref{comparisongauss3} we provide comparison of power and FDR for the two variants of SAFFRON, where $p$-values are computed from Gaussian observations; we take the non-null distribution to be $F_1 =
N(2,1)$ and $F_1 = N(3,1)$, respectively. From left to right, we increase the aggressiveness of the discount sequence $\{\gamma_t\}$. Our simulations indicate that there is no clear winner; which variant of SAFFRON performs better depends on the choice of $\{\gamma_t\}$, strength of the non-null signals, as well as the proportion of non-nulls in the sequence.

\figref{comparisonbeta} gives the same comparison, however for beta alternatives, as described in the previous subsection. The results indicate that $\lambda = 1/2$ seems to be a preferred choice in this case - across all choices of the sequence $\{\gamma_t\}$, as well as all non-null proportions, it uniformly outperforms the alpha-investing version.

\begin{figure}[H]
\centering
\begin{subfigure}[b]{\linewidth}
\includegraphics[width=0.3\textwidth]{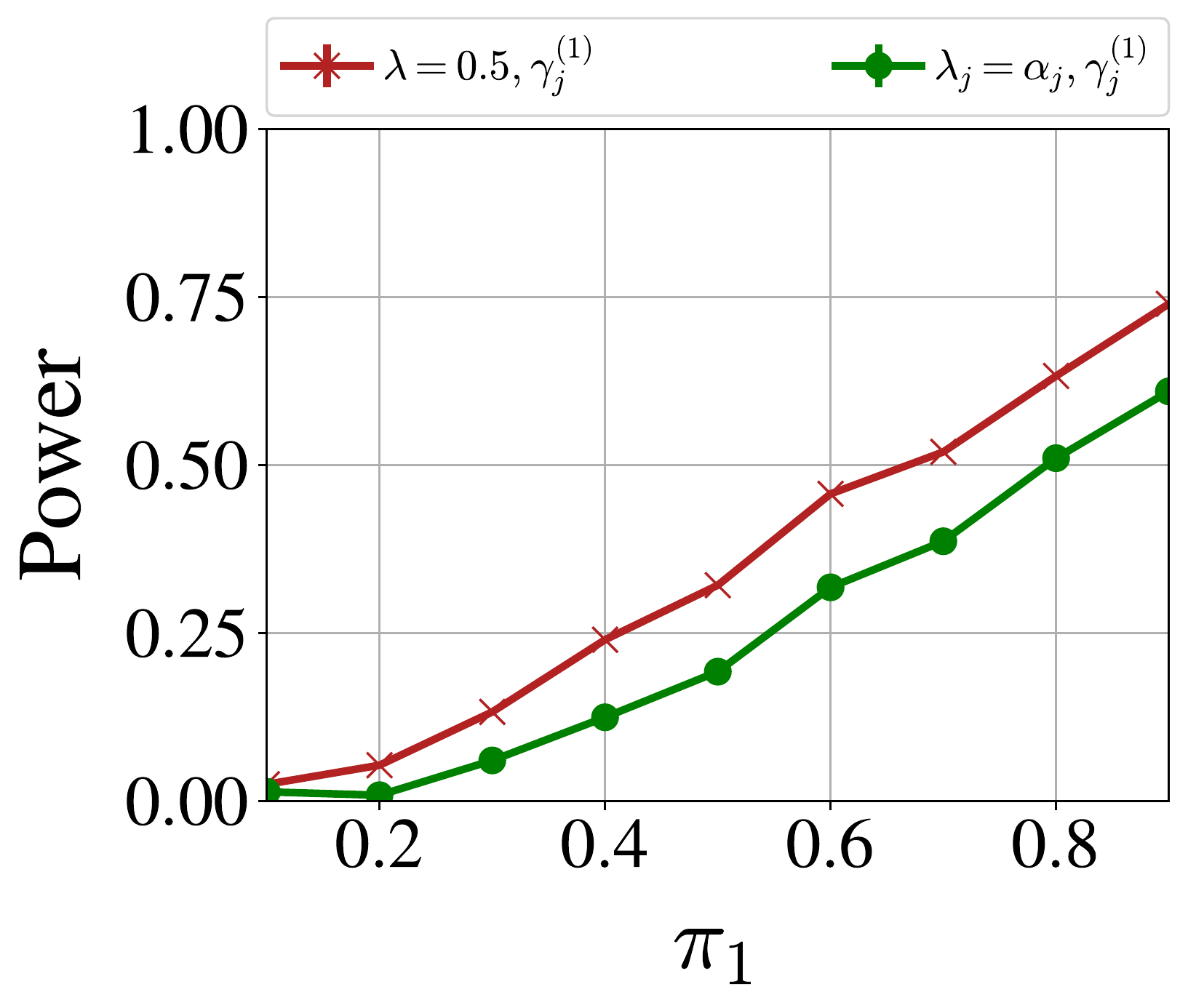}
\includegraphics[width=0.3\textwidth]{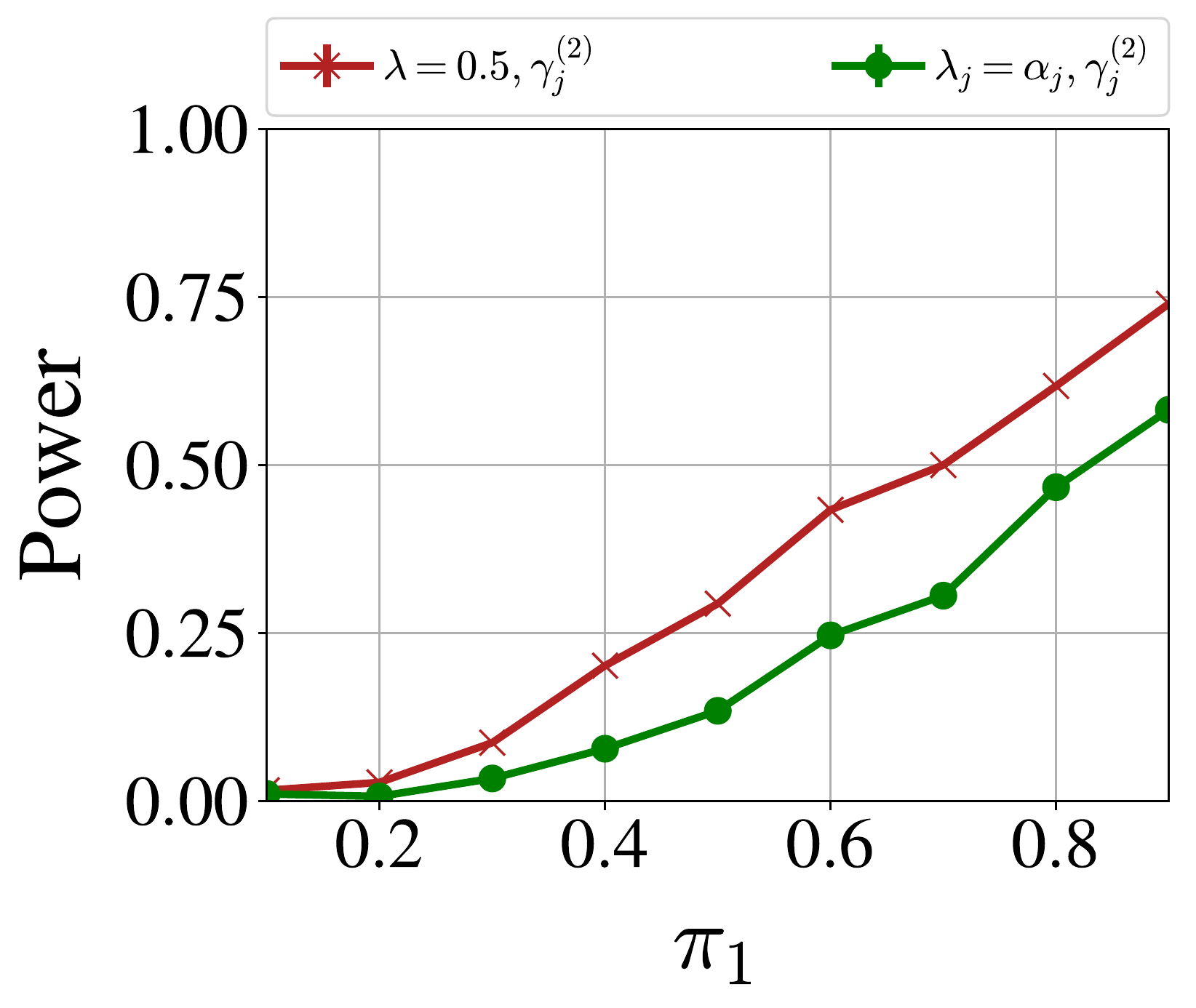}
\includegraphics[width=0.3\textwidth]{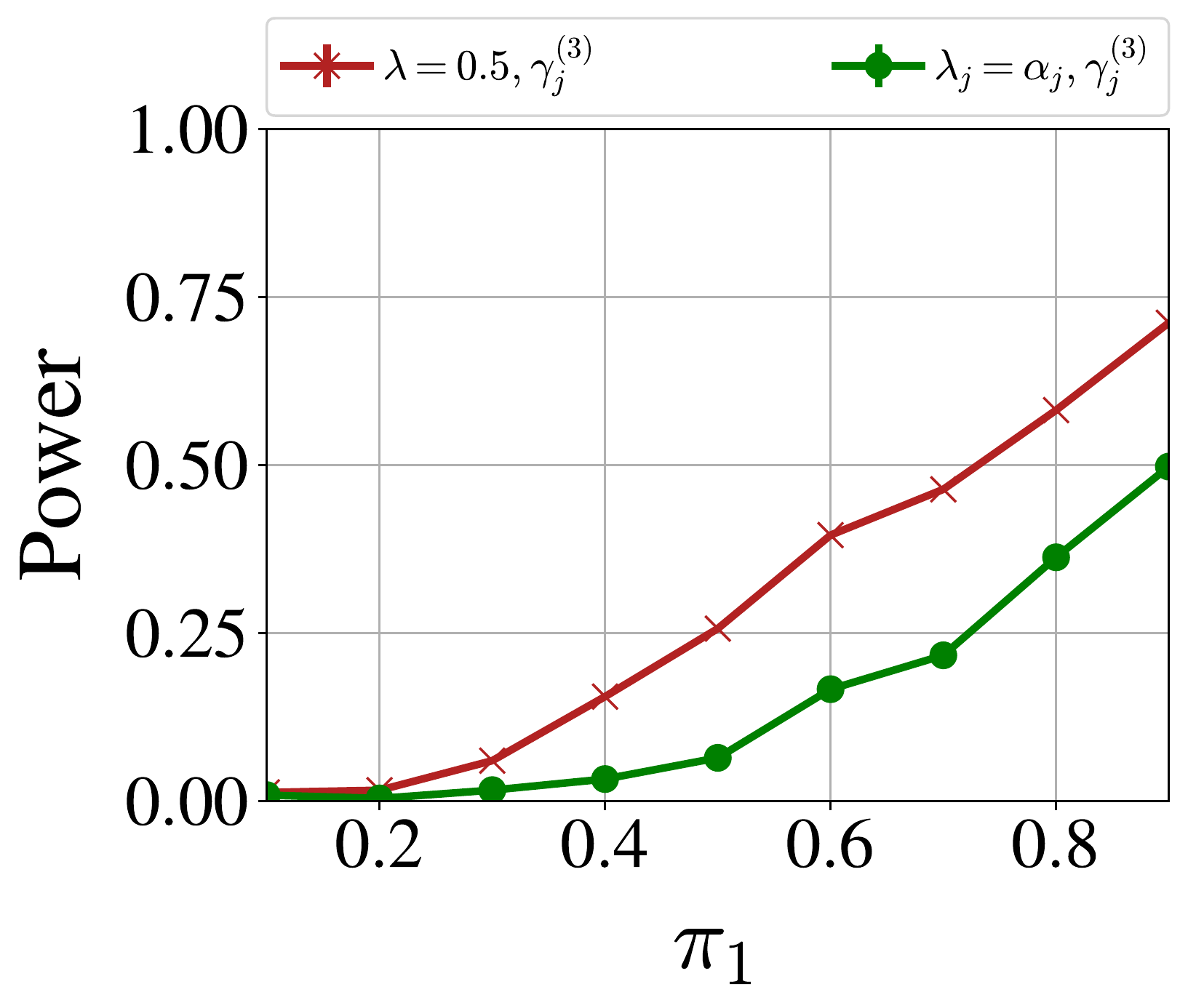}
\end{subfigure}
\begin{subfigure}[b]{\linewidth}
\includegraphics[width=0.3\textwidth]{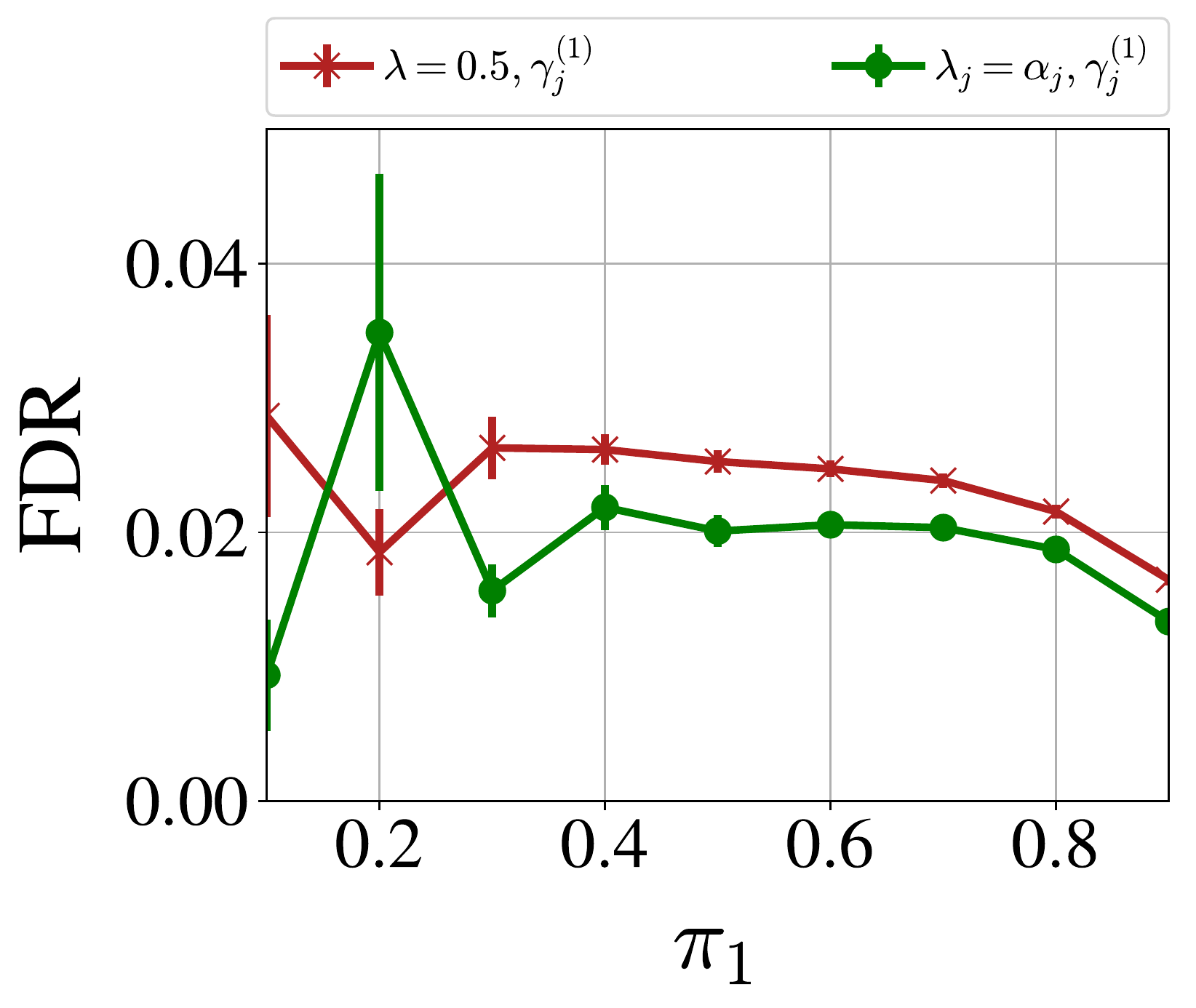}
\includegraphics[width=0.3\textwidth]{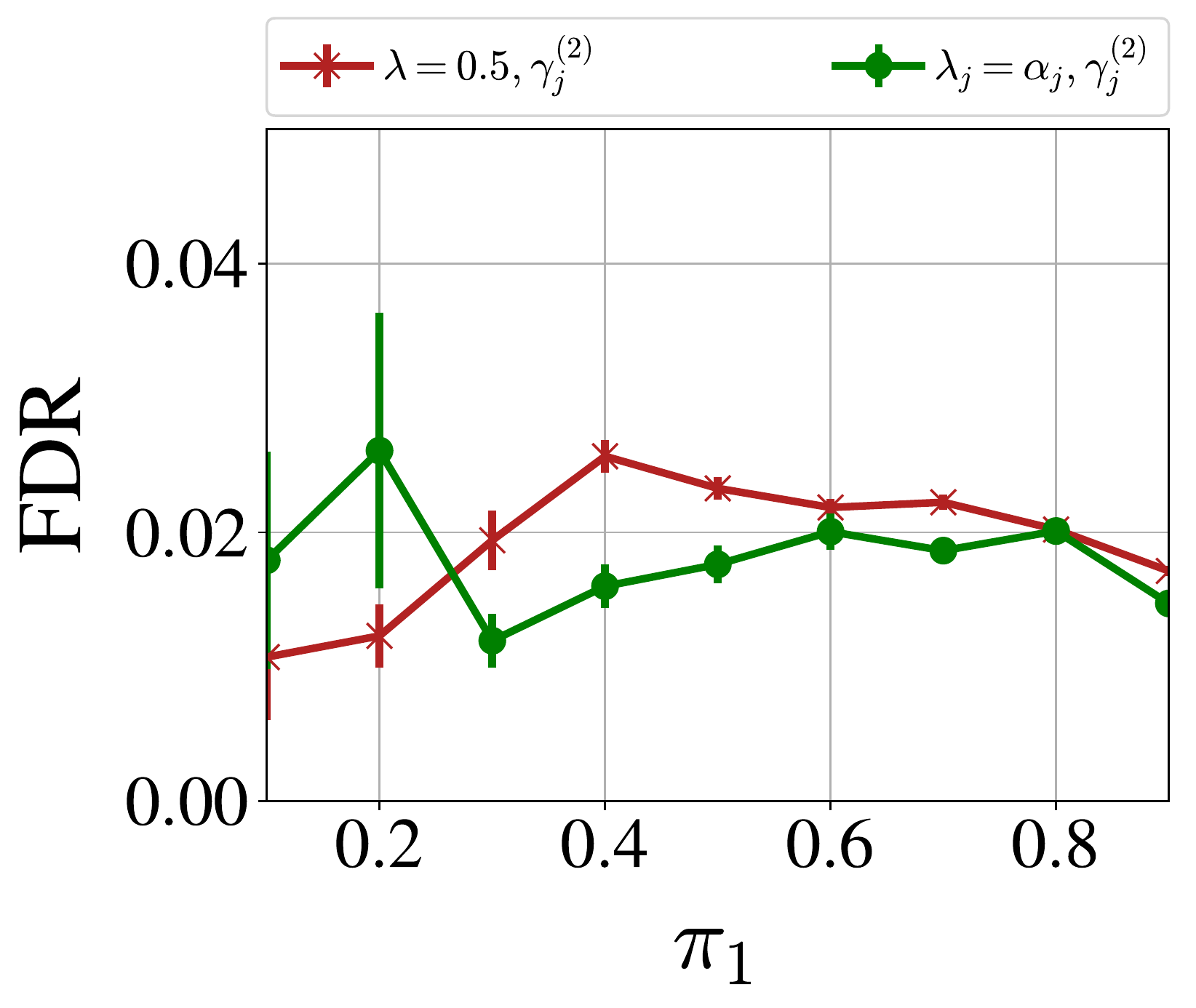}
\includegraphics[width=0.3\textwidth]{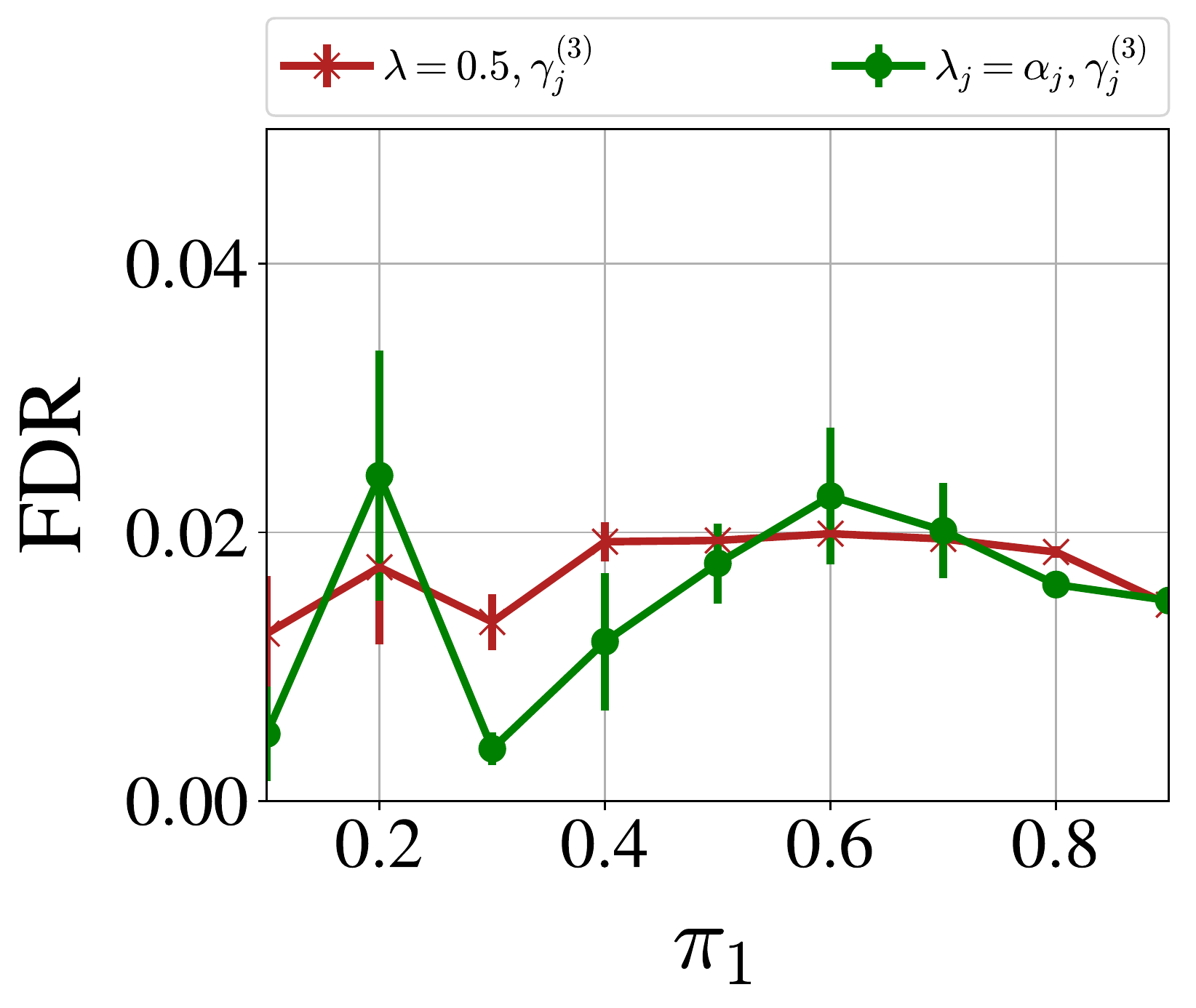}
\end{subfigure}
\caption{Statistical power and FDR versus fraction of non-null
  hypotheses $\pi_1$ for SAFFRON with $\lambda = 1/2$ and SAFFRON with $\lambda_j = \alpha_j$ (at target level $\alpha = 0.05$)
  using three different sequences $\{\gamma_j\}$ of increasing
  aggressiveness. The observations under the alternative are $N(\mu_i,1)$
  with $\mu_i\sim N(2,1)$, and are converted
  into one-sided $p$-values as $P_i=\Phi(-Z_i)$.}
\label{fig:comparisongauss2}
\end{figure}

\begin{figure}[H]
\centering
\begin{subfigure}[b]{\linewidth}
\includegraphics[width=0.3\textwidth]{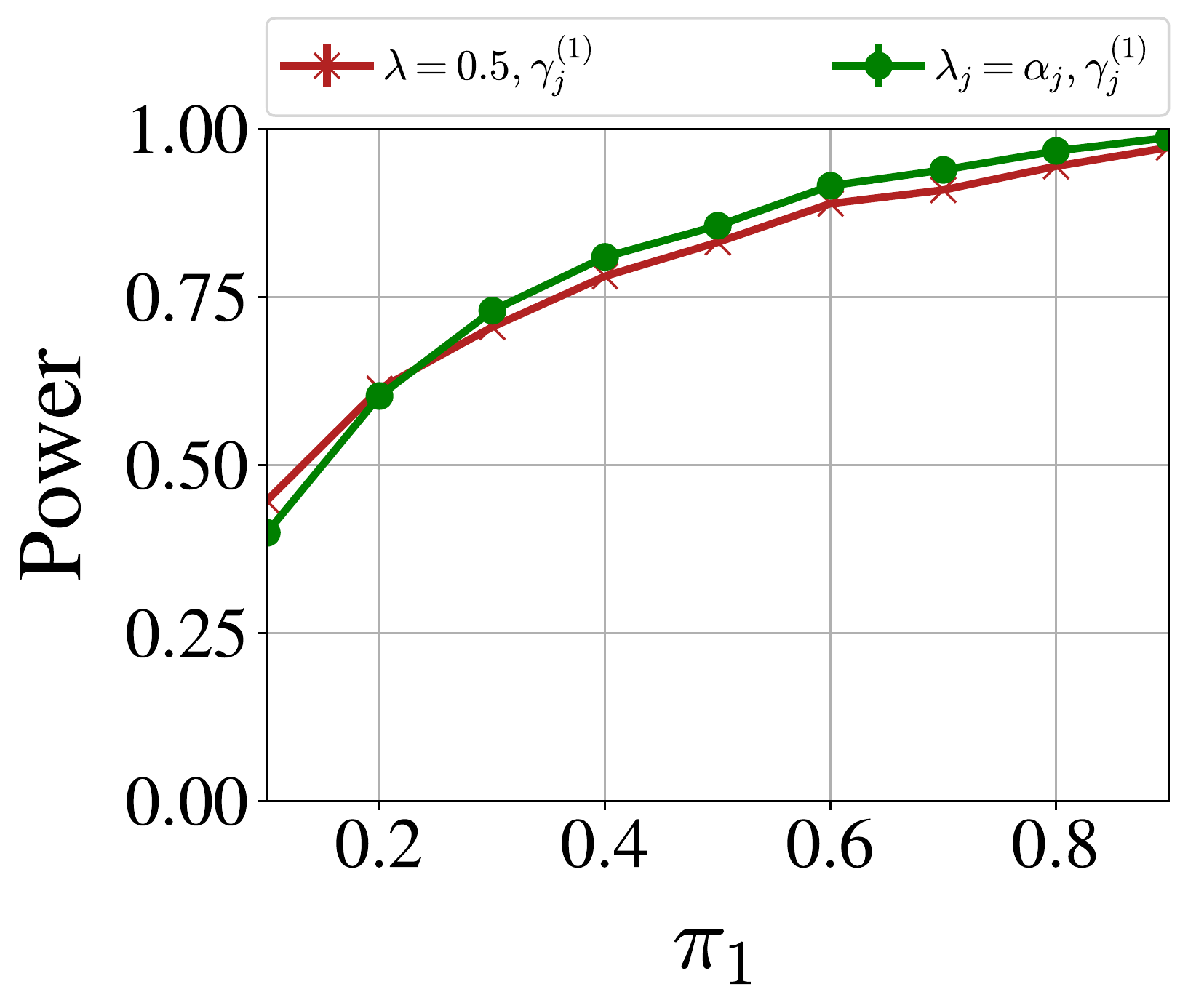}
\includegraphics[width=0.3\textwidth]{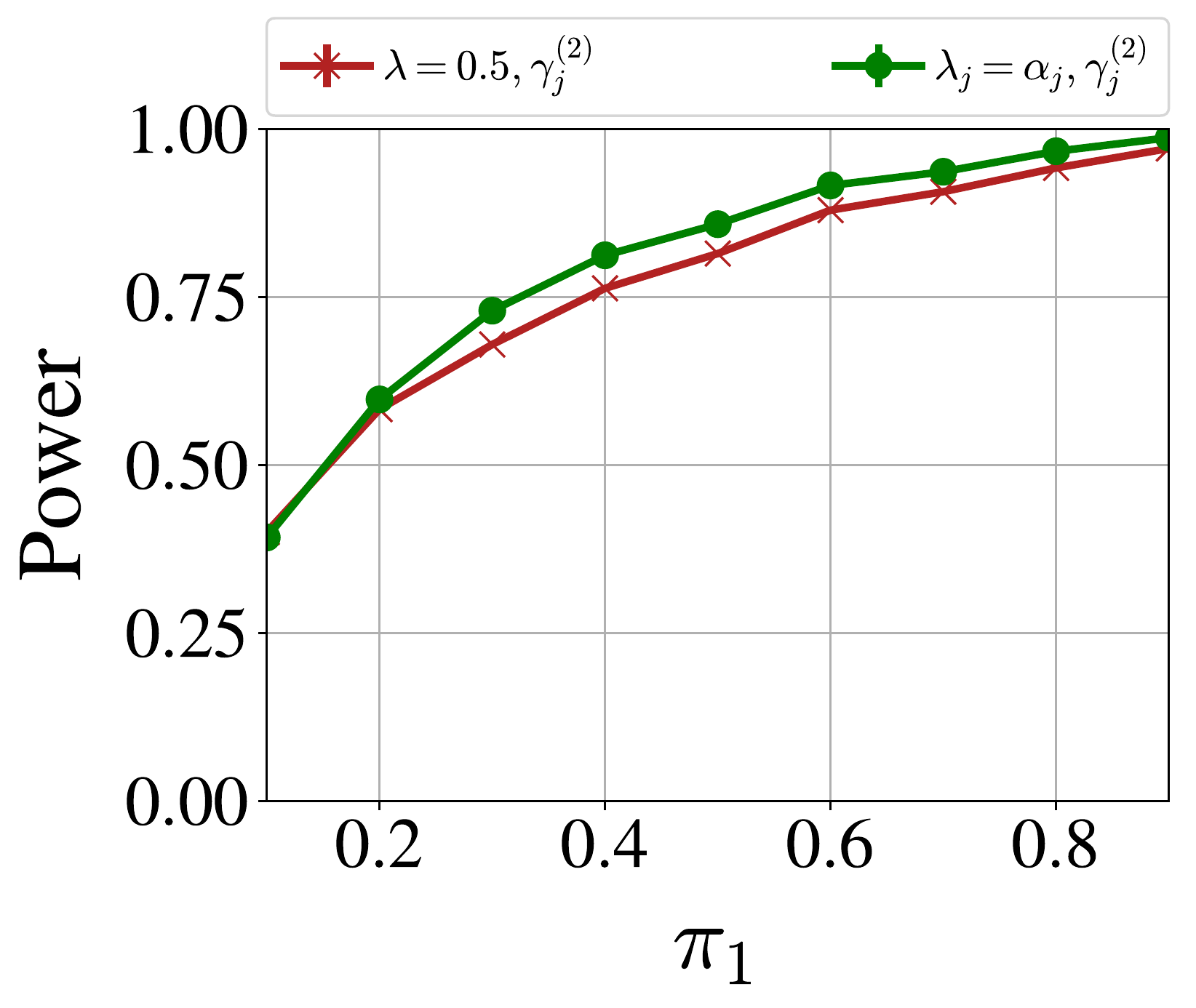}
\includegraphics[width=0.3\textwidth]{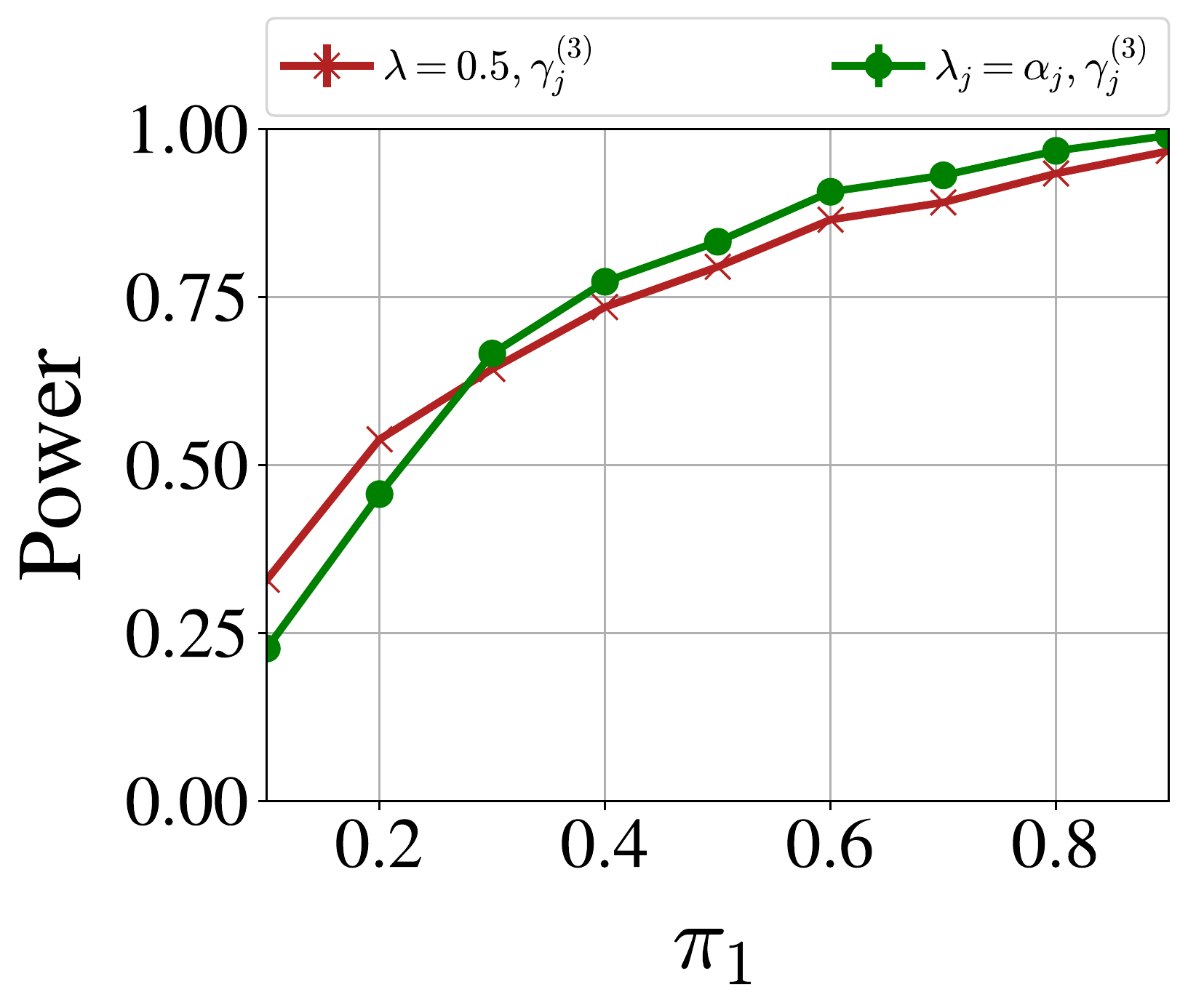}
\end{subfigure}
\begin{subfigure}[b]{\linewidth}
\includegraphics[width=0.3\textwidth]{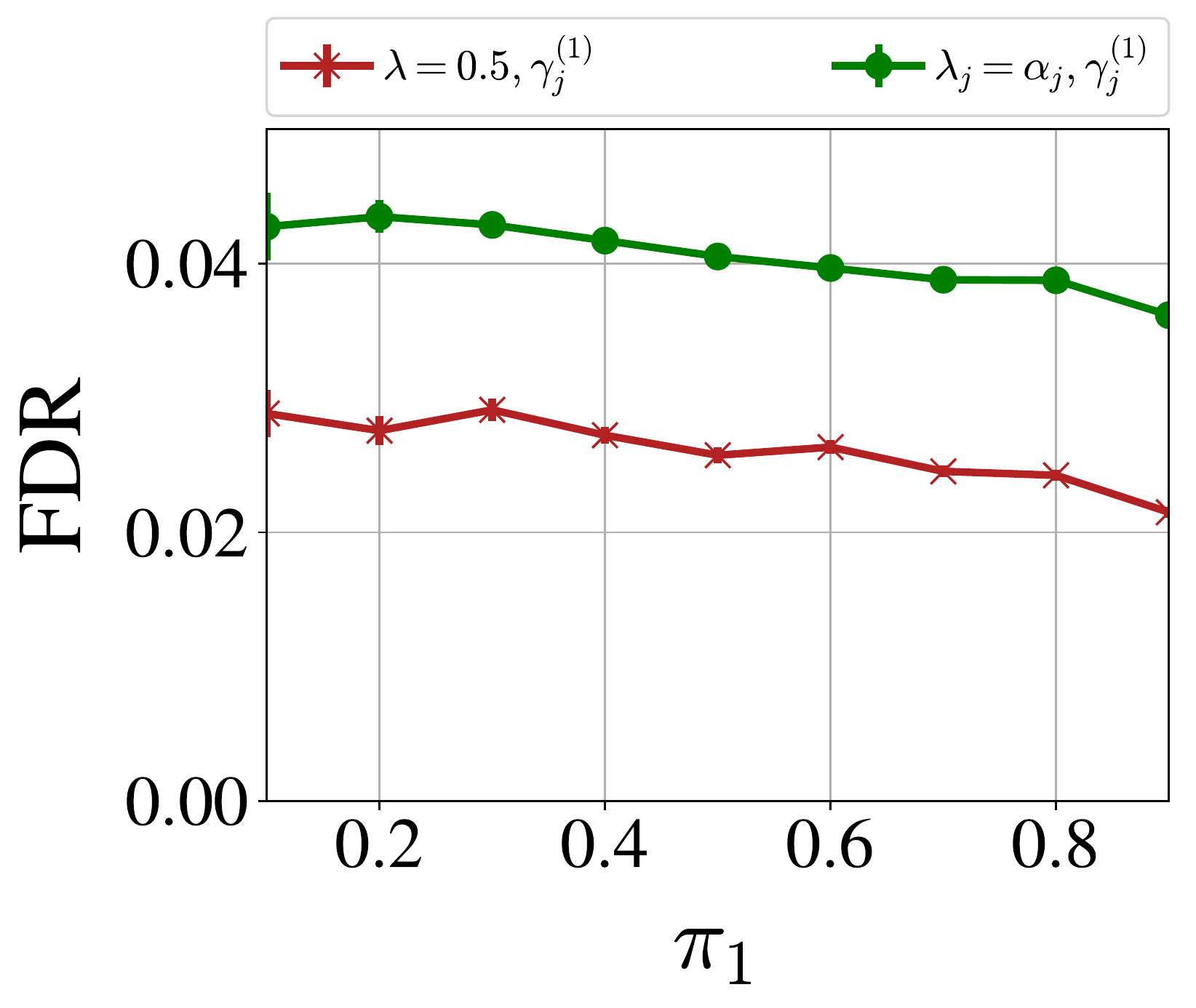}
\includegraphics[width=0.3\textwidth]{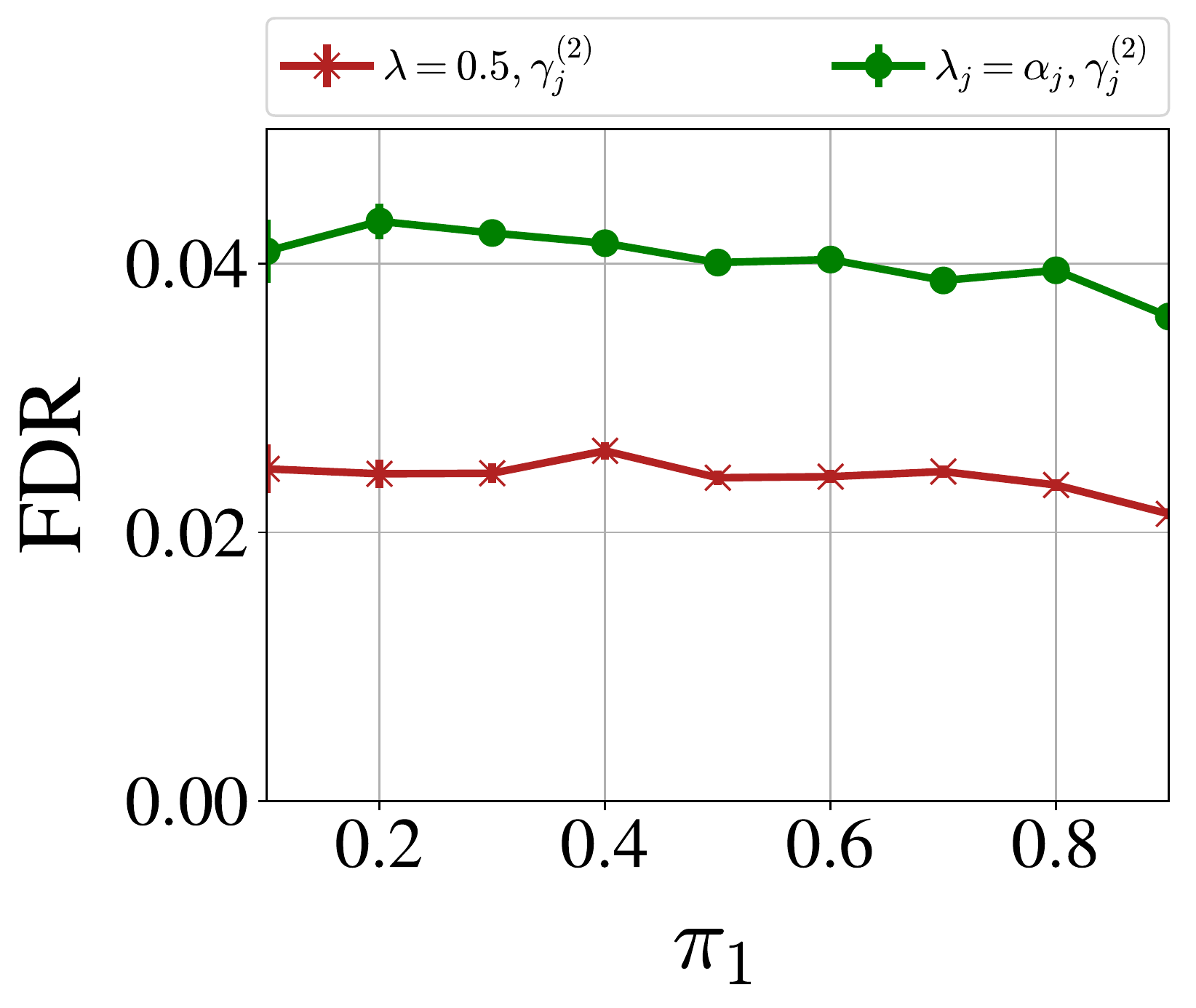}
\includegraphics[width=0.3\textwidth]{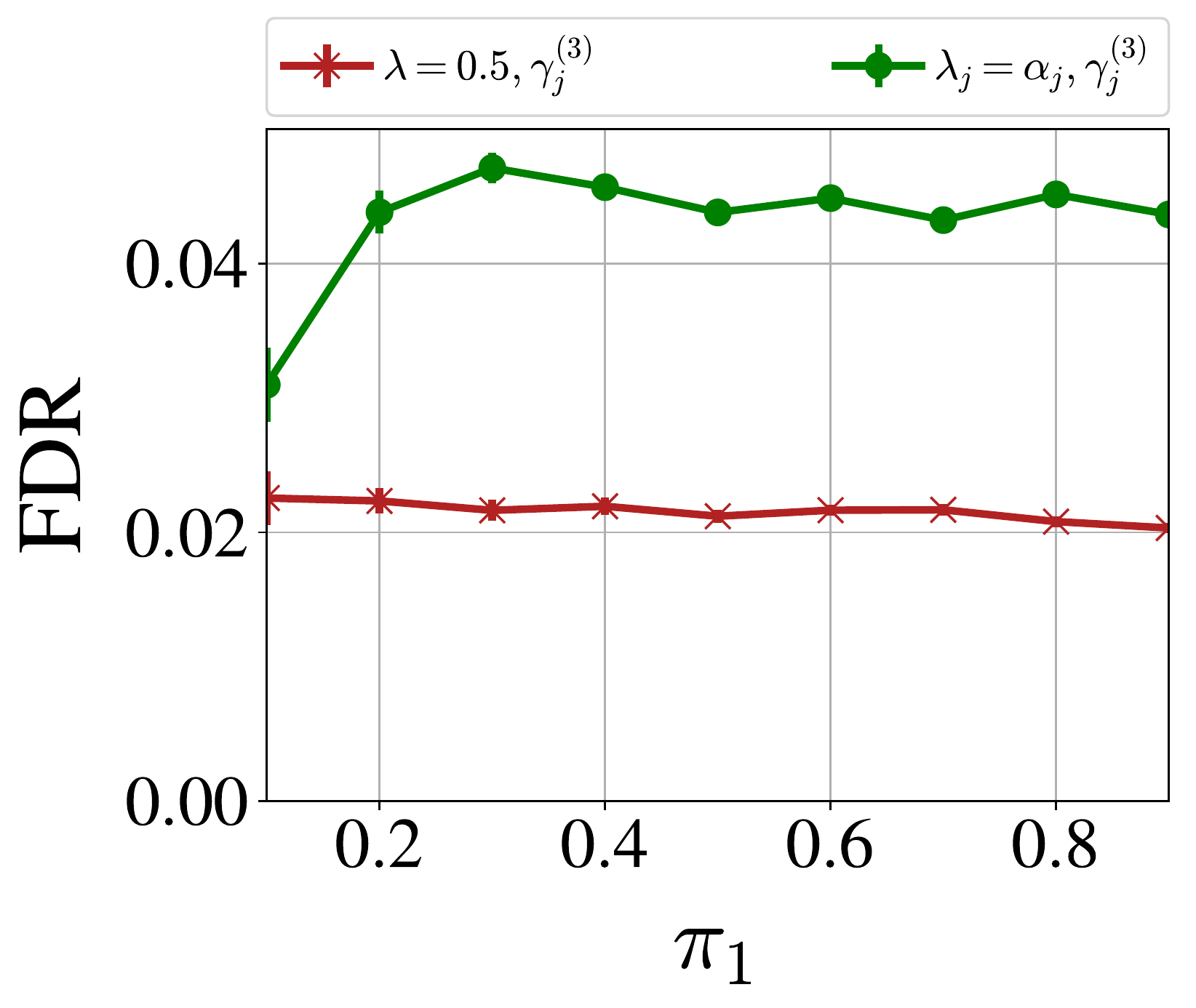}
\end{subfigure}
\caption{Statistical power and FDR versus fraction of non-null
  hypotheses $\pi_1$ for SAFFRON with $\lambda = 1/2$ and SAFFRON with $\lambda_j = \alpha_j$ (at target level $\alpha = 0.05$)
  using three different sequences $\{\gamma_j\}$ of increasing
  aggressiveness. The observations under the alternative are $N(\mu_i,1)$
  with $\mu_i\sim N(3,1)$, and are converted
  into one-sided $p$-values as $P_i=\Phi(-Z_i)$.}
\label{fig:comparisongauss3}
\end{figure}

\begin{figure}[H]
\centering
\begin{subfigure}[b]{\linewidth}
\includegraphics[width=0.3\textwidth]{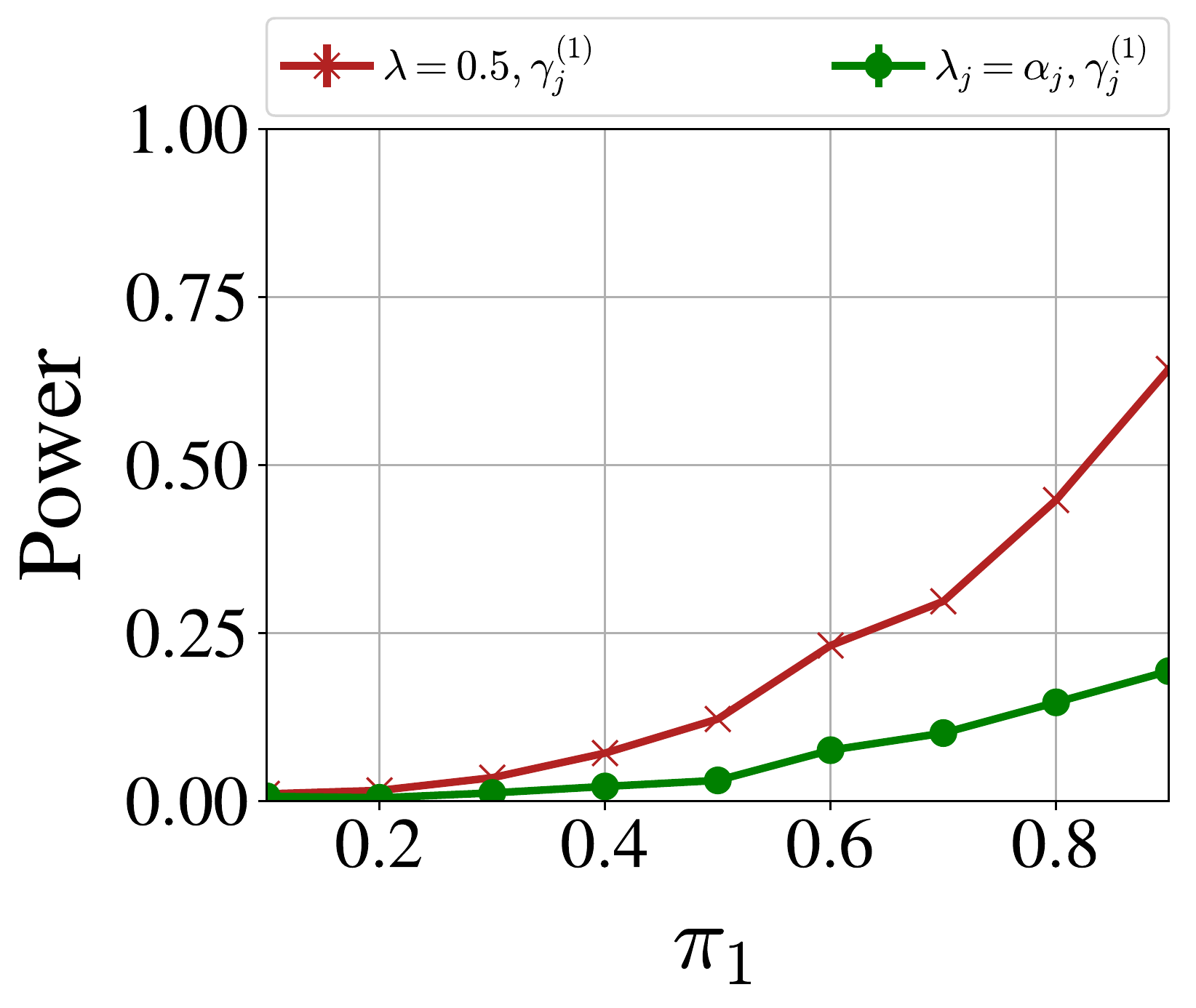}
\includegraphics[width=0.3\textwidth]{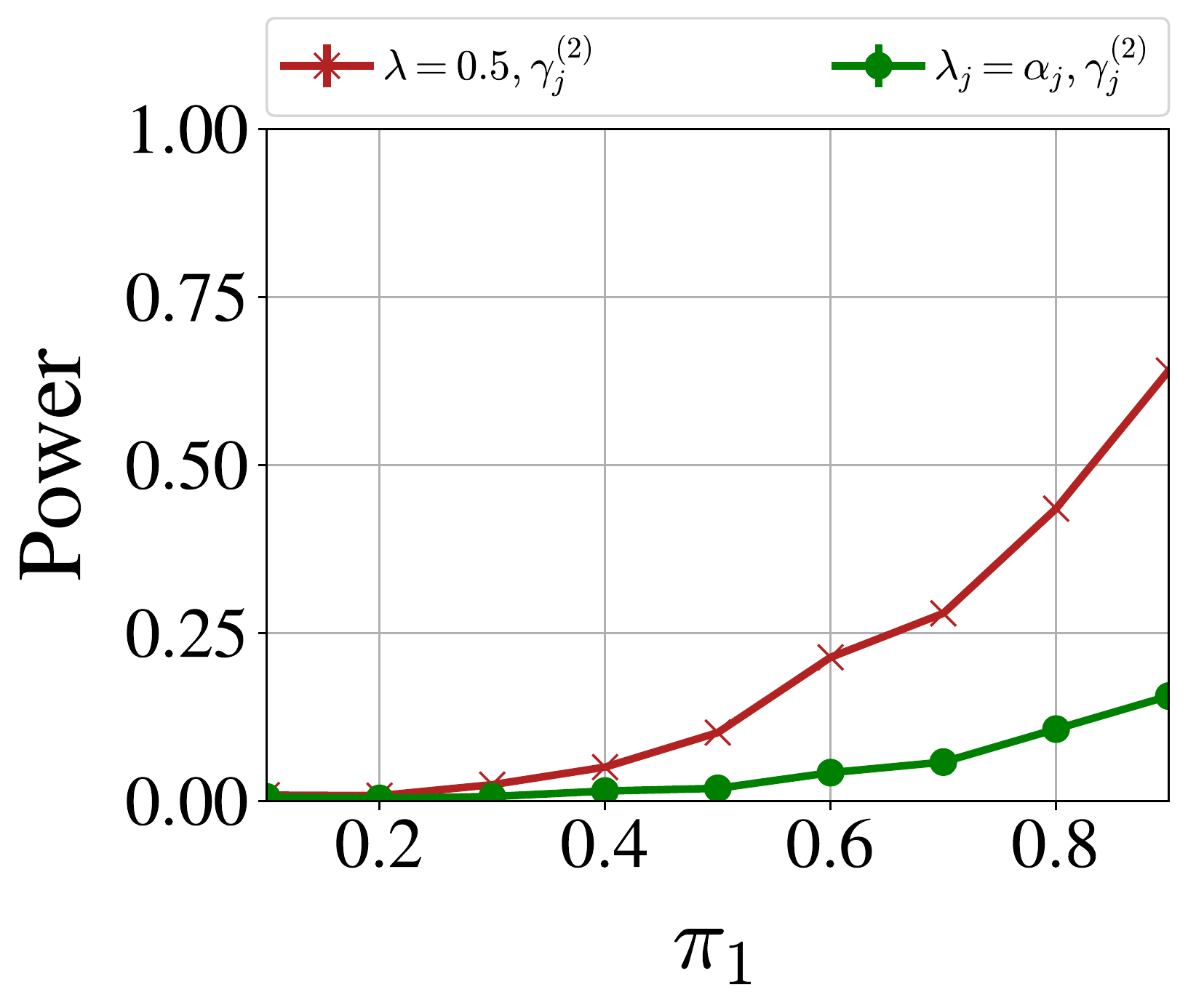}
\includegraphics[width=0.3\textwidth]{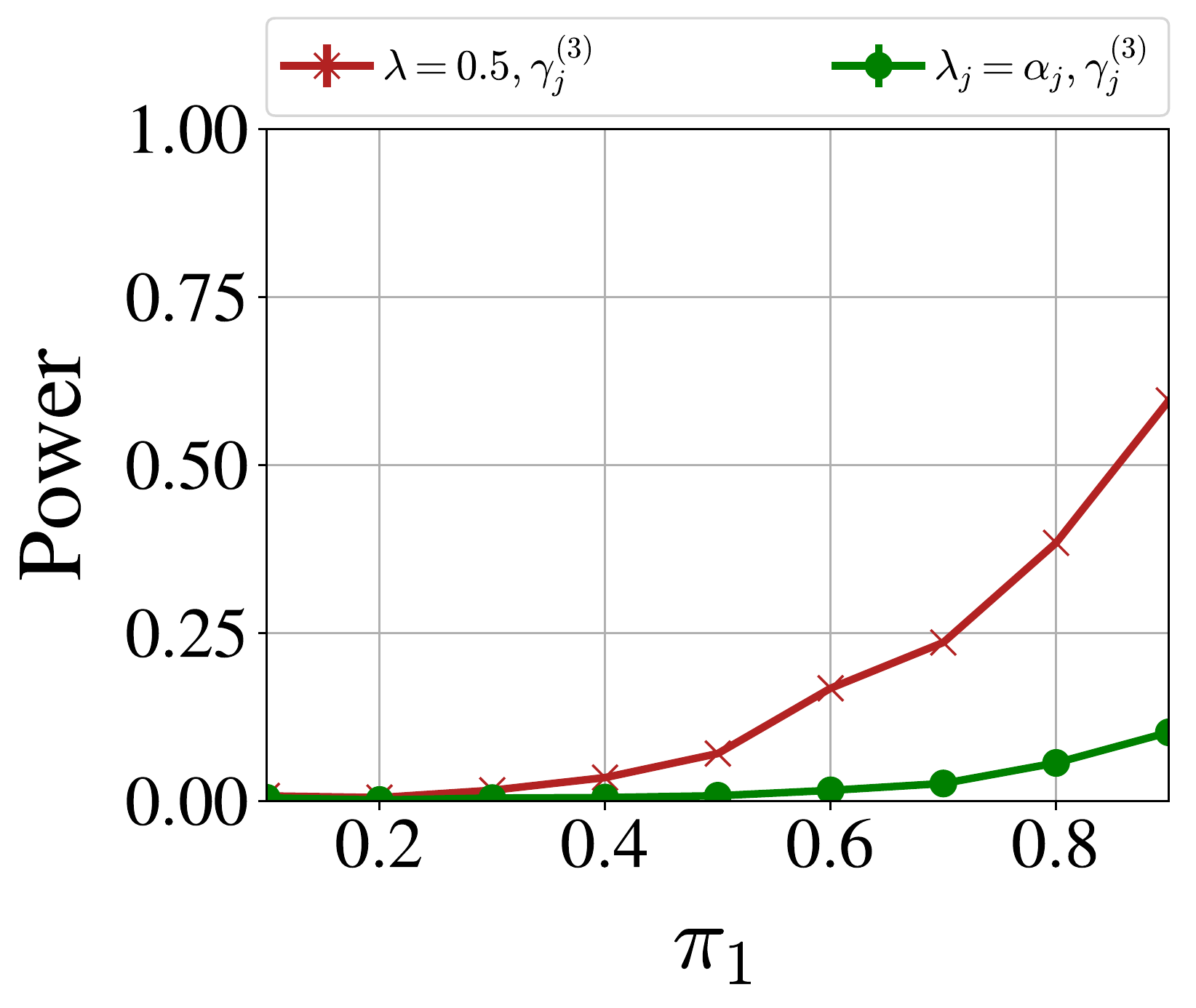}
\end{subfigure}
\begin{subfigure}[b]{\linewidth}
\includegraphics[width=0.3\textwidth]{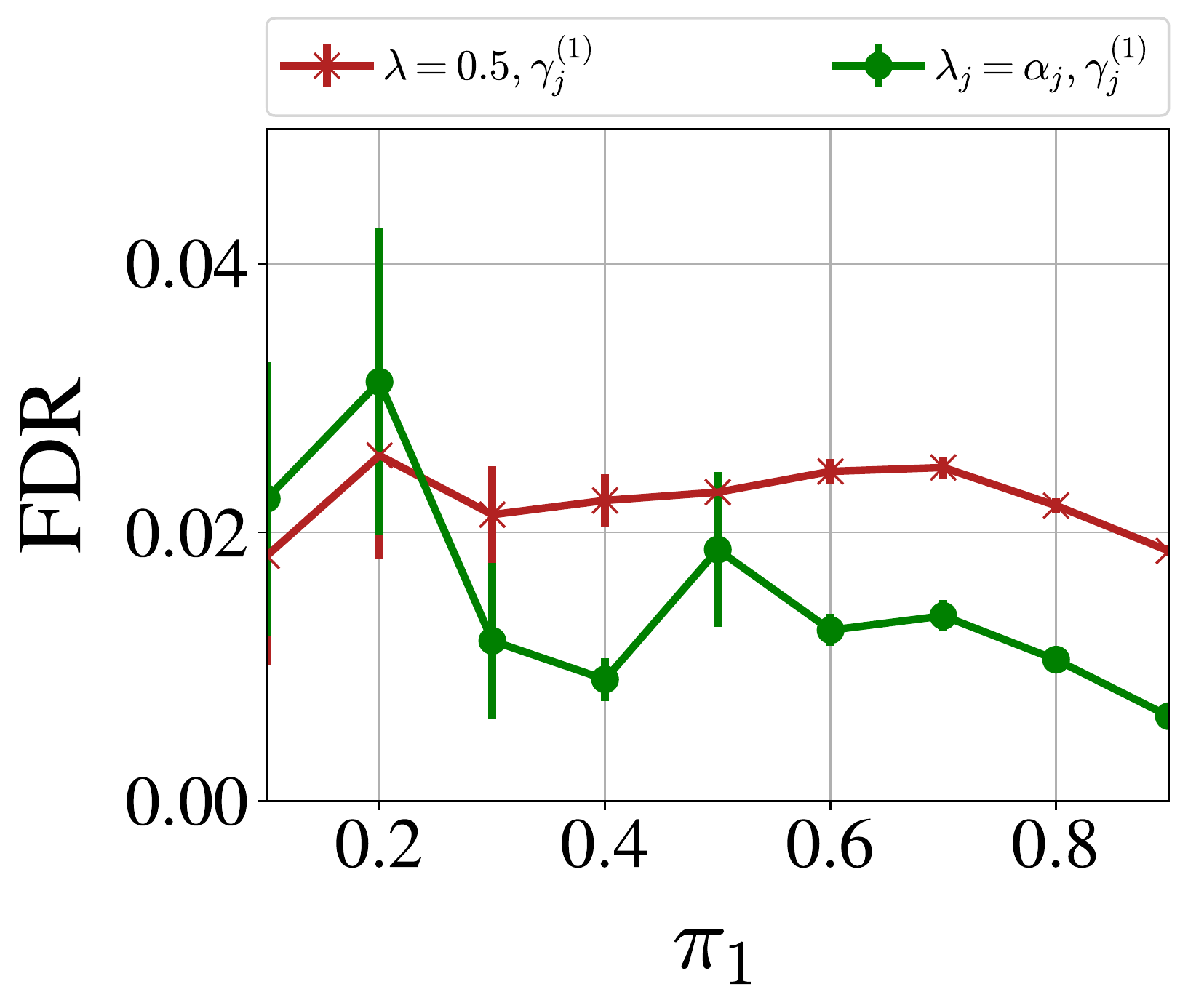}
\includegraphics[width=0.3\textwidth]{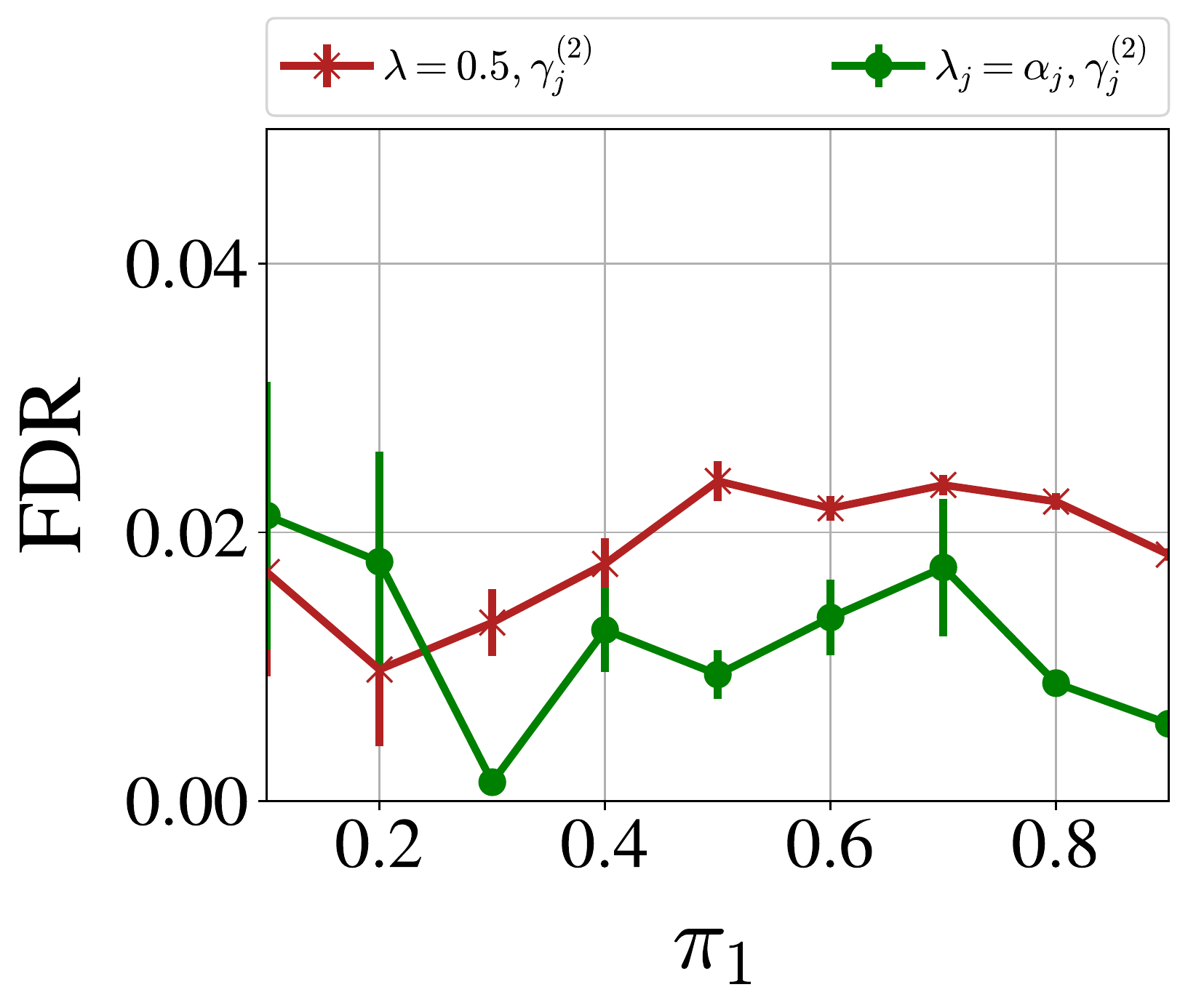}
\includegraphics[width=0.3\textwidth]{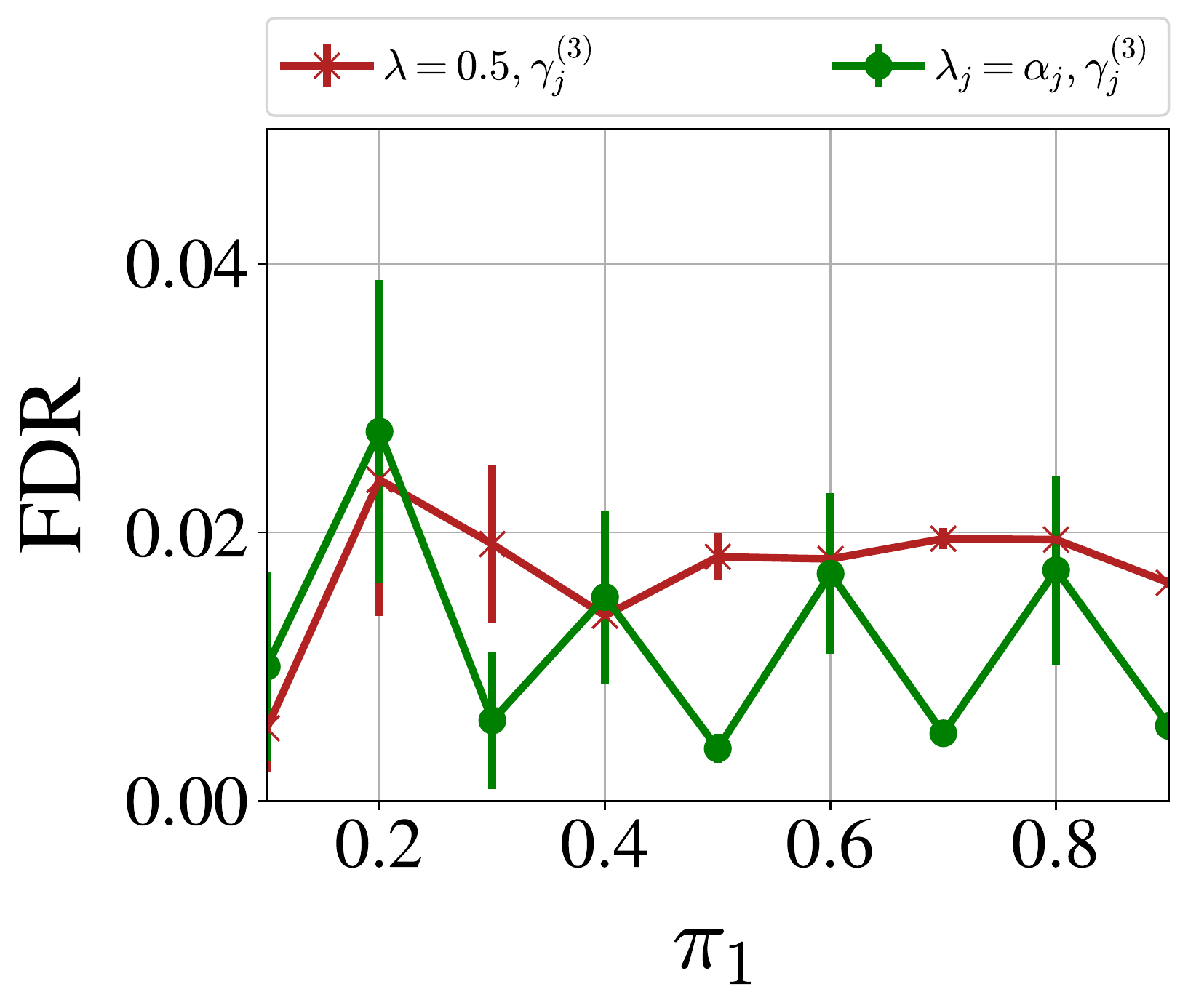}
\end{subfigure}
\caption{Statistical power and FDR versus fraction of non-null
  hypotheses $\pi_1$ for SAFFRON with $\lambda = 1/2$ and SAFFRON with $\lambda_j = \alpha_j$ (at target level $\alpha = 0.05$)
  using three different sequences $\{\gamma_j\}$ of increasing
  aggressiveness. The $p$-values under the alternative are distributed as $\text{Beta}(0.5,5)$.}
\label{fig:comparisonbeta}
\end{figure}


\section{Proofs}
\label{sec:proofs}

Here we provide the proofs of \propref{fdp-oracle} and \thmref{fdp-saff} using the reverse super-uniformity lemma proved in \secref{lemma}, as well as the proof of SAFFRON's monotonicity.

\subsection{Proof of \propref{fdp-oracle}}
\label{sec:proof-prop}

Statement (a) is proved by noting that for any time $t \in \N$, we have:
\begin{align*}
\EE{|\nulls \cap \cR(t)|} &= \sum_{j \leq t, j \in \nulls}
\EE{\One{P_j \leq \alpha_j}} \; \leq \; \sum_{j \leq t, j \in \nulls}
\EE{\alpha_j},
\end{align*}
where the inequality follows after taking iterated expectations by
conditioning on $\F^{j-1}$, and then applying the conditional
super-uniformity property~\eqref{eqn:superuniformity-cond1}. If we have $\fdp^*(t) \defn \frac{1}{|\cR(t)|} \sum\limits_{j\leq t, j \in \nulls} \alpha_j \leq \alpha$, as assumed in statement (b), then it follows that:
\begin{align*}
\sum_{j \leq t, j \in \nulls} \EE{\alpha_j}&=\EE{\sum_{j \leq t, j \in
    \nulls} \alpha_j}\\ &\leq \alpha \EE{|\cR(t)|},
\end{align*}
using linearity of expectation and the assumption on
$\fdp^*(t)$. Using part (a) and rearranging yields the inequality
$\mfdr(t) \defn \frac{\EE{|\nulls \cap \cR(t)|}}{\EE{|\cR(t)|}} \leq
\alpha$, which concludes the proof of part (b).

If, in addition, the null $p$-values are independent of each other and
of the non-nulls and the sequence $\{\alpha_t\}$ is monotone, we can
use the following argument to prove claims (c) and (d).  These claims
establish that the procedure controls the FDR at any time $t \in
\N$. Still assuming the inequality $\fdp^*(t) \leq \alpha$, we have:
\begin{align*}
\fdr(t) & = \EE{\frac{|\nulls \cap \cR(t)|}{|\cR(t)|}} \\
& = \sum_{j \leq t, j \in \nulls} \EE{\dotfrac{\One{P_j \leq
      \alpha_j}}{|\cR(t)|}} \\
& \leq \sum_{j \leq t, j \in \nulls} \EE{\dotfrac{\alpha_j}{|\cR(t)|}}
\\
& = \EE{\fdp^*(t)} \\
& \leq \alpha,
\end{align*}
where the first inequality follows after taking iterated expectations
by conditioning on $\F^{j-1}$, and then applying the super-uniformity
lemma \cite{RYWJ17}, the following equality uses linearity of
expectation, and the final inequality follows by the assumption on
$\fdp^*(t)$. This concludes the proof of both statements (c) and (d).

\subsection{A reverse super-uniformity lemma}
\label{sec:lemma}

The proof of Theorem \ref{thm:fdp-saff} critically depends on the following technical, yet interpretable, lemma.
To provide some context, the reader is encouraged to recall the definition~\eqnref{superuniformity-cond2} of conditional
super-uniformity, as well as its equivalent rephrased
form~\eqnref{superuniformity-cond3}. Lemma~\ref{lem:power} guarantees that for independent
$p$-values, statement~\eqnref{superuniformity-cond3} holds true more generally.

\begin{lemma}
\label{lem:power}
Assume that the $p$-values $P_1,P_2,\dots$ are independent and let $g: \{0,1\}^T \to \R$ be any coordinate-wise non-decreasing function. Further, assume that $\alpha_t = f_t(R_{1:t-1}, C_{1:t-1})$ and $\lambda_t = g_t(R_{1:t-1}, C_{1:t-1})$, for some coordinate-wise non-decreasing functions $f_t$ and $g_t$. Then, for any index $t \leq T$ such that $H_t \in \nulls$, we have:
\begin{align*}
\EEst{ \frac{\alpha_t \One{P_t >
      \lambda_t}}{(1-\lambda_t)g(R_{1:T})}}{\F^{t-1}} &\geq~
\EEst{\frac{\alpha_t}{g(R_{1:T})}}{\F^{t-1}} \\ &\geq~
\EEst{ \frac{ \One{P_t \leq
      \alpha_t}}{g(R_{1:T})}}{\F^{t-1}}.
\end{align*}
\end{lemma}
\begin{proof}
The second inequality is a consequence of super-uniformity lemmas from
past work~\cite{RYWJ17,javanmard2016online}, so we only prove the
first inequality. At a high level, the proof strategy is inverted, and we will hallucinate a vector with one element being set to 1, instead of being set to 0 in the aforementioned works.

Letting $P_{1:T} = (P_1,\ldots,P_T)$ be the
original vector of $p$-values, we define a ``hallucinated'' vector of
$p$-values $\widetilde{P}^{t\to 1}_{1:T} \defn (\widetilde
P_1,\ldots,\widetilde P_T)$ that equals $P_{1:T}$, except that the
$t$-th component is set to one:
\begin{align*}
  \widetilde P_i = \begin{cases} 1 & \mbox{if $i = t$} \\
P_i & \mbox{if $i \neq t$.}
  \end{cases}
\end{align*}
Define hallucinated test levels $\{\widetilde \alpha_i\}$ and candidacy thresholds $\{\widetilde \lambda_i\}$ resulting from $(\widetilde
P_1,\ldots,\widetilde P_T)$, as well as hallucinated candidate and rejection indicators as $\widetilde
C_i = \One{\widetilde P_i \leq \widetilde \lambda_i}$ and $\widetilde R_i =
\One{\widetilde P_i \leq \widetilde \alpha_i}$ respectively.  Let $R_{1:T} = (R_1,\ldots,R_T)$ and
$\widetilde{R}^{t\to 1}_{1:T} = (\widetilde R_1,\ldots, \widetilde
R_T)$ denote the vector of rejections using $P_{1:T}$
and $\widetilde{P}^{t\to 1}_{1:T}$, respectively.  Similarly, let $C_{1:T} =
(C_1,\ldots,C_T)$ and $\widetilde{C}^{t\to 1}_{1:T} = (\widetilde
C_1,\ldots, \widetilde C_T)$ denote the vector of candidates using
$P_{1:T}$ and $\widetilde{P}^{t\to 1}_{1:T}$, respectively.

By construction, we have the following properties:
\begin{enumerate}
\item $\widetilde{R}_i = R_i$ and $\widetilde{C}_i = C_i$ for all $i <
  t$, hence $\alpha_i = \widetilde \alpha_i$ for all $i \leq t$.
 \item $\widetilde{R}_t = \widetilde C_t = 0$, and hence
   $\widetilde{R}_i \leq R_i$ for all $i \geq t$, due to monotonicity
   of the levels $\alpha_i$.
\end{enumerate}
Hence, on the event $\{P_t > \lambda_t\}$, we have $R_t = \widetilde
R_t = 0$ and $C_t = \widetilde C_t =0$, and hence also $R_{1:T} =
\widetilde{R}^{t\to 1}_{1:T}$. This allows us to conclude that:
\begin{align*}
\frac{\alpha_t \One{P_t > \lambda_t}}{(1-\lambda_t)g
  (R_{1:T})} = \frac{\alpha_t \One{P_t >
    \lambda_t}}{(1-\lambda_t)g (\widetilde{R}^{t\to 1}_{1:T})}.
\end{align*}

Since $\widetilde{R}^{t\to 1}_{1:T}$ is independent of $P_t$, we may
take conditional expectations to obtain:
\begin{align*}
\EEst{\frac{\alpha_t \One{P_t >
      \lambda_t}}{(1-\lambda_t)g (R_{1:T})} }{\F^{t-1}} &=
\EEst{\frac{\alpha_t \One{P_t
      >\lambda_t}}{(1-\lambda_t) g (\widetilde{R}^{t\to 1}_{1:T})} }{
  \F^{t-1}} \\
& \geq \EEst{ \frac{ \alpha_t}{g
    (\widetilde{R}^{t\to 1}_{1:T})} }{ \F^{t-1}}
\end{align*}
which follows by taking an expectation only with
respect to $P_t$ by invoking the conditional super-uniformity
property~\eqref{eqn:superuniformity-cond2}, as well as independence of $P_t$ and $\widetilde{R}^{t\to 1}_{1:T}$.

Finally, notice that $R_i \geq \widetilde R_i$ for all $i$. This follows by monotonicity of $\lambda_i$ and $\alpha_i$. In particular, $\alpha_{t+1}\geq \widetilde \alpha_{t+1}$, and $\lambda_{t+1}\geq \widetilde \lambda_{t+1}$ due to $(R_{1:t},C_{1:t})\geq (\widetilde R_{1:t},\widetilde C_{1:t})$, which in turn implies $R_{t+1}\geq \widetilde R_{t+1}, C_{t+1}\geq \widetilde C_{t+1}$, and so on. Recursively we deduce $R_i \geq \widetilde R_i$ for all $i$. Together with the assumption that $g$ is non-decreasing, this implies:
\begin{align*}
    \EEst{ \frac{ \alpha_t}{g
    (\widetilde{R}^{t\to 1}_{1:T})} }{ \F^{t-1}} \geq \EEst{ \frac{ \alpha_t}{g
    (R_{1:T})} }{ \F^{t-1}},
\end{align*}
which concludes the proof of the lemma.
\end{proof}


\subsection{Proof of \thmref{fdp-saff}}
\label{sec:proof-thm}

First note that, for any time $t \in \N$, we have:
\begin{align*}
\EE{|\nulls \cap \cR(t)|} \; = \; \sum_{j \leq t, j \in \nulls}
\EE{\One{P_j \leq \alpha_j}} & \; \stackrel{(i)}{\leq} \sum_{j \leq t, j
  \in \nulls} \EE{\alpha_j} \\
& \; \stackrel{(ii)}{\leq} \EE{\sum_{j \leq t, j \in \nulls} \alpha_j
  \frac{\One{P_j > \lambda_j}}{1-\lambda_j}},
\end{align*}
where inequality (i) first uses the law of iterated expectations by
conditioning on $\F^{j-1}$, and then both (i) and (ii) apply the
conditional super-uniformity
property~\eqref{eqn:superuniformity-cond2}, which concludes the proof
of part (a). To prove part (b), we drop the condition $j\in \nulls$ from the last expectation, and use the assumption that
\mbox{$\fdphat_\lambda(t) \defn \frac{\sum_{j \leq t} \alpha_j
    \frac{\One{P_j > \lambda_j}}{1-\lambda_j}}{|\cR(t)|} \leq \alpha$} to obtain:
\begin{align*}
  \EE{\sum_{j \leq t, j \in \nulls} \alpha_j \frac{\One{P_j >
        \lambda_j}}{1-\lambda_j}}\leq \alpha \EE{|\cR(t)|}.
\end{align*}
Combining this inequality with the result of part (a), and rearranging
the terms, we reach the conclusion that $\mfdr(t)\leq \alpha$, as
desired.

Under the independence and monotonicity assumptions of parts $(c,d)$, we have:
\begin{align*}
\fdr(t) &= \EE{\frac{|\nulls \cap \cR(t)|}{|\cR(t)|}} \\ &= \sum_{j
  \leq t, j \in \nulls} \EE{\dotfrac{\One{P_j \leq
      \alpha_j}}{|\cR(t)|}} \\ &\stackrel{(iii)}{\leq} \sum_{j \leq t,
  j \in \nulls} \EE{\dotfrac{\alpha_j}{|\cR(t)|}} \\
& \stackrel{(iv)}{\leq} \sum_{j \leq t, j \in \nulls} \EE{\dotfrac{
    \alpha_j\One{P_j >\lambda_j}}{(1-\lambda_j)|\cR(t)|}}, \\
\end{align*}
where inequality (iii) first uses iterated
expectations by conditioning on $\F^{j-1}$, and then both (iii) and (iv) apply \lemref{power}. Assuming that the inequality $\fdphat_\lambda(t) \leq \alpha$ holds, it follows that:
\begin{align*}
\sum_{j \leq t, j \in \nulls} \EE{\dotfrac{
    \alpha_j\One{P_j >\lambda_j}}{(1-\lambda_j)|\cR(t)|}}&\stackrel{(v)}{\leq} \EE{\dotfrac{\sum_{j \leq t} \alpha_j\One{P_j
      >\lambda_j}}{(1-\lambda_j)|\cR(t)|}} \\
& \stackrel{(vi)}{=} \EE{\fdphat_\lambda(t)} \\
& \stackrel{(vii)}{\leq} \alpha,
\end{align*}
where inequality (v) follows by linearity of expectation and
summing over a larger set of indices; equality (vi) simply uses the
definition of $\fdphat_\lambda(t)$, and inequality (vii) follows by
the assumption, hence proving parts (c,d).

\subsection{Proof of \lemref{monotone}}

We prove that the update rule (6-7) guarantees that $\alpha_t$ and $\lambda_t$ are monotone, given that $\lambda_t = h_t(\alpha_t)$, for some non-decreasing function $h_t$ such that $h_t(x)\geq x$. Recall that monotonicity means that $\alpha_t$ and $\lambda_t$ can be written as non-decreasing functions of the vector of rejections and candidates, i.e. $\alpha_t = f_t(R_{1:t-1},C_{1:t-1}), \lambda_t = g_t(R_{1:t-1},C_{1:t-1})$, for some coordinate-wise non-decreasing functions $f_t$ and $g_t$. 

Notice that, if we prove monotonicity of $\alpha_t$, monotonicity of $\lambda_t$ is immediate. This follows because $g_t = h_t \circ f_t$ must be non-decreasing as a composition of two non-decreasing functions.

Consider some $(R_{1:t-1},C_{1:t-1})\in\{0,1\}^{2(t-1)}$ and $(\widetilde{R}_{1:t-1},\widetilde{C}_{1:t-1})\in\{0,1\}^{2(t-1)}$, such that
$$(R_{1:t-1},C_{1:t-1})\geq (\widetilde{R}_{1:t-1},\widetilde{C}_{1:t-1})$$
coordinate-wise. Denote $q_t(x) = \frac{x}{1-h_t(x)}$, for $x\in(0,1)$. Notice that $q_t$ is strictly increasing, because, for any $x_1,x_2\in(0,1)$ such that $x_2>x_1$:
$$\frac{x_2}{1-g_t(x_2)} - \frac{x_1}{1-g_t(x_1)} = \frac{x_2(1-g_t(x_1)) - x_1(1-g_t(x_2))}{(1-g_t(x_1))(1-g_t(x_2))}>0,$$
which follows because $g_t(x)\in(0,1)$ and $g_t(x_1)\leq g_t(x_2)$, so $x_2(1-g_t(x_1)) - x_1(1-g_t(x_2))>0$. Therefore, we can conclude that $q_t$ is invertible.

Let $\alpha_t$ denote the test level computed from $(R_{1:t-1},C_{1:t-1})$, and let $\widetilde \alpha_t$ denote the test level computed from $(\widetilde R_{1:t-1},\widetilde C_{1:t-1})$. In a similar fashion we distinguish between other SAFFRON parameters resulting from the two sequences of rejection and candidate indicators. Recall the proposed SAFFRON update rule for the test level:
\begin{align*}
\alpha_t = q_t^{-1}\left(W_0\gamma_{t - C_{0+}} + (\alpha -
W_0) \gamma_{t-\tau_1 - C_{1+}} + \sum_{j \geq 2} \alpha
\gamma_{t-\tau_j - C_{j+}}\right),
\end{align*}
and denote $s_t(R_{1:t-1},C_{1:t-1}) = W_0\gamma_{t - C_{0+}} + (\alpha -
W_0) \gamma_{t-\tau_1 - C_{1+}} + \sum_{j \geq 2} \alpha
\gamma_{t-\tau_j - C_{j+}}$.

First, we claim $s_t(R_{1:t-1},C_{1:t-1})\geq s_t(\widetilde R_{1:t-1},\widetilde C_{1:t-1})$; we prove this via a term-by-term comparison. Because $\widetilde C_j = 1$ implies $C_j = 1$, we have $C_{0+} \geq \widetilde C_{0+}$, hence $W_0\gamma_{t-C_{0+}}\geq W_0\gamma_{t-\widetilde C_{0+}}$ due to $\{\gamma_t\}$ being non-increasing. Next, we claim that for every term in the sum $(\alpha -
W_0) \gamma_{t-\widetilde \tau_1 - \widetilde C_{1+}} + \sum_{j \geq 2} \alpha
\gamma_{t-\widetilde \tau_j - \widetilde C_{j+}}$, there is a unique corresponding term in $(\alpha -
W_0) \gamma_{t-\tau_1 - C_{1+}} + \sum_{j \geq 2} \alpha
\gamma_{t-\tau_j - C_{j+}}$ that is at least as big. For example, take $(\alpha-W_0)\gamma_{t-\widetilde \tau_1 - \widetilde C_{1+}}$. Then, $\widetilde R_{\widetilde \tau_1} = 1$ implies $R_{\widetilde \tau_1} = 1$, which means there is a corresponding term in $s_t(R_{1:t-1},C_{1:t-1})$ that is either $(\alpha-W_0)\gamma_{t-\widetilde \tau_1 - \sum_{j=\widetilde \tau_1+1}^{t-1} C_j}$ or $\alpha\gamma_{t-\widetilde \tau_1 - \sum_{j=\widetilde \tau_1+1}^{t-1} C_j}$, both of which dominate $(\alpha-W_0)\gamma_{t-\widetilde \tau_1 - \widetilde C_{1+}}$, because $\widetilde C_j = 1$ implies $C_j = 1$. By the same analysis, for every term $\alpha
\gamma_{t-\widetilde \tau_j - \widetilde C_{j+}}$, there is a corresponding dominating term in $s_t(R_{1:t-1},C_{1:t-1})$. Note that we are implicitly using the fact that $\tau_1 \leq \widetilde \tau_1$, hence for every contribution that has an $\alpha$ multiplier in $s_t(\widetilde R_{1:t-1},\widetilde C_{1:t-1})$, there is a corresponding contribution also with an $\alpha$ multiplier in $s_t(R_{1:t-1}, C_{1:t-1})$, as opposed to $\alpha-W_0$.

Second, because the inverse of an increasing function is also increasing, we have
$$q_t^{-1}(s_t(R_{1:t-1},C_{1:t-1}))\geq q_t^{-1}(s_t(\widetilde R_{1:t-1},\widetilde C_{1:t-1})),$$
or equivalently $\alpha_t \geq \widetilde \alpha_t$. As a conclusion, the test levels $\alpha_t$ are coordinate-wise non-decreasing in $(R_{1:t-1},C_{1:t-1})$, and as a consequence $\lambda_t$ is likewise coordinate-wise non-decreasing in $(R_{1:t-1},C_{1:t-1})$.


\section{Conclusion}
\label{sec:disc}

This paper introduces SAFFRON, a new algorithmic framework for online mFDR and FDR control. We show empirically that SAFFRON is more powerful than existing algorithms. The derivation and proof of SAFFRON is based on a novel reverse super-uniformity lemma that allows us  to estimate the fraction of alpha-wealth that an algorithm spends on testing null hypotheses. One may interpret SAFFRON as an adaptive version of LORD, just as Storey-BH is an adaptive version of the Benjamini-Hochberg algorithm. Also, a monotone version of the alpha-investing algorithm that is often more stable and powerful than the original, is recovered as a special case of the SAFFRON framework (and hence it controls FDR, not just mFDR like the original). Lastly, the derivation of SAFFRON is rather different from that of earlier generalized alpha-investing (GAI) algorithms, and as such provides a template for the derivation of new algorithms.

\bibliography{FDR} \bibliographystyle{plainnat}

\end{document}